\documentclass{article}%
\usepackage{amsmath}
\usepackage{amsfonts}
\usepackage{amssymb}
\usepackage{graphicx}%
\providecommand{\U}[1]{\protect\rule{.1in}{.1in}}
\newtheorem{theorem}{Theorem}

\newtheorem{conclusion}[theorem]{Conclusion}

\newtheorem{definition}[theorem]{Definition}

\newtheorem{proposition}[theorem]{Proposition}
\newtheorem{remark}[theorem]{Remark}

\newenvironment{proof}[1][Proof]{\noindent\textbf{#1.} }{\ \rule{0.5em}{0.5em}}
\begin{document}

\title{Notes on Conservation Laws, Equations of Motion of Matter and Particle Fields
in Lorentzian and Teleparallel de Sitter Spacetime Structures}
\author{Waldyr A. Rodrigues Jr. and Samuel A. Wainer\\IMECC-UNICAMP\\{\footnotesize walrod@ime.unicamp.br~~~samuelwainer@ime.unicamp.br}}
\maketitle
\tableofcontents

\begin{abstract}
In this paper we discuss the physics of interacting tensor fields and
particles living in a de Sitter manifold $M=\mathrm{S0}(1,4)/\mathrm{S0}%
(1,3)\simeq\mathbb{R}\times S^{3}$ interpreted as a submanifold of
$(\mathring{M}=\mathbb{R}^{5},\boldsymbol{\mathring{g}})$, with
$\boldsymbol{\mathring{g}}$ a metric of signature $(1,4)$. The pair
$(M,\boldsymbol{g})$ \ where $\boldsymbol{g}$ is the pullback metric of
$\boldsymbol{\mathring{g}}$ $(\boldsymbol{g=i}^{\ast}\boldsymbol{\mathring{g}%
})$ is a Lorentzian manifold that is oriented by $\tau_{\boldsymbol{g}}$ and
time oriented by $\uparrow$. It is the structure $(M,\boldsymbol{g,}%
\tau_{\boldsymbol{g}},\uparrow)$ that is primely used to study\ the
energy-momentum conservation law for a system of physical fields (and
particles) living in $M$ and to get the equations of motion of the fields and
also the equations of motion describing the behavior of free particles. To
achieve our objectives we construct two different de Sitter spacetime
structures $M^{dSL}=(M,\boldsymbol{g,D},\tau_{\boldsymbol{g}},\uparrow)$ and
$M^{dSTP}=(M,\boldsymbol{g,\nabla},\tau_{\boldsymbol{g}},\uparrow)$, where
$\boldsymbol{D}$ is the Levi-Civita connection of $\boldsymbol{g}$ and
$\boldsymbol{\nabla}$ is a metric compatible parallel connection. Both
connections are introduced in our study only as mathematical devices, no
special physical meaning is attributed to these objects. In particular
$M^{dSL}$ is not supposed to be the model of any gravitational field in the
General Relativity Theory (\textbf{GRT}). Our approach permit to clarify some
misconceptions appearing in the literature, in particular one claiming that
free particles in the de Sitter structure $(M,\boldsymbol{g})$ do not follows
timelike geodesics. The paper makes use of the Clifford and spin-Clifford
bundles formalism recalled in one of the appendices, something needed for a
thoughtful presentation of the concept of a Komar current $\boldsymbol{J}_{A}$
(in \textbf{GRT}) associated to any vector field $\mathbf{A}$ generating a one
parameter group of diffeomorphisms. The explicit formula for $\boldsymbol{J}%
_{\emph{A}}$ in terms of the energy-momentum tensor of the fields and its
physical meaning is given. Besides that we show how $F=dA$ ($A=\boldsymbol{g}%
(\mathbf{A},~)$ satisfy a Maxwell like equation $\boldsymbol{\partial
}F=\boldsymbol{J}_{\emph{A}}$ which encodes the contents of Einstein equation.
Our results shows that in \textbf{GRT} there are infinitely many conserved
currents, independently of the fact that the Lorentzian spacetime
(representing a gravitational field) possess or not Killing vector fields.
Moreover our results also show that even when the appropriate timelike and
spacelike Killing vector fields exist it is not possible to define a conserved
energy-momentum \emph{covector }(not covector field) as in Special
Relativistic Theories.

\end{abstract}

\section{Introduction}

In this paper we study some aspects of Physics of fields living and
interacting in a manifold $M=\mathrm{SO}(1,4)\mathrm{/SO}(1,3)\simeq
\mathbb{R}^{3}\times S^{3}$. We introduce two different geometrical spacetime
structures that we can form starting from the manifold $M$ which is supposed
to be a vector manifold, i.e., a submanifold of ($\mathring{M}%
,\boldsymbol{\mathring{g}})$ with $\mathring{M}=\mathbb{R}^{5}$ and
$\boldsymbol{\mathring{g}}$ a metric of signature $(1,4)$. If $\boldsymbol{i}%
:M\rightarrow\mathring{M}$ is the inclusion map the structures that will be
studied are the Lorentzian de Sitter spacetime $M^{dSL}=(M,\boldsymbol{g}%
,\boldsymbol{D},\tau_{\boldsymbol{g}},\uparrow)$ and teleparallel de Sitter
spacetime $M^{dSTP}=(M,\boldsymbol{g},\boldsymbol{\nabla},\tau_{\boldsymbol{g}%
},\uparrow)$ where $\boldsymbol{g=i}^{\ast}\boldsymbol{\mathring{g}}$,
$\boldsymbol{D}$ is the Levi-Civita connection of $\boldsymbol{g}$ and
$\boldsymbol{\nabla}$ is a metric compatible teleparallel connection (see
Section 4.1). Our main objective is the following: taking $(M,\boldsymbol{g})$
as the arena where physical fields live and interact\ how do we formulate
conservation laws of energy-momentum and angular momentum for the system of
physical fields. In order to give a meaningful meaning to this question we
recall the fact that in Lorentzian spacetime structures that are models of
gravitational fields in the \textbf{GRT} there are no genuine conservation
laws of energy-momentum (and also angular momentum) for a closed system of
fields and moreover there are no genuine energy-momentum and angular momentum
conservation laws for the system consisting of non gravitational plus the
gravitational field. We discuss in Section 2.1 a pure mathematical result,
namely when there exists some conserved currents $\mathcal{J}_{V}$ $\in\sec
T^{\ast}M$ in a Lorentzian spacetime associated to a tensor field
$\boldsymbol{W\in}\sec T_{1}^{1}M$ and a vector field $\boldsymbol{V}\in\sec
TM$. In Section 2.2\ we briefly recall how a conserved energy-momentum tensor
for the matter fields is constructed in Special Relativity theories and how in
that theory it is possible to construct a conserved energy-momentum
\emph{covector}\footnote{The energy-momentum \emph{covector }is an element of
a vector space and is not a covector field.}\emph{ f}or the matter fields.
After that we recall that in \textbf{GRT} we have a covariantly
\textquotedblleft conserved\textquotedblright\ energy-momentum tensor
$\boldsymbol{T\in}\sec T_{1}^{1}M$ (i.e., $\boldsymbol{D\bullet T}=0$) and so,
using the results of Section 2.1 we can immediately construct conserved
currents when the Lorentzian spacetime modelling the gravitational field
generated by $\boldsymbol{T}$ possess Killing vector fields. However, we show
that is not possible in general in \textbf{GRT} even when some special
conserved currents exist (associated to one timelike and three spacelike
Killing vector fields) to build a conserved \emph{covector} for the system of
fields, as it is the case in special relativistic theories. Immediately after
showing that we ask the question:\smallskip

Is it necessary to have Killing vector fields in a Lorentzian spacetime
modelling a given gravitational field in order to be possible to construct
conserved currents?\smallskip

Well,\ we show that the answer is \emph{no}. In \textbf{GRT} there are an
infinite number of conserved currents. This is showed in Section
2.4\footnote{This section is an improvement of results first presented in
\cite{rrr2012}.} were we introduce the so called Komar currents in a
Lorentzian spacetime modelling a gravitational field generated by a given
(symmetric) energy momentum tensor $\boldsymbol{T}$\ and show how any
diffeomorphism associated to a one parameter group generated by a vector field
$\mathbf{A}$ lead to a conserved current. We show moreover using the Clifford
bundle formalism recalled in Appendix A that $F=dA$ $\boldsymbol{\in}\sec%
%TCIMACRO{\tbigwedge \nolimits^{2}}%
%BeginExpansion
{\textstyle\bigwedge\nolimits^{2}}
%EndExpansion
T^{\ast}M$ where $A=\boldsymbol{g}(\mathbf{A},~)\in\sec%
%TCIMACRO{\tbigwedge \nolimits^{1}}%
%BeginExpansion
{\textstyle\bigwedge\nolimits^{1}}
%EndExpansion
T^{\ast}M$ satisfy, (with $\boldsymbol{\partial}$ denoting the Dirac operator
acting on sections of the Clifford bundle of differential forms) \ a Maxwell
like equation $\boldsymbol{\partial}F=\boldsymbol{J}_{A}$ (equivalent to
$dF=0$ and $\underset{\boldsymbol{g}}{\delta}F=-\boldsymbol{J}_{A}$). The
explicit form of $\mathcal{J}_{A}$ as a function of the energy-momentum tensor
is derived \ (see Eq.(\ref{kocurr})) together with its scalar invariant. We
establish that\footnote{The symbol $\boldsymbol{\partial}$ denotes the the
Dirac operator acting on sections of the Clifford bundle $\mathcal{C\ell
}(M,\mathtt{g})$. See Appendix A.} $\boldsymbol{\partial}F=\boldsymbol{J}_{A}$
encode the contents of Einstein equation. We show moreover that even if we can
get four conserved currents given one time like and three spacelike vector
fields and thus get four scalar invariants these objects \emph{cannot} be
associated to the components of a momentum \emph{covector}\footnote{Not a
covector field.} for the system of fields producing the energy-momentum tensor
$\boldsymbol{T}$. We also give the form of $\boldsymbol{J}_{A}$ when
$\mathbf{A}$ is a Killing vector field and emphasize that even if the
Lorentzian spacetime under consideration has one time like and three spacelike
Killing vector fields we cannot find a conserved momentum covector for the
system of fields.\smallskip

This paper has several appendices necessary for a perfect intelligibility of
the results in the main text. Thus it is opportune to describe what is there
and where their contents are used in the main text\footnote{Some of the
material of the Appendices is well known, but we think that despite this fact
\ theri presentation here will be useful for most of our readers.}. To start,
in Appendix A we briefly recall the main results of the Clifford bundle
formalism used in this paper which permits one to understand how to arrive at
the equation $\boldsymbol{\partial}F\boldsymbol{=J}_{A}$ in Section
2.2.\footnote{The Clifford bundle formalism permits the representation of a
covariant Dirac spinor field as certain equivalence classes of even sections
of the Clifford bundle, called Dirac-Hestenes spinor field (\textbf{DHSF}).
These objects are a key ingredient to clarify the concept of Lie derivative of
spinor fields of and give meaningful definition for such an object , something
necessary to study conservation laws in Lorentzian spacetime structures when
spinor fields are present. Our approach to the subject is descrbed in
\cite{lrw2015} and athoughtful derivation of Dirac equation in de Sitter
structure $(M,\boldsymbol{g})$\ using \textbf{DHSFs} is given in
\cite{RWC2015}.}. Lie derivatives and variations of tensor fields is discussed
in Appendix B. In Section C1 we derive from the Lagrangian formalism conserved
currents for fields living in a general Lorentzian spacetime structure and the
corresponding generalized covariant energy-momentum \textquotedblleft
conservation\textquotedblright\ law. We compare these results in Section C2
with the analogues ones for field theories in special relativistic theories
where the Lorentzian spacetime structure is Minkowski spacetime. We show that
despite the fact that we can derive conserved quantities for fields living and
interacting in $M^{\ell DS}$ we cannot define in this structure a genuine
energy-momentum conserved covector for the system of fields as it is the case
in Minkowski spacetime. A legitimate energy-momentum covector for the system
of fields living in $(M=\mathrm{SO}(1,4)\mathrm{/SO}(1,3),\boldsymbol{g})$
exist only in the teleparallel structure $M^{dSTP}$. This is discussed in
Section 5.2 after recalling the Lie algebra and the Casimir invariants of the
Lie algebra of de Sitter group in Section 5.1. In Appendix E we derive for
completeness and to insert the Remark \ref{dtt=0} the so called covariant
energy-momentum conservation law in \textbf{GRT}. Appendix D recalls the
intrinsic definition of relative tensors and their covariant derivatives.
Appendix E present proofs of some identities used in the main text.

As we already said the main objective of this paper is to discuss the Physics
of interacting fields in de Sitter spacetime structures $M^{dSL}$ and
$M^{dSTP}$. In particular we want also to clarify some misunderstandings
concerning the roles of geodesics in the $M^{dS\ell}$. So, in section 3 we
briefly recall the conformal representation of the de Sitter
spacetime\ structure $M^{dSL}$ and prove that the one timelike and the three
spacelike \textquotedblleft translation\textquotedblright\ Killing vector
fields of $(M,\boldsymbol{g})$ defines a basis for almost all $M$. With this
result we show in Section 4 that the method using in \cite{ps2012} to obtain
the curves which extremizes the length function of timelike curves in de
Sitter spacetime with the result that these curves are not geodesics is
equivocated, since those authors use \emph{constrained variations} instead of
\emph{arbitrary variations} of the length function. Even more, the equation
obtained from the constrained variation in \cite{ps2012} is according to our
view wrongly interpreted in its mathematical (and physical) contents. Indeed,
using some of the results of Section 5.2 and the results of Section 6 which
briefly recall Papapetrou's classical results \cite{papapetrou} deriving the
equation of motion of a probe single-pole particle in \textbf{GRT},\textbf{
}we show in Section 7 that contrary to the authors statement in \cite{ps2012}
it is not true that the equation of motion of a single-pole obtained from a
method similar to Papapetrou's one\ in \cite{papapetrou} but using the
generalized energy-momentum tensor of matter fields in $M^{dSL}$ gives an
equation of motion different from the geodesic equation in $M^{dSL}$ and in
agreement with the one they derived from his constrained variation method.
Indeed, we prove that from the equation describing the motion of a
single-pole\ the geodesic equation follows automatically. Finally, in Section
8 we present our conclusions.

\section{Preliminaries}

Let $(M,\boldsymbol{g},D,\tau_{\boldsymbol{g}},\uparrow)$ be a general
Lorentzian spacetime. Let $\mathcal{U}\subseteq M$ be an open set covered by
coordinates $\{x^{\mu}\}$. Let $\{e_{\mu}=\partial_{\mu}\}$ be a basis of
$T\mathcal{U}$ and $\{\boldsymbol{\vartheta}^{\mu}=dx^{\mu}\}$ \ the basis of
$T^{\ast}\mathcal{U}$ dual to the basis $\{\partial_{\mu}\}$, i.e.,
$\boldsymbol{\vartheta}^{\mu}(\partial_{\nu})=\delta_{\nu}^{\mu}$. We denote
by $\mathtt{g}$ a metric of the cotangent bundle such that if $\boldsymbol{g}%
=g_{\mu\nu}\boldsymbol{\vartheta}^{\mu}\otimes\boldsymbol{\vartheta}^{\nu}$
then $\mathtt{g}=g^{\mu\nu}\partial_{\mu}\otimes\partial_{\nu}$ with
$g^{\mu\rho}g_{\rho\nu}=\delta_{\nu}^{\mu}$. We introduce also $\{\partial
^{\mu}\}$ and $\{\boldsymbol{\vartheta}_{\mu}\}$ respectively as the
reciprocal bases of $\{e_{\mu}\}$ and $\{\boldsymbol{\vartheta}^{\mu}\}$,
i.e., we have
\begin{equation}
\boldsymbol{g}(\partial_{\nu},\partial^{\mu})=\delta_{\nu}^{\mu}%
,~~~\mathtt{g}(\boldsymbol{\vartheta}^{\mu},\boldsymbol{\vartheta}_{\nu
})=\delta_{\nu}^{\mu} \label{p1}%
\end{equation}

Next we introduce in $T\mathcal{U}$ the tetrad basis $\{\boldsymbol{e}%
_{\alpha}=h_{\alpha}^{\mu}\partial_{\mu}\}$ and in $T^{\ast}\mathcal{U}$ the
cotetrad basis \ $\{\boldsymbol{\gamma}^{\alpha}=h_{\mu}^{\alpha}\gamma^{\mu
}\}$ which are dual basis. We introduce moreover the basis $\{\boldsymbol{e}%
^{\alpha}\}$ and $\{\boldsymbol{\gamma}_{\alpha}\}$ as the reciprocal bases of
$\{\boldsymbol{e}_{\alpha}\}$ and $\{\boldsymbol{\gamma}^{\alpha}\}$
satisfying
\begin{equation}
\boldsymbol{g}(\boldsymbol{e}_{\alpha},\boldsymbol{e}^{\beta})=\delta_{\alpha
}^{\beta},~~~\mathtt{g}(\boldsymbol{\gamma}^{\beta},\boldsymbol{\gamma
}_{\alpha})=\delta_{\alpha}^{\beta}. \label{p2}%
\end{equation}

Moreover recall that it is
\begin{align}
\boldsymbol{g}  &  =\eta_{\alpha\beta}\boldsymbol{\gamma}^{\alpha}%
\otimes\boldsymbol{\gamma}^{\beta}=\eta^{\alpha\beta}\boldsymbol{\gamma
}_{\alpha}\otimes\boldsymbol{\gamma}_{\beta},\nonumber\\
\mathtt{g}  &  =\eta^{\alpha\beta}\boldsymbol{e}_{\alpha}\otimes
\boldsymbol{e}_{\beta}=\eta_{\alpha\beta}\boldsymbol{e}^{\alpha}%
\otimes\boldsymbol{e}^{\beta}. \label{p3}%
\end{align}

\subsection{The Currents $\mathcal{J}_{V}$ and $\mathcal{J}_{K}$}

Let $\mathbf{W}=W^{\alpha\beta}\boldsymbol{e}_{\alpha}\otimes\boldsymbol{e}%
_{\beta}\in\sec T_{2}^{0}M$ with $W^{\alpha\beta}=W^{\beta\alpha}$ and
$\mathbf{\check{W}}=W_{\alpha\beta}\boldsymbol{\gamma}^{\alpha}\otimes
\boldsymbol{\gamma}^{\beta}\in\sec T_{0}^{2}M$, $W_{\alpha\beta}=\eta
_{\alpha\varsigma}\eta_{\beta\tau}W^{\varsigma\tau}$ and $\boldsymbol{W}%
=W_{\beta}^{\alpha}\boldsymbol{\gamma}^{\beta}\otimes\boldsymbol{e}_{\alpha
}\in\sec T_{1}^{1}M$. For the applications we have in mind we will say that
$\mathbf{W,\check{W}}$ and $\boldsymbol{W}$ are physically equivalent.

Note that $\boldsymbol{W}$ (an example of an extensor field\footnote{See
Chapter 4 of \cite{rc2007}.}) is such that
\begin{gather}
\boldsymbol{W}:\sec%
%TCIMACRO{\dbigwedge \nolimits^{1}}%
%BeginExpansion
{\displaystyle\bigwedge\nolimits^{1}}
%EndExpansion
T^{\ast}M\rightarrow\sec%
%TCIMACRO{\dbigwedge \nolimits^{1}}%
%BeginExpansion
{\displaystyle\bigwedge\nolimits^{1}}
%EndExpansion
T^{\ast}M,\nonumber\\
\boldsymbol{W}(V)=V_{\alpha}W_{\beta}^{\alpha}\boldsymbol{\gamma}^{\beta}.
\label{P33}%
\end{gather}
Define the divergence of $\boldsymbol{W}$ as the $1$-form field%
\begin{equation}
\boldsymbol{D}\bullet\boldsymbol{W}\mathbf{:}\mathcal{=(}D_{\alpha}W_{\beta
}^{\alpha})\boldsymbol{\gamma}^{\beta} \label{p4}%
\end{equation}
where%
\begin{equation}
\boldsymbol{D}_{\alpha}W_{\beta}^{\alpha}:=(\boldsymbol{D}_{\boldsymbol{e}%
_{\alpha}}\boldsymbol{W})_{\beta}^{\alpha}=\boldsymbol{e}_{\alpha}(W_{\beta
}^{\alpha})+\Gamma_{\alpha\iota}^{\alpha}W_{\beta}^{\iota}-\Gamma_{\alpha
\beta}^{\iota}W_{\iota}^{\alpha}. \label{P4A}%
\end{equation}

Moreover, introduce the $1$-form fields
\begin{equation}
\mathcal{W}^{\beta}:=W^{\alpha\beta}\boldsymbol{\gamma}_{\alpha}\in\sec%
%TCIMACRO{\dbigwedge \nolimits^{1}}%
%BeginExpansion
{\displaystyle\bigwedge\nolimits^{1}}
%EndExpansion
T^{\ast}M. \label{p5}%
\end{equation}

\begin{remark}
Take notice for the developments that follows that the Hodge coderivative of
the $1$-form fields $\mathcal{W}^{\beta}$is \emph{(}see Appendix\emph{)}:%
\begin{align*}
\underset{\boldsymbol{g}}{\delta}\mathcal{W}^{\beta}  &  =-\boldsymbol{\gamma
}^{\kappa}\lrcorner\boldsymbol{D}_{\boldsymbol{e}_{\kappa}}(W^{\alpha\beta
}\boldsymbol{\gamma}_{\alpha})\\
&  =-\boldsymbol{e}_{\kappa}(W^{\alpha\beta})\boldsymbol{\gamma}^{\kappa
}\lrcorner\boldsymbol{\gamma}_{\alpha}-W^{\alpha\beta}\Gamma_{\kappa\alpha
}^{\iota}\boldsymbol{\gamma}^{\kappa}\lrcorner\boldsymbol{\gamma}_{\iota}\\
&  =-\boldsymbol{e}_{\alpha}(W^{\alpha\beta})-W^{\alpha\beta}\Gamma
_{\kappa\alpha}^{\kappa}.
\end{align*}

So, $\boldsymbol{D}\bullet\boldsymbol{W}=0$ does not implies that
$\underset{\boldsymbol{g}}{\delta}\mathcal{W}^{\beta}=0$.
\end{remark}

Now, given a vector field $\mathbf{V}=V^{\alpha}\boldsymbol{e}_{\alpha}$ and
the physically equivalent covector field $V=V^{\alpha}\boldsymbol{\gamma
}_{\alpha}$ define the current%
\begin{equation}
\mathcal{J}_{V}=V^{\alpha}\mathcal{W}_{\alpha}\in\sec%
%TCIMACRO{\dbigwedge \nolimits^{1}}%
%BeginExpansion
{\displaystyle\bigwedge\nolimits^{1}}
%EndExpansion
T^{\ast}M\hookrightarrow\sec\mathcal{C\ell}(M,\mathtt{g}) \label{p6}%
\end{equation}
Of course, writing
\begin{equation}
\mathcal{J}_{V}=J_{\beta}\boldsymbol{\gamma}^{\beta}%
\end{equation}
we have$\cdot$
\begin{equation}
\mathcal{J}_{\beta}=V^{\alpha}W_{\alpha\beta}\text{.} \label{p7}%
\end{equation}

Recalling (see Appendix A) that $\star1=\tau_{\boldsymbol{g}}$ define
\begin{equation}
\mathfrak{T}=W^{\alpha\beta}\underset{\boldsymbol{g}}{\star}1=W^{\alpha\beta
}\tau_{\boldsymbol{g}}\in\sec%
%TCIMACRO{\dbigwedge \nolimits^{4}}%
%BeginExpansion
{\displaystyle\bigwedge\nolimits^{4}}
%EndExpansion
T^{\ast}M\hookrightarrow\mathcal{C}\ell(M,\mathtt{g}\text{)}. \label{p8}%
\end{equation}

Then, we have, with $\boldsymbol{\partial}$ denoting the Dirac operator,
\begin{gather}
d\underset{\boldsymbol{g}}{\star}\mathcal{J}_{V}=\boldsymbol{\partial
\wedge\underset{\boldsymbol{g}}{\star}}\mathcal{J}\boldsymbol{_{V}=\gamma
}^{\alpha}\wedge(\boldsymbol{D}_{\boldsymbol{e}_{\alpha}}%
\underset{\boldsymbol{g}}{\star}\mathcal{J}_{V})\nonumber\\
=\boldsymbol{D}_{\boldsymbol{e}_{\alpha}}(\boldsymbol{\gamma}^{\alpha}%
\wedge\underset{\boldsymbol{g}}{\star}\mathcal{J}_{V})-\boldsymbol{D}%
_{\boldsymbol{e}_{\alpha}}\boldsymbol{\gamma}^{\alpha}\wedge
\underset{\boldsymbol{g}}{\star}\mathcal{J}_{V}. \label{p9}%
\end{gather}

Taking into account that $\boldsymbol{\gamma}^{\alpha}\wedge
\underset{\boldsymbol{g}}{\star}\mathcal{J}_{V}=\underset{\boldsymbol{g}%
}{\star}(\boldsymbol{\gamma}^{\alpha}\lrcorner\mathcal{J}_{V})$ and
$\boldsymbol{D}_{\boldsymbol{e}_{\alpha}}\boldsymbol{\gamma}^{\alpha}%
\wedge\underset{\boldsymbol{g}}{\star}\mathcal{J}_{V}=\underset{\boldsymbol{g}%
}{\star}(\boldsymbol{D}_{\boldsymbol{e}_{\alpha}}\boldsymbol{\gamma}^{\alpha
}\lrcorner\mathcal{J}_{V})$ we can write Eq.(\ref{p9}) as
\begin{equation}
d\underset{\boldsymbol{g}}{\star}\mathcal{J}_{V}=\left(  \boldsymbol{e}%
_{\alpha}(V^{\kappa})W_{\kappa}^{\alpha}+V^{\kappa}\boldsymbol{e}_{\alpha
}(W_{\kappa}^{\alpha})+\Gamma_{\cdot\alpha\beta}^{\alpha\cdot\cdot}V^{\kappa
}W_{\kappa}^{\beta}\right)  \tau_{\boldsymbol{g}} \label{p9a}%
\end{equation}
Also, we can easily verifiy from Eq.(\ref{p4}) that
\begin{gather}
\boldsymbol{\underset{\boldsymbol{g}}{\star}[}(\boldsymbol{D}\bullet
\boldsymbol{W})(\mathbf{V})]=[(\boldsymbol{D}\bullet\boldsymbol{W}%
)(\mathbf{V})]\tau_{\boldsymbol{g}}\nonumber\\
=\left(  V^{\kappa}\boldsymbol{e}_{\alpha}(W_{\kappa}^{\alpha})+\Gamma
_{\cdot\alpha\iota}^{\alpha\cdot\cdot}W_{\kappa}^{\iota}V^{\kappa}%
-\Gamma_{\cdot\alpha\kappa}^{\iota\cdot\cdot}W_{\iota}^{\alpha}V^{\kappa
}\right)  \tau_{\boldsymbol{g}}. \label{p9aa}%
\end{gather}

Now, let $\pounds $ be the (standard) Lie derivative operator. Let us evaluate
the produuct of $\pounds _{\mathbf{V}}\boldsymbol{g~}\mathcal{(}%
\boldsymbol{e}_{\alpha},\boldsymbol{e}_{\beta}\mathcal{)}$ by $\mathfrak{T}$,
i.e.,
\begin{equation}
(\pounds _{\mathbf{V}}\boldsymbol{g~}\mathcal{(}\boldsymbol{e}_{\alpha
},\boldsymbol{e}_{\beta}\mathcal{))~}\mathfrak{T}\in\sec%
%TCIMACRO{\dbigwedge \nolimits^{4}}%
%BeginExpansion
{\displaystyle\bigwedge\nolimits^{4}}
%EndExpansion
T^{\ast}M. \label{p9A}%
\end{equation}
From Cartan magical formula we get \
\begin{gather}
\pounds _{\mathbf{V}}\boldsymbol{\gamma}^{\alpha}=d(V^{\alpha})+V\lrcorner
(\boldsymbol{\partial}\wedge\boldsymbol{\gamma}^{\alpha})\nonumber\\
=\boldsymbol{e}_{\iota}(V^{\alpha})\boldsymbol{\gamma}^{\iota}-V^{\varsigma
}\Gamma_{\cdot\varsigma\iota}^{\alpha\cdot\cdot}\boldsymbol{\gamma}^{\iota
}+V^{\varsigma}\Gamma_{\cdot\iota\varsigma}^{\alpha\cdot\cdot}%
\boldsymbol{\gamma}^{\iota}. \label{p9aaa}%
\end{gather}
Then,%
\begin{align}
\lbrack\pounds _{\mathbf{V}}\boldsymbol{g~}\mathcal{(}\boldsymbol{e}_{\alpha
},\boldsymbol{e}_{\beta}\mathcal{)]~}\mathfrak{T}  &  =[\left(  \eta
_{\iota\kappa}\pounds _{\mathbf{V}}\boldsymbol{\gamma}^{\iota}\otimes
\boldsymbol{\gamma}^{\kappa}+\eta_{\iota\kappa}\boldsymbol{\gamma}^{\iota
}\otimes\pounds _{\mathbf{V}}\boldsymbol{\gamma}^{\kappa}\right)
~\mathcal{(}\boldsymbol{e}_{\alpha},\boldsymbol{e}_{\beta}\mathcal{)]}%
\mathfrak{T}\nonumber\\
&  =[2\boldsymbol{e}_{\kappa}(V_{\iota})W^{\kappa\iota}+2V^{\varsigma}%
\Gamma_{\cdot\iota\varsigma}^{\kappa\cdot\cdot}W_{\kappa}^{\iota}]\mathfrak{T}
\label{p10a}%
\end{align}
and we get from Eqs.(\ref{p9a}), (\ref{p9aa}) and (\ref{p10a}) the important
identity \cite{bt1987}
\begin{equation}
(\pounds _{\mathbf{V}}\boldsymbol{g~}\mathcal{(}\boldsymbol{e}_{\alpha
},\boldsymbol{e}_{\beta}\mathcal{))~}\mathfrak{T}=2d\star\mathcal{J}%
_{\mathbf{V}}-2\underset{\boldsymbol{g}}{\star}[(\boldsymbol{D}\bullet
\boldsymbol{W})(\mathbf{V})]. \label{p11}%
\end{equation}
From Eq.(\ref{p11}) we see that if $\mathbf{V}$ is a conformal Killing vector
field, i.e., $\pounds _{\mathbf{V}}\boldsymbol{g~}\mathcal{(}\boldsymbol{e}%
_{\alpha},\boldsymbol{e}_{\beta}\mathcal{)}=2\lambda\eta_{\alpha\beta}$ we
have
\begin{equation}
\underset{\boldsymbol{g}}{\star}\lambda\mathrm{tr}\boldsymbol{W}%
=d\underset{\boldsymbol{g}}{\star}\mathcal{J}_{V}-\star\lbrack(\boldsymbol{D}%
\bullet\boldsymbol{W})(\mathbf{V})] \label{p12}%
\end{equation}
where $\mathrm{tr}\boldsymbol{W}$ is the trace of the matrix with entries
$W_{\beta}^{\alpha}$.

\subsection{Conserved Currents Associated to a Covariantly Conserved
$\boldsymbol{W}$}

\begin{definition}
We say that $\boldsymbol{W}$ is \textquotedblleft\emph{covariantly
conserved}\textquotedblright\ if
\begin{equation}
\boldsymbol{D}\bullet\boldsymbol{W}=0 \label{p12a}%
\end{equation}
\ 
\end{definition}

In this case, if $\mathbf{V=K}$ is a Killing vector field then
$\pounds _{\mathbf{K}}\boldsymbol{g}=0$ and we have%

\begin{equation}
d\underset{\boldsymbol{g}}{\star}\mathcal{J}_{K}=\underset{\boldsymbol{g}%
}{\star}[(\boldsymbol{D}\bullet\boldsymbol{W})(\mathbf{K})] \label{p13}%
\end{equation}
and the current $3$-form field\ $\underset{\boldsymbol{g}}{\star}%
\mathcal{J}_{K}$ is closed, i.e., $d\underset{\boldsymbol{g}}{\star
}\mathcal{J}_{K}=0$, or equivalently (taking into account the definition of
the Hodge coderivative operator $\underset{\boldsymbol{g}}{\delta}$)%
\begin{equation}
\underset{\boldsymbol{g}}{\delta}\mathcal{J}_{K}=0. \label{p14}%
\end{equation}

In resume, when we have Killing vector fields\footnote{The maximum number is
10 when $\dim M=4$ and that maximum number occurs only for spacetimes of
constant curvature.} $\mathbf{K}_{i}$, $i=1,2,...,n$\ \textquotedblleft
covariant conservation\textquotedblright\ of the tensor field $\boldsymbol{W}%
$, i.e., $\boldsymbol{D}\bullet\boldsymbol{W}=0$ implies in \emph{genuine}
conservation laws for the currents $\mathcal{J}_{K_{i}}$, i.e., from
$\underset{\boldsymbol{g}}{\delta}\mathcal{J}_{K_{I}}=0$ we can using Stokes
theorem build the \emph{scalar} conserved quantities
\begin{equation}
\mathcal{E(}K_{i}\mathcal{)}:=\frac{1}{8\pi}%
%TCIMACRO{\tint \nolimits_{\Sigma^{\prime}}}%
%BeginExpansion
{\textstyle\int\nolimits_{\Sigma^{\prime}}}
%EndExpansion
\underset{\boldsymbol{g}}{\star}\mathcal{J}_{K_{i}} \label{p14a}%
\end{equation}
where $N$ is the region where $\mathcal{J}_{K}$ has support and $\partial
N=\Sigma+\Sigma^{\prime}+\mathbf{\Xi}$ where $\Sigma,\Sigma^{\prime}$ are
spacelike surfaces and $\mathcal{J}_{K_{i}}$ is null at $\mathbf{\Xi}$
(spatial infinity).

\subsection{Conserved Currents in \textbf{GRT} Associated to Killing Vector
Fields}

Before studying the conditions for the existence or not of genuine
energy-momentum conservation laws in \textbf{GRT}, let us recall from Appendix
C.3.4 that in Minkowski spacetime\footnote{See \cite{rc2007} for the remaning
of the notation.} $(M\simeq\mathbb{R}^{4},\boldsymbol{\eta},\mathrm{D}%
,\tau_{\boldsymbol{\eta}},\uparrow)$ we can introduce global coordinates
$\{\mathrm{x}^{\mu}\}$ (in Einstein-Lorentz-Poincar\'{e} gauge) such that
$\boldsymbol{\eta(}e_{\mu},e_{\nu}\boldsymbol{)}=\eta_{\mu\nu}$, and
$\mathrm{D}_{e_{\mu}}e_{\nu}=0$, where the $\{e_{\mu}=\partial/\partial
\mathrm{x}^{\mu}\}$ is simultaneously a global tetrad and a coordinate basis.
Also the $\{\vartheta^{\mu}=d\mathrm{x}^{\mu}\}$ is a global cotetrad and a
coordinate cobasis.

Moreover, the $e_{\mu}=\partial/\partial\mathrm{x}^{\mu}$ are also Killing
vector fields in $(M\simeq\mathbb{R}^{4},\boldsymbol{\eta)}$ \ and thus we
have for a closed physical system (consisting of particles and fields in
interaction living in Minkowski spacetime and whose equations of motion are
derived from a variational principle with a Lagrangian density invariant under
spacetime translations) that the currents\footnote{Keep in mind that in
Eq.(\ref{p15}) the\emph{ }$T_{\alpha\nu}$ are the $\nu$-component of the
current $\mathcal{T}_{\alpha}=\mathcal{J}_{\partial/\partial\mathrm{x}%
^{\alpha}}$ and moreover are here taken as symmetric. See Appendix C.3.3.}
\begin{equation}
\mathcal{T}_{\alpha}:=\mathcal{J}_{\partial/\partial\mathrm{x}^{\alpha}%
}=T_{\alpha\nu}\vartheta^{\nu}\in\sec%
%TCIMACRO{\tbigwedge \nolimits^{1}}%
%BeginExpansion
{\textstyle\bigwedge\nolimits^{1}}
%EndExpansion
TM\hookrightarrow\mathcal{C\ell(}M,\mathtt{\eta}),~~~M\simeq\mathbb{R}^{4}
\label{p15}%
\end{equation}
are\ the conserved energy-momentum $1$-form fields of the physical system
under consideration, for which we know that the quantity (recall
Eq.(\ref{mm10}))
\begin{equation}
\boldsymbol{P}=P_{\alpha}\left.  \vartheta^{\alpha}\right\vert _{o}=P_{\alpha
}\boldsymbol{E}^{\alpha}, \label{P17}%
\end{equation}
with%
\begin{equation}
P_{\alpha}=%
%TCIMACRO{\dint }%
%BeginExpansion
{\displaystyle\int}
%EndExpansion
\underset{\boldsymbol{g}}{\star}\mathcal{T}_{\alpha} \label{p16}%
\end{equation}
are the components of the \emph{conserved} energy-momentum \emph{covector
}($\mathbf{CEMC}$) $\boldsymbol{P}$ of the system.\smallskip

\subsubsection{Limited Possibility to Construct a \textbf{CEMC} in
\textbf{GRT}}

Now, recall that in \textbf{GRT} a gravitational field generated by an
energy-momentum $\boldsymbol{T}$ is modelled by a Lorentzian
spacetime\footnote{In fact, by an equivalence classes of pentuples
$(M,\boldsymbol{g},\boldsymbol{D},\tau_{\boldsymbol{g}},\uparrow)$ modulo
diffeomorphisms.} $(M,\boldsymbol{g},\boldsymbol{D},\tau_{\boldsymbol{g}%
},\uparrow)$ where the relation between $\boldsymbol{g}$ and $\boldsymbol{T}$
is given by Einstein equation which using the orthonormal bases
$\{\boldsymbol{e}_{\alpha}\}$ and $\{\boldsymbol{\gamma}^{\alpha}\}$
introduced above reads%

\begin{equation}
G_{\beta}^{\alpha}=R_{\beta}^{\alpha}-\frac{1}{2}\delta_{\beta}^{\alpha
}R=-T_{\beta}^{\alpha}. \label{p18}%
\end{equation}
Moreover, defining $\boldsymbol{G}=G_{\beta}^{\alpha}\boldsymbol{\gamma
}^{\beta}\otimes\boldsymbol{e}_{\alpha}$ and recalling that $\boldsymbol{T=}%
T_{\beta}^{\alpha}\boldsymbol{\gamma}^{\beta}\otimes\boldsymbol{e}_{\alpha}$
it is%
\begin{equation}
\boldsymbol{D}\bullet\boldsymbol{G}=0,~~~~\boldsymbol{D}\bullet\boldsymbol{T}%
=0 \label{p18a}%
\end{equation}

Based \emph{only} on the contents of Section 2.1, given that $\boldsymbol{D}%
\bullet\boldsymbol{T}=0$, it may seems at first sigh\footnote{See however
Section 2.4 to learn that this naive expectation is incorrect.} that the only
possibility to construct conserved energy-momentum \emph{currents}
$\mathcal{T}_{\alpha}$ in \textbf{GRT}\ are for models of the theory where
appropriate Killing vector fields (such that one is timelike and the other
three spacelike) exist. However, an arbitrarily given Lorentzian manifold
$(M,\boldsymbol{g},D)$ in general does not have such Killing vector fields.

\begin{remark}
Moreover, even if it is the case that if a\ particular model
$(M,\boldsymbol{g},D)$ of \textbf{GRT} there exist one timelike and three
spacelike Killing vector fields we can construct \emph{the }scalar invariants
quantities $P_{\alpha}$, $\alpha=0,1,2,3$ given by \emph{Eq.(\ref{p16})}\ we
cannot\emph{ }define an energy-momentum covector $\boldsymbol{P}$ analogous to
the one given by \emph{Eq.(\ref{P17})} This is so because in this case to have
a conserved covector like $\boldsymbol{P}$ it is necessary to select a
$\boldsymbol{\gamma}^{\alpha}$ at a fixed point of the manifold. But in
general there is no physically meaningful way to do that, except if $M$ is
asymptotically flat\footnote{The concept of asymptotically flat Lorentzian
manifold can be rigorously formulated without the use of coordinates, as e.g.,
in \cite{wald}. However we will not need to enter in details here.} in which
case we can choose a chart such that at spatial infinity $\lim_{\left\vert
\vec{x}\right\vert \rightarrow\infty}\boldsymbol{\gamma}^{\alpha}(x^{0}%
,\vec{x})=\vartheta^{\alpha}=dx^{\alpha}$ and $\lim_{\left\vert \vec
{x}\right\vert \rightarrow\boldsymbol{\gamma}}\boldsymbol{g}_{\mu\nu}%
=\eta_{\mu\nu}$
\end{remark}

Thus, paroding Sachs and Wu \cite{sw} we must say that non existence of
genuine conservation laws for energy-momentum (and also angular momentum) in
\textbf{GRT} is a shame.

\begin{remark}
Despite what has been said above and the results of Section \emph{2.1} we next
show that there exists trivially an infinity of conserved currents \emph{(}the
Komar currents\emph{)} in any Lorentzian spacetime modelling a gravitational
field in $\mathbf{GRT}$. We discuss the meaning and disclose the form of these
currents, a result possible due to a notable decomposition of the square of
the Dirac operator acting on sections of the Clifford bundle $\mathcal{C\ell
}(M,\mathrm{g})$.
\end{remark}

\begin{remark}
We end this subsection recallng that in order to produce genuine conservation
laws in a field theory of gravitation with the gravitational field equations
equivalent \emph{(}in a precise sense\emph{)} to Einstein equation it is
necessary to formulate the theory in a parallelizable manifold and to dispense
the Lorentzian spacetime structure of $\mathbf{GRT}$. Details of such a theory
may be found in \emph{\cite{rod2012}}.
\end{remark}

\subsection{Komar Currents. Its Mathematical and Physical Meaning}

Let $\mathbf{A}\in\sec TM$ be the generator of a one parameter group of
diffeomorphisms of $M$ in the spacetime structure $\langle M,\boldsymbol{g}%
,\boldsymbol{D},\tau_{\boldsymbol{g}},\uparrow\rangle$ which is a model of a
gravitational field generated by $\boldsymbol{T}\in\sec T_{1}^{1}M$ (the
matter fields energy-momentum momentum tensor) in \textbf{GRT}. It is quite
obvious that if we define $F=dA$, where $A=\boldsymbol{g}(\mathbf{A,}$
$)\in\sec%
%TCIMACRO{\tbigwedge \nolimits^{1}}%
%BeginExpansion
{\textstyle\bigwedge\nolimits^{1}}
%EndExpansion
T^{\ast}M\hookrightarrow\mathcal{C\ell}(M,\mathrm{g})$,\ then the current
\begin{equation}
\boldsymbol{J}_{A}=-\underset{\boldsymbol{g}}{\delta}F\label{n0}%
\end{equation}
is conserved, i.e.,
\begin{equation}
\underset{\boldsymbol{g}}{\delta}\boldsymbol{J}_{A}=0.\label{n01}%
\end{equation}
Surprisingly such a trivial \emph{mathematical} result seems to be very
important by people working in \textbf{GRT} who call $\boldsymbol{J}_{A}$ the
Komar current\footnote{Komar called a \emph{related} quantity the generalized
flux.} \cite{komar}. Komar called\footnote{$V$ denotes a spacelike
hypersurface and $S=\partial V$ \ its boundary. Usualy the integral
$\mathfrak{E}$ is calculated at a constant $x^{0}$ time hypersurface and the
limit is taken for $S$ being the boundary at infinity.}
\begin{equation}
\mathfrak{E}=%
%TCIMACRO{\tint \nolimits_{V}}%
%BeginExpansion
{\textstyle\int\nolimits_{V}}
%EndExpansion
\underset{\boldsymbol{g}}{\star}\boldsymbol{J}_{A}=%
%TCIMACRO{\tint \nolimits_{\partial V}}%
%BeginExpansion
{\textstyle\int\nolimits_{\partial V}}
%EndExpansion
\underset{\boldsymbol{g}}{\star}F\label{n02}%
\end{equation}
the \emph{generalized energy.}

To understand why $\boldsymbol{J}_{A}$ is considered important in \textbf{GRT}
write the action for the gravitational plus matter and non gravitational
fields as
\begin{equation}
\mathcal{A=}%
%TCIMACRO{\tint }%
%BeginExpansion
{\textstyle\int}
%EndExpansion
\mathcal{L}_{g}+%
%TCIMACRO{\tint }%
%BeginExpansion
{\textstyle\int}
%EndExpansion
\mathcal{L}_{m}=-\frac{1}{2}%
%TCIMACRO{\tint }%
%BeginExpansion
{\textstyle\int}
%EndExpansion
R\boldsymbol{\tau}_{\boldsymbol{g}}+%
%TCIMACRO{\tint }%
%BeginExpansion
{\textstyle\int}
%EndExpansion
\mathcal{L}_{m}. \label{n022}%
\end{equation}

Now, the equations of motion for $\boldsymbol{g}$ can be obtained considering
its variation under an (infinitesimal) diffeormorphism $h:M\rightarrow M$
generated by $\mathbf{A}$. We have that $\boldsymbol{g}\mapsto\boldsymbol{g}%
^{\prime}=h^{\ast}\boldsymbol{g=g+\delta}^{0}\boldsymbol{g}$
where\footnote{Please do not confuse $\boldsymbol{\delta}^{0}$ with
$\underset{\boldsymbol{g}}{\delta}$.} the variation $\boldsymbol{\delta}%
^{0}\boldsymbol{g}=-\pounds _{\mathbf{A}}\boldsymbol{g}$. Taking into account
Cartan's magical formula ($\pounds _{\mathbf{A}}\boldsymbol{M}%
=A\underset{\boldsymbol{g}}{\boldsymbol{\lrcorner}}d\boldsymbol{M}%
+d(A\underset{\boldsymbol{g}}{\boldsymbol{\lrcorner}}\boldsymbol{M})$, for any
$\boldsymbol{M}\in\sec%
%TCIMACRO{\tbigwedge }%
%BeginExpansion
{\textstyle\bigwedge}
%EndExpansion
T^{\ast}M$) we have
\begin{align}
\boldsymbol{\delta}^{0}\mathcal{A}  &  =%
%TCIMACRO{\tint }%
%BeginExpansion
{\textstyle\int}
%EndExpansion
\boldsymbol{\delta}^{0}\mathcal{L}_{g}+%
%TCIMACRO{\tint }%
%BeginExpansion
{\textstyle\int}
%EndExpansion
\boldsymbol{\delta}^{0}\mathcal{L}_{m}\nonumber\\
&  =-%
%TCIMACRO{\tint }%
%BeginExpansion
{\textstyle\int}
%EndExpansion
\pounds _{\mathbf{A}}\mathcal{L}_{g}-%
%TCIMACRO{\tint }%
%BeginExpansion
{\textstyle\int}
%EndExpansion
\pounds _{\mathbf{A}}\mathcal{L}_{m}\nonumber\\
&  =-%
%TCIMACRO{\tint }%
%BeginExpansion
{\textstyle\int}
%EndExpansion
d(A\lrcorner\mathcal{L}_{g})-%
%TCIMACRO{\tint }%
%BeginExpansion
{\textstyle\int}
%EndExpansion
d(A\lrcorner\mathcal{L}_{m})\nonumber\\
&  :=%
%TCIMACRO{\tint }%
%BeginExpansion
{\textstyle\int}
%EndExpansion
d(\underset{\boldsymbol{g}}{\star}\mathcal{C)} \label{n03}%
\end{align}
where
\begin{equation}
\underset{\boldsymbol{g}}{\star}\mathcal{C=-}A\lrcorner\mathcal{L}%
_{g}-A\lrcorner\mathcal{L}_{m}+K \label{n03a}%
\end{equation}
with $dK=0.$

To proceed introduce a coordinate chart with coordinates $\{x^{\mu}\}$ for the
region of interest $U\subset M$.\ Recall that $\mathcal{G}^{\mu}%
=\mathcal{R}^{\mu}-\frac{1}{2}R\boldsymbol{\vartheta}^{\mu}$ are the Einstein
$1$-form fields\footnote{$\mathcal{G}^{\mu}=G_{\nu}^{\mu}\boldsymbol{\vartheta
}^{\nu}$ where $G_{\nu}^{\mu}=R_{\nu}^{\mu}-\frac{1}{2}\delta_{\nu}^{\mu}$ are
the components of the Einstein tensor. Moreover, we write $\mathcal{E}^{\mu
}=E_{\nu}^{\mu}\boldsymbol{\vartheta}^{\nu}$.}, with $\mathcal{R}^{\mu}%
=R_{\nu}^{\mu}\boldsymbol{\vartheta}^{\mu}$ the Ricci $1$-forms and $R$ the
curvature scalar. Einstein equation obtained form the variation principle
$\boldsymbol{\delta}^{0}\mathcal{A}=0$ is $\mathcal{E}^{\mu}:=\mathcal{G}%
^{\mu}+\mathcal{T}^{\mu}=0$, with $\mathcal{T}^{\mu}=T_{\nu}^{\mu
}\boldsymbol{\vartheta}^{\nu}$ the energy-momentum $1$-form fields and
moreover $D_{\mu}G^{\mu\nu}=0=D_{\mu}T^{\mu\nu}$.

Next write explicitly the action as \cite{ll}
\begin{equation}
\mathcal{A=-}\frac{1}{2}%
%TCIMACRO{\tint }%
%BeginExpansion
{\textstyle\int}
%EndExpansion
R\sqrt{-\det\boldsymbol{g}}dx^{0}dx^{1}dx^{2}dx^{3}+%
%TCIMACRO{\tint }%
%BeginExpansion
{\textstyle\int}
%EndExpansion
L_{m}\sqrt{-\det\boldsymbol{g}}dx^{0}dx^{1}dx^{2}dx^{3}. \label{n04}%
\end{equation}

We have immediately%
\begin{align}
\boldsymbol{\delta}^{0}\mathcal{A}  &  =\mathcal{-}\frac{1}{2}%
%TCIMACRO{\tint }%
%BeginExpansion
{\textstyle\int}
%EndExpansion
E^{\mu\nu}(\pounds _{\mathbf{A}}\boldsymbol{g})_{\mu\nu}\sqrt{-\det
\boldsymbol{g}}dx^{0}dx^{1}dx^{2}dx^{3}\nonumber\\
&  =-%
%TCIMACRO{\tint }%
%BeginExpansion
{\textstyle\int}
%EndExpansion
E^{\mu\nu}D_{\mu}A_{\nu}\sqrt{-\det\boldsymbol{g}}dx^{0}dx^{1}dx^{2}%
dx^{3}\nonumber\\
&  =-%
%TCIMACRO{\tint }%
%BeginExpansion
{\textstyle\int}
%EndExpansion
D_{\mu}(E^{\mu\nu}A_{\nu})\sqrt{-\det\boldsymbol{g}}dx^{0}dx^{1}dx^{2}%
dx^{3}\nonumber\\
&  =-%
%TCIMACRO{\tint }%
%BeginExpansion
{\textstyle\int}
%EndExpansion
(\boldsymbol{\partial}\underset{\boldsymbol{g}}{\boldsymbol{\lrcorner}%
}\mathcal{E}^{\nu}A_{\nu})\boldsymbol{\tau}_{\boldsymbol{g}}\nonumber\\
&  =%
%TCIMACRO{\tint }%
%BeginExpansion
{\textstyle\int}
%EndExpansion
\underset{\boldsymbol{g}}{\star}\underset{g}{\delta}(\mathcal{E}^{\nu}A_{\nu
})\nonumber\\
&  =-%
%TCIMACRO{\tint }%
%BeginExpansion
{\textstyle\int}
%EndExpansion
d(\underset{\boldsymbol{g}}{\star}\mathcal{E}^{\nu}A_{\nu}). \label{n05}%
\end{align}

From Eqs.(\ref{n03}) and (\ref{n05}) we have
\begin{equation}%
%TCIMACRO{\tint }%
%BeginExpansion
{\textstyle\int}
%EndExpansion
d(\underset{\boldsymbol{g}}{\star}\mathcal{E}^{\nu}A_{\nu}%
)+d(\underset{\boldsymbol{g}}{\star}\mathcal{C)}=0, \label{n06}%
\end{equation}
and thus
\begin{equation}
\underset{\boldsymbol{g}}{\delta}(\mathcal{E}^{\nu}A_{\nu}%
)+\underset{\boldsymbol{g}}{~\delta}\mathcal{C}=0 \label{n07}%
\end{equation}
Thus, the current $\mathcal{C}\in\sec%
%TCIMACRO{\tbigwedge \nolimits^{1}}%
%BeginExpansion
{\textstyle\bigwedge\nolimits^{1}}
%EndExpansion
T^{\ast}M$ is conserved if the field equations $\mathcal{E}^{\nu}=0$ are
satisfied. An equation (in component form) equivalent to Eq.(\ref{n07})
already appears in \cite{komar} (and also previously in \cite{berg} ) who took
$\mathcal{C=E}^{\nu}K_{\nu}+N$ with $\underset{\boldsymbol{g}}{\delta}N=0$.

Here, to continue we prefer to write an identity involving only
$\boldsymbol{\delta}^{0}\mathcal{A}_{g}=%
%TCIMACRO{\tint }%
%BeginExpansion
{\textstyle\int}
%EndExpansion
\boldsymbol{\delta}^{0}\mathcal{L}_{g}$. Proceeding exactly as before we get
putting $\mathcal{G(}A\mathcal{)=G}^{\mu}A_{\mu}$ that there exists
$\boldsymbol{N}\in\sec%
%TCIMACRO{\tbigwedge \nolimits^{1}}%
%BeginExpansion
{\textstyle\bigwedge\nolimits^{1}}
%EndExpansion
T^{\ast}M$ such that%
\begin{equation}
\underset{\boldsymbol{g}}{\delta}\mathcal{G(}A\mathcal{)~+~}%
\underset{\boldsymbol{g}}{\delta}\boldsymbol{N}=0. \label{n1}%
\end{equation}
and we see that we can identify%
\begin{equation}
\boldsymbol{N}:=-\mathcal{G}^{\mu}A_{\mu}+L \label{n4}%
\end{equation}
where $\underset{\boldsymbol{g}}{\delta}L=0$. Now, we claim that

\begin{proposition}
There exists a $L\in\sec%
%TCIMACRO{\tbigwedge \nolimits^{1}}%
%BeginExpansion
{\textstyle\bigwedge\nolimits^{1}}
%EndExpansion
T^{\ast}M$ such that%
\begin{equation}
\boldsymbol{N}=-\mathcal{G}^{\mu}A_{\mu}+L=\underset{\boldsymbol{g}}{\delta
}dA=-\boldsymbol{J}_{A}, \label{N44}%
\end{equation}
where $\boldsymbol{J}_{A}$ was defined in \emph{Eq.(\ref{n0})}.
\end{proposition}

\begin{proof}
To prove our claim we suppose from now one that $%
%TCIMACRO{\tbigwedge }%
%BeginExpansion
{\textstyle\bigwedge}
%EndExpansion
T^{\ast}M\hookrightarrow\mathcal{C\ell(}M,\mathtt{g})$%
\footnote{$\mathcal{C\ell(}M,\mathtt{g})$ is the Clifford bundle of
differential forms, see Appendix and if more details are necessary, cosnult,
e.g., \cite{rc2007}.}. Then it is possible to write
\begin{align}
\mathcal{G}^{\mu}A_{\mu}  &  =\mathcal{R}^{\mu}A_{\mu}-\frac{1}{2}%
RA\nonumber\\
&  =\boldsymbol{\partial}\wedge\boldsymbol{\partial}A-\frac{1}{2}RA\nonumber\\
&  =\boldsymbol{\partial}\wedge\boldsymbol{\partial}A+\boldsymbol{\partial
}\cdot\boldsymbol{\partial}A-\frac{1}{2}RA-\boldsymbol{\partial}%
\cdot\boldsymbol{\partial}A\nonumber\\
&  =\boldsymbol{\partial}^{2}A-\frac{1}{2}RA-\boldsymbol{\partial}%
\cdot\boldsymbol{\partial}A\nonumber\\
&  =-\underset{\boldsymbol{g}}{\delta}dA-d\underset{\boldsymbol{g}}{\delta
}A-\frac{1}{2}RA-\boldsymbol{\partial}\cdot\boldsymbol{\partial}A \label{n5}%
\end{align}
where $\boldsymbol{\partial}\wedge\boldsymbol{\partial}$ is the Ricci operator
and $\boldsymbol{\partial}\cdot\boldsymbol{\partial=\square}$ is the
D'Alembertian operator. Then we take%
\begin{equation}
-\mathcal{G}^{\mu}A_{\mu}-d\underset{\boldsymbol{g}}{\delta}A-\frac{1}%
{2}RA-\boldsymbol{\partial}\cdot\boldsymbol{\partial}%
A=~\underset{\boldsymbol{g}}{\delta}dA \label{N6}%
\end{equation}
and of course it is\footnote{Note that since $\underset{\boldsymbol{g}%
}{\delta}(\mathcal{G}^{\mu}A_{\mu})=0$ it follows from Eq.(\ref{n6aa}) that
indeed $\underset{\boldsymbol{g}}{\delta}L=0$.} \
\begin{equation}
L=-d\underset{\boldsymbol{g}}{\delta}A-\frac{1}{2}RA-\boldsymbol{\partial
}\cdot\boldsymbol{\partial}A \label{n6aa}%
\end{equation}
proving the proposition.
\end{proof}

Now that we found a $L$ satisfying Eq.(\ref{N44}) we investigate if we can
give some nontrivial physical meaning to such $\boldsymbol{N=}-\boldsymbol{J}%
_{A}\in\sec%
%TCIMACRO{\tbigwedge \nolimits^{1}}%
%BeginExpansion
{\textstyle\bigwedge\nolimits^{1}}
%EndExpansion
T^{\ast}M$.

\subsubsection{Determination of the Explicit Form of $\boldsymbol{J}_{A}$}

We recall that the extensor field $\boldsymbol{T}$ acts on $A$ as
\[
\boldsymbol{T}(A)=\mathcal{T}^{\mu}A_{\mu}%
\]
Thus, since%
\begin{equation}
\mathcal{G}^{\mu}A_{\mu}=-\boldsymbol{T}(A)
\end{equation}
we have from Eq.(\ref{N6})%

\begin{equation}
\underset{\boldsymbol{g}}{\delta}dA=\boldsymbol{T}\mathcal{(}A\mathcal{)}%
-d\underset{\boldsymbol{g}}{\delta}A-\frac{1}{2}RA-\boldsymbol{\partial}%
\cdot\boldsymbol{\partial}A. \label{n6a}%
\end{equation}
We can write Eq.(\ref{n6a}) taking into account that $R=\mathrm{tr}%
\boldsymbol{T}=T_{\mu}^{\mu}$ and putting $F:=dA$ that
\begin{equation}
\delta F=-\boldsymbol{J}_{A} \label{the equation}%
\end{equation}
where \cite{rrr2012}%
\begin{equation}
\boldsymbol{J}_{A}=-\boldsymbol{T}(A)+\frac{1}{2}\mathrm{tr}\boldsymbol{T}%
A\mathcal{+}\text{ }d\underset{\boldsymbol{g}}{\delta}A+\boldsymbol{\partial
}\cdot\boldsymbol{\partial}A \label{kocurr}%
\end{equation}
Eq.(\ref{kocurr}) gives the explicit form for the Komar
current\footnote{Something that is not given in \cite{komar}.}. Moreover,
taking into account that $\underset{\boldsymbol{g}}{\delta}%
F=\underset{\boldsymbol{g}}{\mathcal{\star}}d\underset{\boldsymbol{g}%
}{\mathcal{\star}}F$ it is
\begin{align}
d\star F  &  =\underset{\boldsymbol{g}}{\mathcal{\star}}^{-1}\left(
-\boldsymbol{T}\mathcal{(}A\mathcal{)}+\frac{1}{2}\mathrm{tr}\boldsymbol{T}%
A\mathcal{+}\text{ }d\underset{\boldsymbol{g}}{\delta}A+\boldsymbol{\partial
}\cdot\boldsymbol{\partial}A\right) \nonumber\\
&  =\underset{\boldsymbol{g}}{\mathcal{\star}}\left(  -\boldsymbol{T}%
\mathcal{(}A\mathcal{)}+\frac{1}{2}\mathrm{tr}\boldsymbol{T}A\mathcal{+}\text{
}d\underset{\boldsymbol{g}}{\delta}A+\boldsymbol{\partial}\cdot
\boldsymbol{\partial}A\right)  \label{n9}%
\end{align}
and thus taking into account Stokes theorem
\[%
%TCIMACRO{\tint \nolimits_{\mathcal{V}}}%
%BeginExpansion
{\textstyle\int\nolimits_{\mathcal{V}}}
%EndExpansion
d\underset{\boldsymbol{g}}{\star}F=%
%TCIMACRO{\tint \nolimits_{\partial\mathcal{V}}}%
%BeginExpansion
{\textstyle\int\nolimits_{\partial\mathcal{V}}}
%EndExpansion
\underset{\boldsymbol{g}}{\star}F=%
%TCIMACRO{\tint \nolimits_{\mathcal{V}}}%
%BeginExpansion
{\textstyle\int\nolimits_{\mathcal{V}}}
%EndExpansion
\underset{\boldsymbol{g}}{\mathcal{\star}}\boldsymbol{J}_{A}%
\]
Moreover, since $\underset{\boldsymbol{g}}{d\mathcal{\star}}\boldsymbol{J}%
_{A}=0$ we have that%
\[
0=%
%TCIMACRO{\tint \nolimits_{N}}%
%BeginExpansion
{\textstyle\int\nolimits_{N}}
%EndExpansion
\underset{\boldsymbol{g}}{d\mathcal{\star}}\boldsymbol{J}_{A}=%
%TCIMACRO{\dint \nolimits_{\partial N}}%
%BeginExpansion
{\displaystyle\int\nolimits_{\partial N}}
%EndExpansion
\underset{\boldsymbol{g}}{\mathcal{\star}}\boldsymbol{J}_{A}%
\]
and thus $%
%TCIMACRO{\dint \nolimits_{\Sigma_{1}}}%
%BeginExpansion
{\displaystyle\int\nolimits_{\Sigma_{1}}}
%EndExpansion
\underset{\boldsymbol{g}}{\mathcal{\star}}\boldsymbol{J}_{A}$ $\ $($\partial
N=\Sigma_{1}-\Sigma_{2}+\Xi$) is a conserved quantity. We arrive at the
conclusion that Taking $\mathcal{V}\subset$ $\Sigma_{1}$ ($\partial
\mathcal{V=S}_{R}$) as a ball of radius $R$ and making $R\rightarrow\infty$
the quantity
\begin{align}
\mathcal{E}  &  \mathcal{\equiv E(}A\mathcal{)}:=\frac{1}{8\pi}\int%
_{\mathcal{S}_{R}}\underset{\boldsymbol{g}}{\star}F\label{n10}\\
&  =\frac{1}{8\pi}%
%TCIMACRO{\tint \nolimits_{V}}%
%BeginExpansion
{\textstyle\int\nolimits_{V}}
%EndExpansion
\underset{\boldsymbol{g}}{\star}\left(  \boldsymbol{T}\mathcal{(}%
A\mathcal{)}-\frac{1}{2}A\mathrm{tr}\boldsymbol{T}\mathcal{-}\text{
}d\underset{\boldsymbol{g}}{\delta}A-\boldsymbol{\partial}\cdot
\boldsymbol{\partial}A\right)  \label{n11}%
\end{align}
is conserved.

\begin{remark}
It is very important to realize that quantity $\mathcal{E}$ defined by
\emph{Eq.(\ref{n10}) }is an scalar invariant, i.e., its value does not depend
on the particular reference frame $\boldsymbol{Z}$ and to the \emph{(}%
nacs$\mid Z)$\ \emph{(}the naturally adpated coordinate chart adapted to
$\boldsymbol{Z}$\emph{)\footnote{Recall that in Relativity theory (both
special and general) a reference frame is modelled by a time like vector field
$\boldsymbol{Z}$ pointing into the future. A \ naturally adapted coordinate
chart to $\boldsymbol{Z}$ (with coordinate functions $\{x^{\mu}\}$ (denoted
(nacs$\mid\boldsymbol{Z}$)) is one such that the spatial components of
$\boldsymbol{Z}$ are null. More details may be found, e.g., in Chapter 6 of
\cite{rc2007}.}. But, of course, }for each particular vector field
$\mathbf{A}\in\sec TM$ \emph{(}which generates a one parameter group of
diffeomorphisms\emph{) we have a different}.$\mathcal{E(}A\mathcal{)}$ and the
different $\mathcal{E(}A\mathcal{)}$
%TCIMACRO{\U{b4}}%
%BeginExpansion
\'{}%
%EndExpansion
s are not related as components of a covector.
\end{remark}

\begin{remark}
As we already remarked an equation equivalent to \emph{Eq.(\ref{n10})} has
already been obtained in \emph{\cite{komar}} who called (as said above) that
quantity the conserved \emph{generalized energy.} But according to our best
knowledge \emph{Eq.(\ref{n11})} is new and appears for the first time in
\emph{\cite{rrr2012}}.

However, considering that for each $\mathbf{A}$ $\in\sec TM$ that generates a
one parameter group of diffeomorphisms of $M$ we have a conserved quantity it
is not in our opinion appropriate to think about this quantity as a
generalized energy. Indeed, why should the energy depends on terms like
$d\underset{\boldsymbol{g}}{\delta}A$ and $\boldsymbol{\partial}%
\cdot\boldsymbol{\partial}A$ if $A$ is not a dynamical field?

We know\emph{ }that \emph{\cite{R2010}} when $\mathbf{A=K}$ is a Killing
vector field it is $\underset{\boldsymbol{g}}{\delta}A=0$ and
$\boldsymbol{\partial}\cdot\boldsymbol{\partial}A=-\boldsymbol{T}%
\mathcal{(}A\mathcal{)}+\frac{1}{2}\mathrm{tr}\boldsymbol{T}A$ and thus
\emph{Eq.(\ref{n11}) reads}%

\begin{equation}
\mathcal{E}=\frac{1}{4\pi}%
%TCIMACRO{\tint \nolimits_{N}}%
%BeginExpansion
{\textstyle\int\nolimits_{N}}
%EndExpansion
\underset{\boldsymbol{g}}{\star}\mathcal{(}\boldsymbol{T}(A)-\frac{1}%
{2}A\mathrm{tr}\boldsymbol{T}) \label{n12}%
\end{equation}
which is a well known conserved quantity\footnote{An equivalent formula
appears, e.g., \ as Eq.(11.2.10) in \cite{wald}. However, it is to be
emphasized here the simplicity and transparency of our approach concerning
traditional ones based on classical tensor calculus.}. For a Schwarzschild
spacetime, as well known, $\mathbf{A=\partial/\partial}t$ is a timelike
Killing vector field and in his case since the components of $\boldsymbol{T}$
are $T_{\nu}^{\mu}=\frac{8\pi}{\sqrt{-\det\boldsymbol{g}}}\rho(r)v^{\mu}%
v_{\nu}$ and $v^{i}v_{j}=0$ \emph{(}since $v^{\mu}=\frac{1}{\sqrt{g_{00}}%
}(1,0,0,0)$\emph{)} we get $\mathcal{E}=m$.
\end{remark}

\begin{remark}
Note that that the conserved quantity given by \emph{Eq.(\ref{n12})}\ differs
in general\ from the conserved quantity obtained with the current defined in
\emph{Eq.(\ref{p6}) when }$V=K$\emph{ }which holds in any structure
$(M,\boldsymbol{g},D)$ with the conditions given there. However, in the
particular case analyzed above \emph{Eq.(\ref{n12})} and \emph{Eq.(\ref{p14a}%
)} give the same result.
\end{remark}

\begin{remark}
Originally Komar obtained the same result as in \emph{Eq.(\ref{n12})} directly
from \emph{Eq.(\ref{n10}) }supposing that the generator of the one parameter
group of diffeomorphism was\emph{ }$\mathbf{A}=\partial/\partial t$. So, he
got $\mathcal{E}=m$ by pure chance. If he had picked another vector field
generator of a one parameter group of diffeomorphisms $\mathbf{A\neq
\partial/\partial}t$, he of course, would not obtained that result.
\end{remark}

\begin{remark}
The previous remark shows clearly that the construction of Komar currents does
not to solve the energy-momentum conservation problem for a system consisting
of the matter and non gravitational fields plus the gravitational field in
$\mathbf{GRT}$.

Indeed, too claim that a solution for a meaningful definition for the
energy-momentum of the total system\emph{\footnote{The total system is the
system consisting of the gravitational plus matter and non gravitational
fields.}} exist it is necessary to find a way to define a total conserved
energy-momentum covector for the total system as it is possible to do in field
theories in Minkowski spacetime (recall Section \emph{C.3.4}). This can only
be done if the spacetime structure modeling a gravitational field
\emph{(}generated by the matter fields energy-momentum tensor $\boldsymbol{T}%
$\emph{)} possess appropriate additional structure, or if we interpret the
gravitational field as a field in the Faraday sense living in Minkowski
spacetime. More details in \emph{\cite{rc2007,rod2012}}.
\end{remark}

\subsubsection{The Maxwell Like Equation $\boldsymbol{\partial}%
F=\boldsymbol{J}_{A}$ Encodes Einstein Equation}

From Eq.(\ref{the equation})\emph{ }where $F=dA\ $($A=\boldsymbol{g}%
(\mathbf{A,})$) with $\mathbf{A}$ $\in\sec TM$ an arbitrary generator of a one
parameter group of diffeomorphisms of $M$ \emph{(}part of the structure
$(M,\boldsymbol{g},D,\mathbf{\tau}_{\boldsymbol{g}}\boldsymbol{,\uparrow)}%
$\emph{)} taking into account that $dF=0$, we get the Maxwell like equation
(\textbf{MLE})%
\begin{equation}
\boldsymbol{\partial}F=\boldsymbol{J}_{A}. \label{EM}%
\end{equation}
with a well defined conserved current. Of course, as we already said, there is
an infinity of such equations. Each one \emph{encodes} Einstein equation,
i.e., given the form of $\boldsymbol{J}_{A}$ (Eq.(\ref{kocurr})) we can get
back Eq.(\ref{n5}), which gives immediately Einstein equation (\textbf{EE}).
In this sense we can claim:

\begin{center}%
\begin{tabular}
[c]{l}%
\textbf{EE} $\boldsymbol{G}=\boldsymbol{T}$ and the \textbf{MLE}
$\boldsymbol{\partial}F\boldsymbol{=J}_{A}$ are equivalent.
\end{tabular}

\end{center}

\begin{remark}
Finally it is worth to emphasize that the above results show that in
$\mathbf{GRT}$ there are infinity of conservation laws, one for each vector
field generator of a one parameter group of diffeomorphisms and so, Noether's
theorem in $\mathbf{GRT}$ which follows from the supposition that the
Lagrangian density is invariant under the diffeomorphism group gives only
identities, i.e., an infinite set of conserved currents, each one encoding as
we saw above Einstein equation.
\end{remark}

It is now time to analyze the possible generalized conservation laws and their
implications for the motions of probe single-pole particles in
\emph{Lorentzian} and \emph{teleparallel} de Sitter spacetime structures,
where these structures are \emph{not} supposed to represent models of
gravitational fields in \textbf{GRT} and compare these results with the ones
in \textbf{GRT}. This will be one in the next sections.

\section{The Lorentzian de Sitter $M^{dS\ell}$ Structure and its Conformal
Representation}

Let $SO(1,4)$ and $SO(1,3)$ be respectively the special pseudo-orthogonal
groups in $\mathbb{R}^{1,4}=\{\mathring{M}=\mathbb{R}^{5}%
,\boldsymbol{\mathring{g}}\}$ and in $\mathbb{R}^{1,3}=\{\mathbb{R}%
^{4},\boldsymbol{\eta}\}$ where $\boldsymbol{\mathring{g}}$ is a metric of
signature $(1,4)$ and $\boldsymbol{\eta}$ a metric of signature $(1,3)$.\ The
manifold $M=SO(1,4)/SO(1,3)$ will be called the \emph{de Sitter manifold}.
Since
\begin{equation}
M=SO(1,4)/SO(1,3)\approx\mathbb{R\times}S^{3} \label{0}%
\end{equation}
this manifold can be viewed as a brane (a submanifold) in the structure
$\mathbb{R}^{1,4}$. We now introduce a Lorentzian spacetime, i.e., the
structure $M^{dSL}=(M=\mathbb{R\times}S^{3},\boldsymbol{g},\boldsymbol{D}%
,\tau_{\boldsymbol{g}},\uparrow)$ which will be called\emph{ Lorentzian de
Sitter spacetime structure} where if $\boldsymbol{\iota}:\mathbb{R\times}%
S^{3}\rightarrow\mathbb{R}^{5}$, $\boldsymbol{g=\iota}^{\ast}%
\boldsymbol{\mathring{g}}$ and $\boldsymbol{D}$ is the parallel projection on
$M$ of the pseudo Euclidian metric compatible connection in $\mathbb{R}^{1,4}$
(details in \cite{rs2013}). As well known, $M^{dSL}$ is a spacetime of
constant Riemannian curvature. It has ten Killing vector fields. The Killing
vector fields are the generators of infinitesimal actions of the group
$SO(1,4)$ (called the de Sitter group) in $M=\mathbb{R\times}S^{3}\approx
SO(1,4)/SO(1,3)$. The group $SO(1,4)$ acts transitively\footnote{A group $G$
of transformations in a manifold $M$ ($\sigma:G\times M\rightarrow M$ by
$(g,x)\mapsto\sigma(g,x)$) is said to act transitively on $M$ if for
arbitraries $x,y\in M$ there exists $g\in G$ such that $\sigma(g,x)=y$.} in
$SO(1,4)/SO(1,3)$, which is thus a homogeneous space (for $SO(1,4)$).

We now recall the description of the manifold $\mathbb{R\times}S^{3}$ as a
pseudo-sphere (a submanifold) of radius $\ell$ of the pseudo Euclidean space
$\mathbb{R}^{1,4}=\{\mathbb{R}^{5},\boldsymbol{\mathring{g}}\}$. If
$(X^{0},X^{1},X^{2},X^{3},X^{4})$ are the global coordinates of $\mathbb{R}%
^{1,4}$ then the equation representing the pseudo sphere is%
\begin{equation}
(X^{0})^{2}-(X^{1})^{2}-(X^{2})^{2}-(X^{3})^{2}-(X^{4})^{2}=-\ell^{2}.
\label{ds4}%
\end{equation}

Introducing \emph{conformal} coordinates\footnote{Figure 1 appears also in
author's paper \cite{RWC2015}.} $\{x^{\mu}\}$ by projecting the points of
$\mathbb{R\times}S^{3}$ from the \textquotedblleft
north-pole\textquotedblright\ to a plane tangent to the \textquotedblleft
south pole" we see immediately that $\{x^{\mu}\}$ covers\ all $\mathbb{R\times
}S^{3}$ except the \textquotedblleft north-pole\textquotedblright. We
immediately find that%

\begin{equation}
\boldsymbol{g}=i_{\ast}\boldsymbol{\mathring{g}}=\Omega^{2}\eta_{\mu\nu
}dx^{\mu}\otimes dx^{\nu} \label{ds1}%
\end{equation}
where%
\begin{equation}
X^{\mu}=\Omega x^{\mu},~~~~X^{4}=-\ell\Omega\left(  1+\frac{\sigma^{2}}%
{4\ell^{2}}\right)  \label{ds1a}%
\end{equation}%
\begin{equation}
\Omega=\left(  1-\frac{\sigma^{2}}{4\ell^{2}}\right)  ^{-1} \label{ds2}%
\end{equation}
and%
\begin{equation}
\sigma^{2}=\eta_{\mu\nu}x^{\mu}x^{\nu}. \label{ds3}%
\end{equation}

Since the north pole of the pseudo sphere is not covered by the coordinate
functions we see that (omitting two dimensions) the region of the spacetime as
seem by an observer living the south pole is the region inside the so called
absolute of \emph{Cayley-Klein} of equation
\begin{equation}
t^{2}-x^{2}=4\ell^{2}. \label{ds5}%
\end{equation}

In Figure 1 we can see that all timelike curves (1) and (2) and lightlike (3)
starts in the \textquotedblleft past horizon\textquotedblright\ and end on the
\textquotedblleft future horizon\textquotedblright.%

%TCIMACRO{\FRAME{ftbpFU}{4.2557in}{3.9185in}{0pt}{\Qcb{Conformal Representation
%of de Sitter Spacetime. Note that the "observer" spacetime is interior to the
%Caley-Klein absolute $t^{2}-\vec{x}^{2}=4\ell^{2}$.}}{\Qlb{ds2.eps}}%
%{ds2.eps}{\special{ language "Scientific Word";  type "GRAPHIC";
%width 4.2557in;  height 3.9185in;  depth 0pt;  original-width 4.2056in;
%original-height 3.87in;  cropleft "0";  croptop "1";  cropright "1";
%cropbottom "0";  filename 'ds2.eps';file-properties "XNPEU";}} }%
%BeginExpansion
\begin{figure}[ptb]%
\centering
\includegraphics[
height=3.9185in,
width=4.2557in
]%
{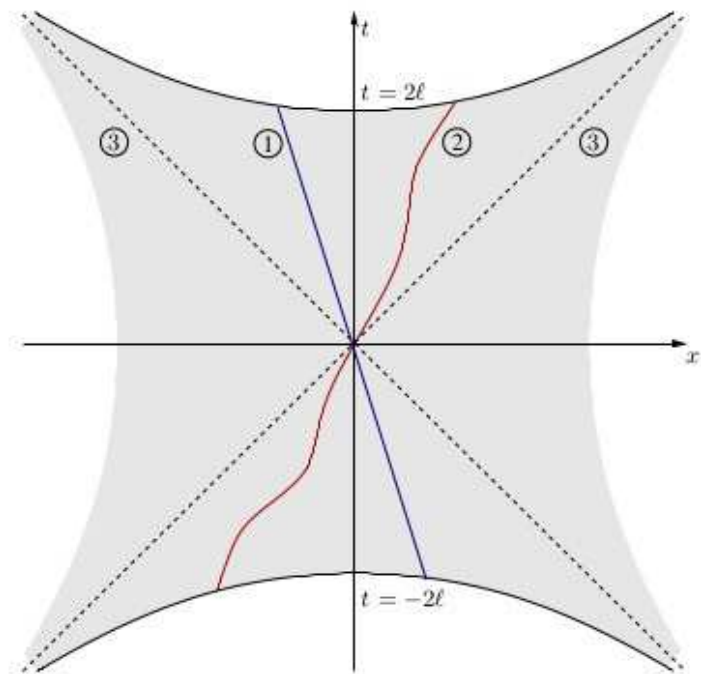}%
\caption{Conformal Representation of de Sitter Spacetime. Note that the
"observer" spacetime is interior to the Caley-Klein absolute $t^{2}-\vec
{x}^{2}=4\ell^{2}$.}%
\label{ds2.eps}%
\end{figure}
%EndExpansion

\section{On the Geodesics of the $M^{dSL}$}

In a classic book by Hawking and Ellis \cite{he} we can read on page 126 the
following statement:

\begin{quote}
de Sitter spacetime is geodesically complete; however there are points in the
space which cannot be joined to each other by any geodesic.
\end{quote}

Unfortunately many people do not realize that the points that cannot be joined
by a geodesic are some points which can be joined by a spacelike curve (living
in region 3 in Figure 1). So, these curves are never the path of any particle.
A complete and thoughtful discussion of this issue is given in an old an
excellent article by Schmidt \cite{schmidt}.

\begin{remark}
Having said that, let us recall that among the Killing vector
fields\emph{\footnote{More details in Section 5.1.}} of $(M=\mathbb{R\times
}S^{3},\boldsymbol{g})$ there are one timelike and three spacelike vector
fields \emph{(}which are called the \emph{translation} Killing vector fields
in physical literature\emph{)}. So, many people though since a long time ago
that this permits the formulation of an energy-momentum conservation law in de
Sitter spacetime $M^{dSL}$ structure . However, the fact is that what can
really be done is the obtainment of conserved quantities like the $P_{\alpha}$
in \emph{Eq.(\ref{p16})}. But we cannot obtain in $M^{dSL}$ structure an
energy-momentum covector like $\boldsymbol{P}$ for a given closed physical
system using an equation similar to \emph{Eq.(\ref{P17})}. This is because de
Sitter spacetime is \emph{not} asymptotically flat and so there is no way to
physically determine a point to fix the $\boldsymbol{\gamma}^{\alpha}$ to use
in an equation similar to \emph{Eq.(\ref{P17})}.\smallskip
\end{remark}

Now, the \textquotedblleft translation\textquotedblright\ Killing vector
fields of de Sitter spacetime are expressed in the coordinate basis
$\{\partial_{\mu}=\frac{\partial}{\partial x^{\mu}}\}$ (where $\{x^{\mu}\}$
are the conformal coordinates introduced in the last section) by\footnote{We
are using here a notation similar to the ones in \cite{ps2012} for compariosn
of some of our results with the ones obtained there.}:%
\begin{equation}
\mathbf{\Pi}_{\alpha}=\xi_{\alpha}^{\mu}\partial_{\mu} \label{K1}%
\end{equation}
with (\emph{putting} $\ell=1$, \textit{for simplicity}) are
\begin{equation}
\xi_{\alpha}^{\mu}=\delta_{\alpha}^{\mu}-(\frac{\eta_{\alpha\nu}x^{\nu}x^{\mu
}}{2}-\frac{\sigma^{2}}{4}\delta_{\alpha}^{\mu}). \label{k2}%
\end{equation}%
\begin{equation}%
\begin{array}
[c]{c}%
\xi_{\alpha}^{0}=\delta_{\alpha}^{0}-(\frac{\sigma^{2}}{4}\frac{\eta
_{\alpha\nu}x^{\nu}x^{0}}{2}-\frac{\sigma^{2}}{4}\delta_{\alpha}^{0}),\\
\xi_{\alpha}^{1}=\delta_{\alpha}^{1}-(\frac{\sigma^{2}}{4}\frac{\eta
_{\alpha\nu}x^{\nu}x^{1}}{2}-\frac{\sigma^{2}}{4}\delta_{\alpha}^{1}),\\
\xi_{\alpha}^{2}=\delta_{\alpha}^{2}-(\frac{\sigma^{2}}{4}\frac{\eta
_{\alpha\nu}x^{\nu}x^{2}}{2}-\frac{\sigma^{2}}{4}\delta_{\alpha}^{2}),\\
\xi_{\alpha}^{3}=\delta_{\alpha}^{3}-(\frac{\sigma^{2}}{4}\frac{\eta
_{\alpha\nu}x^{\nu}x^{3}}{2}-\frac{\sigma^{2}}{4}\delta_{\alpha}^{3}).
\end{array}
\label{k2a}%
\end{equation}
So, we want now to investigate if there is some region of de Sitter spacetime
where the \textquotedblleft translational\textquotedblright\ Killing vector
fields are linearly independent.

In order to proceed, we take an arbitrary vector field $\mathbf{V}=V^{\mu
}\partial_{\mu}\in\sec T\mathcal{U}$. If the $\mathbf{\Pi}_{\alpha}$ are
linearly independent in a region $\mathcal{U}^{\prime}\subset\mathcal{U}%
\subset\mathbb{R\times}S^{3}$ then we can write
\begin{equation}
\mathbf{V}=V^{\mu}\partial_{\mu}=\boldsymbol{V}^{\boldsymbol{\alpha}%
}\mathbf{\Pi}_{\boldsymbol{\alpha}}=\boldsymbol{V}^{\boldsymbol{\alpha}}%
\xi_{\boldsymbol{\alpha}}^{\mu}\partial_{\mu}. \label{k3}%
\end{equation}
So, the condition for the existence of a nontrivial solution for the
$\boldsymbol{V}^{\alpha}$ is that
\begin{equation}
\det\left(
\begin{array}
[c]{cccc}%
\xi_{0}^{0} & \xi_{0}^{1} & \xi_{0}^{2} & \xi_{0}^{3}\\
\xi_{1}^{0} & \xi_{1}^{1} & \xi_{1}^{2} & \xi_{1}^{3}\\
\xi_{2}^{0} & \xi_{\boldsymbol{1}}^{1} & \xi_{1}^{2} & \xi_{1}^{3}\\
\xi_{3}^{0} & \xi_{3}^{1} & \xi_{3}^{2} & \xi_{3}^{3}%
\end{array}
\right)  \neq0 \label{k3a}%
\end{equation}

Thus, putting
\[
\chi_{\mu}:=(\frac{\sigma^{2}}{4}\frac{(x^{\mu})^{2}}{2}-\frac{\sigma^{2}}%
{4})
\]
we need to evaluate the determinant of the matrix%

\begin{equation}
\left(
\begin{array}
[c]{cccc}%
1-\chi_{0} & (\frac{\sigma^{2}}{4}\frac{tx}{2}) & (\frac{\sigma^{2}}{4}%
\frac{ty}{2}) & (\frac{\sigma^{2}}{4}\frac{tz}{2})\\
-(\frac{\sigma^{2}}{4}\frac{xt}{2}) & 1-\chi_{1} & -(\frac{\sigma^{2}}{4}%
\frac{xy}{2}) & -(\frac{\sigma^{2}}{4}\frac{xz}{2})\\
-(\frac{\sigma^{2}}{4}\frac{yt}{2}) & -(\frac{\sigma^{2}}{4}\frac{yx}{2}) &
1-\chi_{2} & -(\frac{yz}{\frac{m^{2}}{4}2})\\
-(\frac{\sigma^{2}}{4}\frac{zt}{2}) & -(\frac{\sigma^{2}}{4}\frac{zx}{2}) &
-(\frac{\sigma^{2}}{4}\frac{zy}{2}) & 1-\chi_{3}%
\end{array}
\right)  . \label{k4}%
\end{equation}
Then%

\begin{gather}
\det[\xi_{\alpha}^{\mu}]=-\frac{1}{512}\sigma^{10}-\frac{3}{32}\sigma
^{6}+\frac{1}{16}\sigma^{6}+\sigma^{2}+1-\frac{3}{128}\sigma^{8}-\frac{1}%
{8}\sigma^{4}\nonumber\\
+\frac{1}{1024}\sigma^{8}(3t^{2}x^{2}y^{2}z^{2}+t^{2}y^{2}z^{2}-x^{2}%
y^{2}z^{2}+y^{2}z^{2})+\frac{1}{512}\sigma^{8}(t^{2}x^{2}y^{2}+t^{2}x^{2}%
z^{2})\nonumber\\
+\frac{1}{128}\sigma^{6}t^{2}x^{2}(y^{2}+z^{2})+\frac{1}{256}\sigma^{6}%
(t^{2}y^{2}z^{2}-x^{2}y^{2}z^{2}+2y^{2}z^{2})\nonumber\\
+\frac{1}{64}\sigma^{4}(t^{2}x^{2}y^{2}z^{2}+y^{2}z^{2}-x^{2}y^{2}z^{2}%
)+\frac{1}{16}\sigma^{2}(t^{2}y^{2}z^{2}-^{2}x^{2}y^{2}z^{2}-2y^{2}%
z^{2})\nonumber\\
+\frac{3}{8}\sigma^{4}-\frac{1}{4}y^{2}z^{2}+\frac{1}{256}\sigma^{8}.
\label{k5}%
\end{gather}

In order to analyze this expression we put (without loss of generality, see
the reason in \cite{schmidt}) $y=z=0$. In this case%

\begin{equation}
\det[\xi_{\alpha}^{\mu}]=\frac{1}{512}\left(  \sigma^{2}+4\right)  ^{3}\left(
-\sigma^{4}+2\sigma^{2}+8\right)  \label{k6}%
\end{equation}
which is null on the Caley-Klein absolute, i.e., at the points
\begin{equation}
t^{2}-\vec{x}^{2}=4. \label{k7}%
\end{equation}
and also on the spacelike hyperbolas given by%
\[
t^{2}-\vec{x}^{2}=-2\text{ and }t^{2}-\vec{x}^{2}=-4
\]

So, we have proved that the translational Killing vector fields are linearly
independent in all the sub region inside the Cayley-Klein absolute except in
the points of the hyperbolas $t^{2}-\vec{x}^{2}=-2$ and $t^{2}-\vec{x}^{2}%
=-4$.$\smallskip$

This result is very much important for the following reason. Timelike
geodesics in de Sitter spacetime structure $M^{dSL}$ are (as well known) the
curves $\sigma:s\rightarrow\sigma(s)$, where $s$ is propertime along $\sigma$,
that extremizes the length function \cite{oneill} , i.e., calling
$\sigma_{\ast}=\boldsymbol{u}$ the velocity of a \textquotedblleft
particle\textquotedblright\ of mass $m$ following a time like geodesic we have
that the equation of the geodesic is obtained by finding an extreme of the
action, writing here in sloop notation, as
\begin{equation}
I[\sigma]=-m\int_{\sigma}ds=-m\int_{\sigma}(g_{\mu\nu}dx^{\mu}dx^{\nu}%
)^{\frac{1}{2}}. \label{k8}%
\end{equation}
As well known, the determination of an extreme for $I[\sigma]$ is given by
evaluating the first variation\footnote{In Eq.(\ref{k88}) $\boldsymbol{Y}$ is
the deformation vector field determining the curves $\sigma(s,l)$\ necessary
to calculate the first variation of $I[\sigma].$} of $I[\sigma]$, i.e.,%
\begin{equation}
\boldsymbol{\delta}^{0}I[\sigma]=m\int\pounds _{\boldsymbol{Y}}(g_{\mu\nu
}dx^{\mu}dx^{\nu})^{\frac{1}{2}.} \label{k88}%
\end{equation}
and putting $\boldsymbol{\delta}^{0}I[\sigma]=0$. The result, as well known is
the geodesic equation
\begin{equation}
\boldsymbol{D}_{\boldsymbol{u}}\boldsymbol{u}=0. \label{k9}%
\end{equation}

Now, taking into account that the Killing vector fields determines a basis
inside the Cayley-Klein absolute we write
\begin{equation}
\boldsymbol{u}=u^{\mu}\partial_{\mu}=U^{\alpha}\mathbf{\Pi}_{\alpha}%
=U^{\alpha}\xi_{\alpha}^{\mu}\partial_{\mu}. \label{k10}%
\end{equation}
We now define the \textquotedblleft hybrid\textquotedblright\ connection
\ coefficients\footnote{Take notice that $\mathbf{\Gamma}_{\cdot
\mu\boldsymbol{\alpha}}^{\boldsymbol{\beta}\cdot\cdot}\neq\Gamma_{\cdot
\mu\alpha}^{\beta\cdot\cdot}$ where $D_{\partial_{\mu}}\partial_{\alpha
}=\Gamma_{\cdot\mu\alpha}^{\beta\cdot\cdot}\partial_{\beta}$, for otherwsise
confusion will arise.} $\mathbf{\Gamma}_{\cdot\mu\boldsymbol{\alpha}%
}^{\boldsymbol{\beta}\cdot\cdot}$ \ by:
\begin{equation}
\boldsymbol{D}_{\partial_{\mu}}\mathbf{\Pi}_{\alpha}:=\mathbf{\Gamma}%
_{\cdot\mu\boldsymbol{\alpha}}^{^{\boldsymbol{\beta}}\cdot\cdot}\Pi_{\beta
},~~~~~\boldsymbol{D}_{\partial_{\mu}}\mathbf{\Pi}^{\alpha}:=-\mathbf{\Gamma
}_{\cdot\mu\beta}^{\alpha\cdot\cdot}\mathbf{\Pi}^{\beta} \label{K11}%
\end{equation}
and write the geodesic equation as
\begin{gather}
\boldsymbol{D}_{\boldsymbol{u}}\boldsymbol{u}=u^{\mu}\boldsymbol{D}%
_{\partial\mu}(U^{\alpha}\mathbf{\Pi}_{\alpha})\nonumber\\
=(u^{\mu}\partial_{\mu}U^{\alpha})\mathbf{\Pi}_{\alpha}+u^{\mu}U^{\alpha
}\boldsymbol{D}_{\partial\mu}\mathbf{\Pi}_{\alpha}\nonumber\\
=(u^{\mu}\partial_{\mu}U^{\boldsymbol{\beta}})\mathbf{\Pi}_{\boldsymbol{\beta
}}+u^{\mu}U^{\boldsymbol{\alpha}}\mathbf{\Gamma}_{\cdot\mu\alpha}^{\beta
\cdot\cdot}\mathbf{\Pi}_{\beta}\nonumber\\
=(\frac{dU^{\beta}}{ds}+u^{\mu}U^{\alpha}\mathbf{\Gamma}_{\cdot\mu\alpha
}^{\beta\cdot\cdot})\mathbf{\Pi}_{\beta}=0. \label{k11}%
\end{gather}
or
\begin{equation}
\frac{dU^{\beta}}{ds}+u^{\mu}U^{\alpha}\mathbf{\Gamma}_{\cdot\mu\alpha}%
^{\beta\cdot\cdot}=0, \label{k12}%
\end{equation}
which on multiplying by the \textquotedblleft mass\textquotedblright\ $m$ and
calling $\pi^{\beta}=mU^{\boldsymbol{\beta}}$ and $\pi_{\rho}=g_{\rho\beta}%
\pi^{\beta}$ can be written equivalently as%
\begin{equation}
\frac{d\pi_{\rho}}{ds}-\mathbf{\Gamma}_{\cdot\mu\boldsymbol{\rho}}^{\beta
\cdot\cdot}u^{\mu}\pi_{\beta}=0, \label{k13}%
\end{equation}
which looks like Eq.(37) in \cite{ps2012}.\medskip

We want now to investigate the question: is Eq.(\ref{k13}) the same as Eq.(37)
in \cite{ps2012}?

To know the answer to the above question recall that in \cite{ps2012} authors
investigate the variation of $I[\sigma]$ under a variation of the curves
$\sigma(s)\mapsto$ $\sigma(s,\ell)$ induced by a coordinate transformation
$x^{\mu}\mapsto x^{\prime\mu}+\mathbf{\delta}_{\mathbf{\Pi}}x^{\mu}$ where
they put
\begin{equation}
\mathbf{\delta}_{\mathbf{\Pi}}x^{\mu}=\xi_{\boldsymbol{\rho}}^{\mu
}(x)\mathbf{\delta}x^{\rho} \label{k13a}%
\end{equation}
with $\xi_{\boldsymbol{\rho}}^{\mu}(x)$ being the components of the Killing
vector fields $\mathbf{\Pi}_{\alpha}$\ (recall Eq.(\ref{k2})) and where it is
said that $\mathbf{\delta}x^{\rho}$ is an \emph{ordinary} variation. However
in \cite{ps2012} we cannot find what authors mean by ordinary variation, and
so $\mathbf{\delta}x^{\rho}$ is not defined.

So, to continue our analysis we recall, that if the $\mathbf{\delta}x^{\rho
}=\varepsilon^{\alpha}$ are constants (but arbitrary) then $\boldsymbol{\delta
}_{\Pi}x^{\mu}$ corresponds to a diffeomorphism generated by a Killing vector
field $\mathbf{\Pi}=\varepsilon^{\alpha}\mathbf{\Pi}_{\alpha}$. However, if
the $\mathbf{\delta}x^{\rho}=\lambda^{\boldsymbol{\rho}}(x)$ are infinitesimal
\emph{arbitrary} functions (i.e., $\left\vert \lambda^{\boldsymbol{\rho}%
}(x)\right\vert <<1$), then the notation $\mathbf{\delta}_{\Pi}x^{\mu}$ is
misleading\ since $\xi_{\boldsymbol{\rho}}^{\mu}\lambda^{\boldsymbol{\rho}%
}(x)$ is the $\boldsymbol{\delta}^{0}I[\sigma]$ variation generated by a quite
\emph{arbitrary} vector field $\boldsymbol{Y}=\lambda^{\boldsymbol{\rho}%
}(x)\xi_{\boldsymbol{\rho}}^{\mu}(x)\partial_{\mu}=\lambda^{\boldsymbol{\rho}%
}(x)\mathbf{\Pi}_{\rho}=Y^{\mu}\partial_{\mu}$. In this case we get from
$\boldsymbol{\delta}^{0}I[\sigma]=0$ the geodesic equation.

\subsection{Curves Obtained From Constrained Variations.}

So, let us study the \emph{constrained variation} when the $\mathbf{\delta
}x^{\rho}=\varepsilon^{\alpha}$ are constants (but arbitrary). We denote the
constrained variation by $\boldsymbol{\delta}^{c}I[\sigma]$. In this case
starting from Eq.(\ref{k8})
\begin{equation}
\boldsymbol{\delta}^{c}I[\sigma]=-m\int_{\sigma}\{u^{\gamma}(\boldsymbol{D}%
_{\gamma}u_{\beta})\xi_{\boldsymbol{\rho}}^{\beta}\}\mathbf{\delta}x^{\rho}ds
\label{k13b}%
\end{equation}
where (taking account of notations already introduced) it is
\begin{align}
u^{\gamma}\boldsymbol{D}_{e_{\gamma}}(u_{\beta}e^{\beta})  &  =u^{\gamma
}(\boldsymbol{D}_{\gamma}u_{\beta})e^{\beta},\nonumber\\
\boldsymbol{D}_{\gamma}u_{\beta}  &  =\partial_{\gamma}u_{\beta}-\Gamma
_{\cdot\gamma\beta}^{\tau\cdot\cdot}u_{\tau},\nonumber\\
\boldsymbol{D}_{e_{\gamma}}e_{\beta}  &  =\Gamma_{\cdot\gamma\beta}^{\tau
\cdot\cdot}e_{\tau},~~~\boldsymbol{D}_{e_{\gamma}}e^{\beta}=-\Gamma
_{\cdot\gamma\tau}^{\beta\cdot\cdot}e^{\tau}, \label{k13c}%
\end{align}
we can write%
\begin{gather}
\boldsymbol{\delta}^{c}I[\sigma]=-m\int\{u^{\gamma}(\boldsymbol{D}_{\gamma
}u_{\beta})\xi_{\boldsymbol{\rho}}^{\beta}\}\boldsymbol{\delta}x^{\rho
}ds\nonumber\\
=-m\int\{u^{\gamma}\{(\boldsymbol{D}_{\gamma}(u^{\beta}\xi_{\boldsymbol{\rho
}\beta}))\}\boldsymbol{\delta}x^{\rho}ds+\frac{m}{2}\int\{u^{\gamma}u^{\beta
}(\boldsymbol{D}_{\gamma}\xi_{\beta\rho}+\boldsymbol{D}_{\beta}\xi_{\rho
}^{\gamma}\}\boldsymbol{\delta}x^{\rho}ds\nonumber\\
=-m\int\{u^{\gamma}(\boldsymbol{D}_{\gamma}(\boldsymbol{u\cdot}\Pi
_{\boldsymbol{\rho}}))\}\delta x^{\rho}ds=-m\int[u^{\gamma}\boldsymbol{D}%
_{\gamma}(\mathcal{U}_{\boldsymbol{\rho}})\}\boldsymbol{\delta}x^{\rho}ds.
\label{k13da}%
\end{gather}
and $\boldsymbol{\delta}^{c}I[\sigma]=0$ implies
\begin{equation}
u^{\gamma}\boldsymbol{D}_{\gamma}(mU_{\boldsymbol{\rho}})=0. \label{k13e}%
\end{equation}

Eq.(\ref{k13e}) with $\pi_{\boldsymbol{\rho}}=mU_{\boldsymbol{\rho}}$ can be
written as
\begin{equation}
\frac{d}{ds}\pi_{\boldsymbol{\rho}}-u^{\gamma}\pi_{\beta}\Gamma_{\cdot
\gamma\boldsymbol{\rho}}^{\beta\boldsymbol{\cdot\cdot}}=0 \label{k13f}%
\end{equation}
which is Eq.(37) in \ \cite{ps2012}. Note that this equations looks like the
geodesic equation written as Eq.(\ref{k13}) above, but it is in fact
\emph{different} since of course, recalling Eq.(\ref{K11}) it is
$\mathbf{\Gamma}_{\cdot\gamma\boldsymbol{\rho}}^{\beta\boldsymbol{\cdot\cdot}%
}\neq\Gamma_{\cdot\gamma\boldsymbol{\rho}}^{\beta\boldsymbol{\cdot\cdot}}$. In
\cite{ps2012} $\pi_{\boldsymbol{\rho}}=mU_{\boldsymbol{\rho}}$ is
unfortunately wrongly interpreted as the components of a covector field over
$\sigma$ supposed to be the energy-momentum covector of the particle, because
authors of \cite{ps2012} supposed to have proved that this equation could be
derived from Papapetrou's method, what is not the case as we show in Section 6.

\section{Generalized Energy-Momentum Conservation Laws in de Sitter Spacetime
Structures}

\subsection{Lie Algebra of the de Sitter Group}

Given a structure $(\mathring{M}\simeq\mathbb{R}^{5},\boldsymbol{\mathring{g}%
},)$ introduced in Section 3 define $\mathbf{J}_{AB}\in\sec T\mathring{M}$ by
\begin{equation}
\mathbf{J}_{AB}:=\eta_{AC}X^{C}\frac{\partial}{\partial X^{B}}-\eta_{BC}%
X^{C}\frac{\partial}{\partial X^{A}}. \label{em0}%
\end{equation}
These objects are generators of the Lie algebra\ $\mathrm{so}(1,4)$ of the de
Sitter group.

Using the bases $\{\partial_{\mu}\},\{dx^{\mu}\}$ introduced above the ten
Killing vector fields of de Sitter spacetime are the fields $\mathbf{J}%
_{\alpha4}\in\sec TM$ and $\mathbf{J}_{\mu\nu}\in\sec TM$ and it
is\footnote{Proofs of Eqs.(\ref{em01}) and (\ref{em1}) are in Appendix F. Of
course, in those equations, it is $\boldsymbol{P}_{\alpha}=e_{\alpha}$ (as
introduced \ is Section 2).}
\begin{align}
\mathbf{J}_{\alpha4}  &  =\eta_{\alpha C}X^{C}\frac{\partial}{\partial X^{4}%
}-\eta_{4C}X^{C}\frac{\partial}{\partial X^{\alpha}}=\ell\boldsymbol{P}%
_{\alpha}-\frac{1}{4\ell}\boldsymbol{K}_{4\alpha}\label{em01}\\
&  =\ell\partial_{\alpha}-\frac{1}{4\ell}(2\eta_{\alpha\lambda}x^{\lambda
}x^{\nu}-\sigma^{2}\delta_{\alpha}^{\nu})\partial_{\nu}.\nonumber\\
\mathbf{J}_{\mu\nu}  &  =\eta_{\mu\kappa}X^{\kappa}\frac{\partial}{\partial
X^{\nu}}-\eta_{\nu\kappa}X^{\kappa}\frac{\partial}{\partial X^{\mu}}=\eta
_{\mu\lambda}x^{\lambda}\boldsymbol{P}_{\nu}-\eta_{\nu\lambda}x^{\lambda
}\boldsymbol{P}_{\mu}. \label{em1}%
\end{align}
The $J_{\alpha4}\in\sec TM$ and $J_{\mu\nu}\in\sec TM$ satisfy the Lie algebra
$\mathrm{so}(1,4)$ of the de Sitter group, this time acting as a
transformation group acting transitively in de Sitter spacetime. We have
\begin{align}
\lbrack\mathbf{J}_{\alpha4},\mathbf{J}_{\beta4}]  &  =\mathbf{J}_{\alpha\beta
},\nonumber\\
\lbrack\mathbf{J}_{\alpha\beta},\mathbf{J}_{\lambda4}]  &  =\eta_{\lambda
\beta}\mathbf{J}_{\alpha4}-\eta_{\lambda\alpha}\mathbf{J}_{\beta4},\nonumber\\
\lbrack\mathbf{J}_{\alpha\beta},\mathbf{J}_{\lambda\tau}]  &  =\eta
_{\alpha\lambda}\mathbf{J}_{\beta\tau}+\eta_{\beta\tau}\mathbf{J}%
_{\alpha\lambda}-\eta_{\beta\lambda}\mathbf{J}_{\alpha\tau}-\eta_{\alpha\tau
}\mathbf{J}_{\beta\lambda}. \label{em2}%
\end{align}
It is usual in physical applications to define
\begin{equation}
\boldsymbol{\Pi}_{\alpha}=\mathbf{J}_{\alpha4}/\ell\label{em3}%
\end{equation}
for then we have
\begin{align}
\lbrack\boldsymbol{\Pi}_{\alpha},\boldsymbol{\Pi}_{\beta}]  &  =\frac{1}%
{\ell^{2}}\mathbf{J}_{\alpha\beta},\nonumber\\
\lbrack\mathbf{J}_{\alpha\beta},\mathbf{\Pi}_{\lambda}]  &  =\eta
_{\lambda\beta}\boldsymbol{\Pi}_{\alpha}-\eta_{\lambda\alpha}\boldsymbol{\Pi
}_{\beta},\nonumber\\
\lbrack\mathbf{J}_{\alpha\beta},\mathbf{J}_{\lambda\tau}]  &  =\eta
_{\alpha\lambda}\mathbf{J}_{\beta\tau}+\eta_{\beta\tau}\mathbf{J}%
_{\alpha\lambda}-\eta_{\beta\lambda}\mathbf{J}_{\alpha\tau}-\eta_{\alpha\tau
}\mathbf{J}_{\beta\lambda}. \label{em4}%
\end{align}

The Killing vector fields $\mathbf{J}_{\alpha\beta}$ satisfy the Lie algebra
$so(1,3)$ (of the special Lorentz group).

\begin{remark}
From \emph{Eq.(\ref{em4})} we see that when $\ell\mapsto\infty$ the Lie
algebra of $so(1,4)$ goes into the Lie algebra of the Poincar\'{e} group
$\mathfrak{P}$ which is the \ semi-direct sum of the group of translations in
$\mathbb{R}^{4}$ plus component of the special Lorentz group, i.e.,
$\mathfrak{P}=(\mathbb{R}^{4}\boxplus SO(1,3))$. This is eventually the
justification for \ physicists to call the $\boldsymbol{\Pi}_{\alpha}$ the
\textquotedblleft translation\textquotedblright\ generators of the de Sitter group.

However, it is necessary to\ have in mind that whereas the translation
subgroup $\mathbb{R}^{4}$ of $\mathfrak{P}$ acts transitively in Minkowski
spacetime manifold, the set $\{\boldsymbol{\Pi}_{\alpha}\}$ does not close in
a subalgebra of the de Sitter algebra and thus it is impossible in general to
find $\exp(\lambda^{\alpha}\boldsymbol{\Pi}_{\alpha})$ such that\ given
arbitrary $x,y\in\mathbb{R\times}S^{3}$ it is $y=\exp(\lambda^{\alpha
}\boldsymbol{\Pi}_{\alpha})x$. Only the whole group $SO(1,4)$ acts
transitively on $\mathbb{R\times}S^{3}$.\smallskip
\end{remark}

\paragraph{Casimir Invariants}

Now, if $\{\mathbf{E}^{A}=dX^{A}\}$ is an orthonormal basis for the structure
$\mathbb{R}^{1,4}=(\mathbb{R}^{5},\boldsymbol{\mathring{g}})$ define the
angular momentum operator as the Clifford algebra valued operator
\begin{equation}
\mathbf{J=}\frac{1}{2}\mathbf{E}^{A}\wedge\mathbf{E}^{B}\mathbf{J}_{AB}.
\label{MN1}%
\end{equation}

Taking into account the results of Appendix A its (Clifford) square is%
\begin{equation}
\mathbf{J}^{2}=\mathbf{J\lrcorner J+J\wedge J=-J\cdot J+J\wedge J.}
\label{MN2}%
\end{equation}

It is immediate to verify that $\mathbf{J}^{2}$ is invariant under the
transformations of the de Sitter group. $\mathbf{J\lrcorner J}$ is (a constant
appart) the first invariant Casimir operator of the de Sitter group. The
second invariant Casimir operator of the de Sitter is related to
$\mathbf{J\wedge J}$. Indeed, defining%

\begin{equation}
\mathbf{W:=}\underset{\boldsymbol{\mathring{g}}}{\star}\frac{1}{8\ell
}(\mathbf{J\wedge J})=\frac{1}{8\ell}(\mathbf{J\wedge J})\lrcorner
\tau_{\boldsymbol{\mathring{g}}}. \label{c11a}%
\end{equation}
one can easily show (details in \cite{RWC2015}) that
\begin{align}
\mathbf{W\cdot W}  &  =\mathbf{WW}=\mathbf{W}^{2}=-\frac{1}{64\ell^{2}%
}(\mathbf{J\wedge J})\lrcorner(\mathbf{J\wedge J})=-\frac{1}{64\ell^{2}%
}(\mathbf{J\wedge J})\cdot(\mathbf{J\wedge J})\nonumber\\
&  =-\frac{1}{64\ell^{2}}(\mathbf{J\wedge J})(\mathbf{J\wedge J}). \label{c12}%
\end{align}
is indeed an invariant operator.

As well known, the representations of the de Sitter group are classified by
their Casimir invariants $I_{1}$ and $I_{2}$ which here following
\cite{gursey} we take as%

\begin{align}
I_{1}  &  =-\mathbf{J\lrcorner J=}\frac{1}{2\ell^{2}}\mathbf{J}_{AB}%
\mathbf{J}^{AB}=\eta^{\alpha\beta}\boldsymbol{\Pi}_{\alpha}\boldsymbol{\Pi
}_{\beta}+\frac{1}{2\ell^{2}}\text{ }\eta^{\alpha\lambda}\eta^{\beta\tau
}\mathbf{J}_{\alpha\beta}\mathbf{J}_{\lambda\tau}=M^{2},\nonumber\\
I_{2}  &  =\mathbf{W}^{2}=\boldsymbol{W}^{A}\boldsymbol{W}_{A}=\eta
^{\alpha\beta}\boldsymbol{V}_{\alpha}\boldsymbol{V}_{\beta}+\frac{1}{\ell^{2}%
}(\boldsymbol{W}_{4})^{2}, \label{m5}%
\end{align}
where $M\in\mathbb{R}$ and the fields $\boldsymbol{W}_{A}$ and $\boldsymbol{V}%
_{\alpha}$ are defined by:%
\begin{align}
\boldsymbol{W}_{A}  &  :=\frac{1}{8\ell}\varepsilon_{ABCDE}\mathbf{J}%
^{AB}\mathbf{J}^{DE}\nonumber\\
\boldsymbol{V}_{\alpha}  &  :=-\frac{1}{2}\varepsilon_{4\alpha\lambda\mu\nu
}\eta^{\lambda\rho}\boldsymbol{\Pi}_{\rho}\mathbf{J}^{\mu\nu}, \label{m6}%
\end{align}
from where it follows that
\begin{equation}
\boldsymbol{W}_{4}=\frac{1}{8}\varepsilon_{4\mu\nu\rho\tau}\mathbf{J}^{\mu\nu
}\mathbf{J}^{\mu\nu}. \label{m7}%
\end{equation}
In the limit when $\ell\mapsto\infty$ we get the Casimir operators of the
special Lorentz group
\begin{align}
I_{1}  &  \mapsto\boldsymbol{P}_{\alpha}\boldsymbol{P}^{\alpha}=m^{2}%
\mathbf{,}\nonumber\\
I_{2}  &  \mapsto\eta^{\alpha\beta}\boldsymbol{V}_{\alpha}\boldsymbol{V}%
_{\beta}=m^{2}s(s+1) \label{m8}%
\end{align}
where $m\in\mathbb{R}$ and $s=0,1/2,1,3/2,...$\smallskip

We see that $\boldsymbol{\Pi}_{\alpha}$ looks like the components of an
\emph{energy-momentum vector\ }$\boldsymbol{P}=\left.  P_{\alpha}%
\vartheta^{\alpha}\right\vert _{o}$ \emph{of a closed physical system }(see
Eq.(\ref{P17})) in the Minkowski spacetime of Special Relativity, for which
$\boldsymbol{P}^{2}=m^{2}$, with $m$ the mass of the system. However, take
into account that whereas the $P_{\alpha}$ are simple real numbers, the
$\boldsymbol{\Pi}_{\alpha}$ are vector fields.

Moreover, take into account that $\eta^{\alpha\beta}\boldsymbol{\Pi}_{\alpha
}\boldsymbol{\Pi}_{\beta}$ is not an invariant, i.e., it does not commute with
the generators of the Lie algebra of the de Sitter group.\medskip

\subsection{Generalized Energy-Momentum\ Covector for a Closed System in the
Teleparallel de Sitter Spacetime Structure}

In this section we suppose that $(M=\mathbb{R\times}S^{3},\boldsymbol{g}%
,\tau_{\boldsymbol{g}},\uparrow)$ is the physical arena where Physics take place.

We know that $(M=\mathbb{R\times}S^{3},\boldsymbol{g}$) has ten Killing vector
fields and four of them (one timelike and three spacelike, $\boldsymbol{\Pi
}_{\alpha}$, $\alpha=0,1,2,3$ ) generated \textquotedblleft
translations\textquotedblright. Thus if we suppose that $(M=\mathbb{R\times
}S^{3},\boldsymbol{g},\tau_{\boldsymbol{g}},\uparrow)$ is populated
by\ interacting matter fields $\{\phi_{1,}...,\phi_{n}\}$ with dynamics
described by a Lagrangian formalism we can construct as described in the
Appendix C.1 the conserved currents
\begin{equation}
\mathcal{J}_{\mathbf{\Pi}_{\alpha}}=\mathcal{J}_{\alpha}^{\beta}%
\boldsymbol{\gamma}_{\beta}=\mathcal{J}_{\alpha}^{\mu}\boldsymbol{\vartheta
}_{\mu} \label{CCURR}%
\end{equation}
were taking into account that $\boldsymbol{\delta}_{\mathbf{\Pi}_{\alpha}}%
^{0}=-\pounds _{\mathbf{\Pi}_{\alpha}}$ and denoting by $\boldsymbol{\delta
}_{\mathbf{\Pi}_{\alpha}}$ each particular\emph{ local variation} generated by
$\mathbf{\Pi}_{\alpha}$ (recall Eq.(\ref{K1})) we have%
\begin{equation}
\boldsymbol{\delta}_{\mathbf{\Pi}_{\alpha}}\phi_{A}=\boldsymbol{\delta
}_{\mathbf{\Pi}_{\alpha}}^{0}\phi_{A}+\mathbf{\delta}x^{\nu}\partial_{\nu}%
\phi_{A} \label{ok}%
\end{equation}
and thus we write
\begin{align}
\mathcal{J}_{\alpha}^{\mu}  &  =\pi_{A}^{\mu}\boldsymbol{\delta}_{\mathbf{\Pi
}_{\alpha}}\phi_{A}+\Upsilon_{\nu}^{\mu}\boldsymbol{\delta}_{\mathbf{\Pi
}_{\alpha}}^{0}x^{\mu}\nonumber\\
&  =\Lambda_{\nu}^{\mu}\boldsymbol{\xi}_{\alpha}^{\nu}+\Upsilon_{\nu}^{\mu
}\boldsymbol{\xi}_{\alpha}^{\nu} \label{ok2}%
\end{align}
where
\begin{equation}
\Lambda_{\nu}^{\mu}\boldsymbol{\xi}_{\alpha}^{\nu}:=\pi_{A}^{\mu
}\boldsymbol{\delta}_{\mathbf{\Pi}_{\alpha}}\phi_{A}. \label{ok3}%
\end{equation}
In Eq.(\ref{CCURR}) $\{\boldsymbol{\gamma}^{\beta}\}$ is the dual basis of
the\ orthonormal basis $\{\boldsymbol{e}_{\alpha}\}$ defined
by\footnote{Recall that the fields $\boldsymbol{e}_{\alpha}$ are only defined
in\ subset $\{S^{3}-$north pole$\}.$}
\begin{equation}
\boldsymbol{e}_{\alpha}:=\frac{\boldsymbol{\Pi}_{\alpha}}{\boldsymbol{g}%
(\boldsymbol{\Pi}_{\alpha},\boldsymbol{\Pi}_{\alpha})}. \label{m10}%
\end{equation}
\ \ From the conserved currents $\mathcal{J}_{\mathbf{\Pi}_{\alpha}}$ we can
obtain four conserved quantities (the $P_{\alpha}$ in Eq.(\ref{p16})).

\begin{remark}
It is crucial to observe that the above results have been deduced
\emph{without} introduction of any connection in the structure
$(M=\mathbb{R\times}S^{3},\boldsymbol{g,}\tau_{\boldsymbol{g}},\uparrow)$.
However that structure is not enough for using the $P_{\alpha}$ to build a
\emph{(}generalized\emph{)} covector analogous to the energy-momentum covector
$\boldsymbol{P}$ \emph{(}see \emph{Eq.(\ref{P17}))} of special relativistic theories.
\end{remark}

If we add $\boldsymbol{D}$, the Levi-Civita connection of $\boldsymbol{g}$ to
$(M=\mathbb{R\times}S^{3},\boldsymbol{g,}\tau_{\boldsymbol{g}},\uparrow)$ we
get the Lorentzian de Sitter spacetime structure $M^{dSP}$ and defining a
generalized energy-momentum tensor
\begin{equation}
\mathbf{\Theta}=\mathcal{J}_{\mathbf{\Pi}_{\alpha}}\otimes e^{\alpha}\in\sec
T_{1}^{1}M\label{M10A}%
\end{equation}
\qquad we know that $\underset{\boldsymbol{g}}{\delta}\mathcal{J}%
_{\mathbf{\Pi}_{\alpha}}=0$ implies $\boldsymbol{D}\cdot\mathbf{\Theta}=0$, a
covariant \textquotedblleft conservation\textquotedblright\ law.

However, the introduction of $M^{dSP}$ in our game is of no help to construct
a covector like $\boldsymbol{P}$ since in $M^{dSP}$ vectors at different
spacetime points cannot be directly compared..\smallskip

So, the question arises: is it possible to define a structure where $\forall
x,y\in\mathbb{R\times}S^{3}$ we can define objects $P_{\alpha}^{\mathrm{dS}}$
such that
\begin{equation}
\boldsymbol{P}_{\mathrm{dS}}=P_{\alpha}^{\mathrm{dS}}\left.
\boldsymbol{\gamma}^{\alpha}\right\vert _{x}=P_{\alpha}^{\mathrm{dS}}\left.
\boldsymbol{\gamma}^{\alpha}\right\vert _{y} \label{m11}%
\end{equation}
defines a legitimate covector for a closed physical system living in a de
Sitter structure $(M=\mathbb{R\times}S^{3},\boldsymbol{g,}\tau_{\boldsymbol{g}%
},\uparrow)$ and for which $\boldsymbol{P}_{\mathrm{dS}}$ can be said to be a
kind of generalization of the momentum of the closed system in Minkowski
spacetime?\smallskip

We show now the answer is positive. We recall that\emph{ }$\boldsymbol{D}$ has
been introduced in our developments only as a useful mathematical device and
is quite irrelevant in the construction of legitimate conservation laws since
the conserved currents $\mathcal{J}_{\mathbf{\Pi}_{\alpha}}$ have been
obtained without the use of any connection. So, we now introduce for our goal
a \emph{teleparallel de Sitter spacetime}, i.e., the structure $M^{dSTP}%
=(M=\mathbb{R\times}S^{3},\boldsymbol{\nabla},\tau_{\boldsymbol{g}},\uparrow)$
where $\boldsymbol{\nabla}$ is a metric compatible teleparallel connection
defined by
\begin{equation}
\boldsymbol{\nabla}_{\boldsymbol{e}_{\alpha}}\boldsymbol{e}_{\beta}%
=\omega_{\cdot\alpha\beta}^{\kappa\cdot\cdot}\boldsymbol{e}_{\kappa}=0.
\label{m9}%
\end{equation}
Under this condition we know that we can identify all tangent and all
cotangent spaces. So, we have for $\forall x,y\in\mathbb{R\times}S^{3},$
\[
\left.  \boldsymbol{e}_{\alpha}\right\vert _{x}\simeq\left.  \boldsymbol{e}%
_{\alpha}\right\vert _{y}\text{~~and~~}\left.  \boldsymbol{\gamma}^{\alpha
}\right\vert _{x}\simeq\left.  \boldsymbol{\gamma}^{\alpha}\right\vert
_{y}=\boldsymbol{E}^{\alpha},
\]
where $\{\boldsymbol{E}^{\alpha}\}$ is a basis of a vector space
$\mathcal{V\simeq}\mathbb{R}^{4}$.

Thus, in the structure $M^{dSTP}$ Eq.(\ref{m11}) defines indeed a legitimate
covector in $\mathcal{V\simeq}\mathbb{R}^{4}$ and thus permits a legitimate
\emph{generalization} of the concepts of energy-momentum covector obtained for
physical theories in Minkowski spacetime. The term generalization is a good
one here because in the limit where $\ell\rightarrow\infty$, $\boldsymbol{\Pi
}_{\alpha}\mapsto\partial/\partial x^{\alpha},$ $\Lambda_{\alpha}^{\kappa}=0$
and thus $\Theta_{\alpha}^{\kappa}=\Upsilon_{\alpha}^{\kappa}$.

\subsection{The Conserved Currents $\boldsymbol{J}_{\mathbf{\Pi}_{\alpha}%
}=\boldsymbol{g}(\boldsymbol{\Pi}_{\alpha},\boldsymbol{~})$}

To proceed we show that the translational Killing vector fields of the de
Sitter structure $(M,\boldsymbol{g})$ determines trivially conserved currents%

\[
\boldsymbol{J}_{\mathbf{\Pi}_{\alpha}}=\boldsymbol{g}(\boldsymbol{\Pi}%
_{\alpha},\boldsymbol{~})
\]

Indeed \ using the $M^{dSTP}$ structure as a convenient device we recall the
result proved in \cite{R2010} that for each vector Killing $\mathbf{K}$ the
one form field $K=\boldsymbol{g}(\mathbf{K},\boldsymbol{~})$ is such that
$\underset{\boldsymbol{g}}{\delta}K=0$. So, it is
\begin{equation}
\underset{\boldsymbol{g}}{\delta}\boldsymbol{J}_{\mathbf{\Pi}_{\alpha}}=0.
\label{m9a}%
\end{equation}
and of course, also
\begin{equation}
\underset{\boldsymbol{g}}{\delta}\boldsymbol{J}_{\mathbf{\Pi}}=0 \label{m9aa}%
\end{equation}
where
\begin{equation}
\boldsymbol{J}_{\mathbf{\Pi}}:=\boldsymbol{g}(\boldsymbol{\Pi},\boldsymbol{~}%
),~~~~\boldsymbol{\Pi:=\varepsilon}^{\alpha}\boldsymbol{\Pi}_{\alpha}
\label{m10a}%
\end{equation}
with the $\boldsymbol{\varepsilon}^{\alpha}$ real constants such that
$\left\vert \boldsymbol{\varepsilon}^{\alpha}\right\vert <<1$. Recalling from
Eq.(\ref{em1}) that in projective conformal coordinate bases ($\{e_{\mu
}=\partial_{\mu}\}$ and $\{\boldsymbol{\vartheta}^{\mu}=dx^{\mu}\}$) the
components of $\boldsymbol{\Pi}_{\alpha}$ are
\begin{equation}
\xi_{\alpha}^{\mu}=\delta_{\alpha}^{\mu}-\frac{1}{4\ell^{2}}(2\eta_{\alpha
\rho}x^{\rho}x^{\mu}-\sigma^{2}\delta_{\alpha}^{\mu}) \label{newa}%
\end{equation}
we get the conserved current $\boldsymbol{J}_{\mathbf{\Pi}}$ is%
\begin{equation}
\boldsymbol{J}_{\mathbf{\Pi}}=\boldsymbol{\varepsilon}^{\alpha}\boldsymbol{J}%
_{\mathbf{\Pi}_{\alpha}}=\boldsymbol{\varepsilon}^{\alpha}\xi_{\alpha}^{\mu
}\boldsymbol{\vartheta}_{\mu}. \label{new}%
\end{equation}

\begin{remark}
Take notice that of course, this current is not the conserved current that we
found in the previous section.
\end{remark}

\begin{remark}
Take notice also that in \emph{\cite{ps2012}} authors trying to generalize the
results that follows from the canonical formalism for the case of field
theories in Minkowski spacetime suppose that they can eliminate the term
$\pi_{A}^{\mu}\boldsymbol{\delta}\phi_{A}$ from \emph{Eq.(\ref{ok2})} by
\emph{decree} postulating a \textquotedblleft new kind of \textquotedblleft
local variation\textquotedblright, call it $\boldsymbol{\delta}^{\prime}$\ for
which $\boldsymbol{\delta}^{\prime}\phi_{A}=0$. The fact is that such a
\textquotedblleft new kind of local variation\textquotedblright\ never appears
in the canonical Lagrangian formalism, only $\boldsymbol{\delta}\phi_{A}$
appears and in general it is not zero.
\end{remark}

Now, the \emph{generalized} \emph{canonical de Sitter energy-momentum tensor
is}%
\begin{equation}
\mathbf{\Theta}=\mathcal{J}_{\mathbf{\Pi}_{\alpha}}\otimes e^{\alpha}%
=\Theta_{\alpha}^{\kappa}\boldsymbol{\vartheta}_{\kappa}\otimes e^{\alpha}%
\in\sec\bigwedge\nolimits^{1}T^{\ast}M\otimes TM \label{NEW2A}%
\end{equation}
and making analogy with the case of Minkowski spacetime where the
$\Upsilon_{\alpha}^{\kappa}$ have been defined as the components of the
canonical energy-momentum tensor\footnote{The explict form of the $K_{\alpha
}^{\kappa}$ can be determined withouth difficulty if needed.}
$\boldsymbol{\Upsilon}$, we write $\Theta_{\alpha}^{\kappa},$ as%
\begin{equation}
\Theta_{\alpha}^{\kappa}=\Upsilon_{\alpha}^{\kappa}-\frac{1}{4\ell^{2}%
}K_{\alpha}^{\kappa}+\Lambda_{\alpha}^{\kappa}, \label{nwaa}%
\end{equation}
Thus each one of the genuine conservation laws $d\underset{\boldsymbol{g}%
}{\star}\mathcal{J}_{\mathbf{\Pi}_{\alpha}}=0$, $\alpha=0,1,2,3$ read in
coordinate basis components%
\begin{gather}
(\partial_{\mu}\Upsilon_{\alpha}^{\kappa}-\Gamma_{\cdot\lambda\nu}%
^{\lambda\cdot\cdot}\Upsilon_{\alpha}^{\kappa})-\frac{1}{4\ell^{2}}%
(\partial_{\mu}K_{\alpha}^{\mu}-\Gamma_{\cdot\lambda\nu}^{\lambda\cdot\cdot
}K_{\alpha}^{\kappa})-(\partial_{\mu}\Lambda_{\alpha}^{\kappa}-\Gamma
_{\cdot\lambda\nu}^{\lambda\cdot\cdot}\Lambda_{\alpha}^{\kappa})\nonumber\\
=\frac{1}{\sqrt{-\det\boldsymbol{g}}}\partial_{\mu}(\sqrt{-\det\boldsymbol{g}%
}\Upsilon_{\alpha}^{\mu})-\frac{1}{4\ell^{2}}\frac{1}{\sqrt{-\det
\boldsymbol{g}}}\partial_{\mu}(\sqrt{-\det\boldsymbol{g}}K_{\alpha}^{\mu
})\label{NEW3}\\
+\frac{1}{\sqrt{-\det\boldsymbol{g}}}\partial_{\mu}(\sqrt{-\det\boldsymbol{g}%
}\Lambda_{\alpha}^{\mu})=0\nonumber
\end{gather}
or%
\begin{equation}
\partial_{\mu}\left[  \sqrt{-\det\boldsymbol{g}}\left(  \Upsilon_{\alpha}%
^{\mu}-\frac{1}{4\ell^{2}}K_{\alpha}^{\mu}+\Lambda_{\alpha}^{\mu}\right)
\right]  =0\Leftrightarrow\boldsymbol{D}_{\mu}\Theta_{\alpha}^{\mu}=0.
\label{new4}%
\end{equation}

\begin{remark}
Now, recalling the relation of the connection coefficients of the bases
$\{\boldsymbol{e}_{\alpha}\}$ of the Levi-Civita connection $\boldsymbol{D}$
of $\boldsymbol{g}$ \emph{(denoted }$\mathbf{\Gamma}_{\cdot\alpha\beta
}^{\kappa\cdot\cdot}$\emph{) }and the coefficients of the basis
$\{\boldsymbol{e}_{\alpha}\}$ of the teleparallel connection $\nabla$ of
$\boldsymbol{g}$ \emph{(denoted }$\mathbf{\bar{\Gamma}}_{\cdot\alpha\beta
}^{\kappa\cdot\cdot}$\emph{) }are \emph{\cite{rc2007}}
\begin{equation}
\mathbf{\bar{\Gamma}}_{\cdot\alpha\beta}^{\kappa\cdot\cdot}=\mathbf{\Gamma
}_{\cdot\alpha\beta}^{\kappa\cdot\cdot}+\mathrm{\triangle}_{\cdot\alpha\beta
}^{\kappa\cdot\cdot} \label{REL}%
\end{equation}
where
\begin{equation}
\mathrm{\triangle}_{\cdot\alpha\beta}^{\kappa\cdot\cdot}:=-\frac{1}{2}\left(
\mathrm{T}_{\alpha\cdot\beta}^{\cdot\kappa\cdot}+\mathrm{T}_{\beta\cdot\alpha
}^{\cdot\kappa\cdot}-\mathrm{T}_{\cdot\alpha\beta}^{\kappa\cdot\cdot}\right)
\label{COT}%
\end{equation}
are the components of the contorsion tensor and $\mathrm{T}_{\cdot\alpha\beta
}^{\kappa\cdot\cdot}$ are the components of the torsion tensor of the
connection $\nabla$, we can write $\boldsymbol{D}\bullet\boldsymbol{\Theta}=0$
\emph{(}taking into account that $\mathbf{\bar{\Gamma}}_{\cdot\alpha\beta
}^{\kappa\cdot\cdot}=0$\emph{)} in components relative to orthonormal basis as%
\[
\boldsymbol{D}_{\alpha}\Theta_{\beta}^{\alpha}:=(\boldsymbol{D}%
_{\boldsymbol{e}_{\alpha}}\boldsymbol{\Theta})_{\beta}^{\alpha}=\boldsymbol{e}%
_{\alpha}(\Theta_{\beta}^{\alpha})+\mathrm{\triangle}_{\alpha\iota}^{\alpha
}\Theta_{\beta}^{\iota}-\mathrm{\triangle}_{\alpha\beta}^{\iota}\Theta_{\iota
}^{\alpha}=0.
\]
On the other hand since $\nabla_{\alpha}\Theta_{\beta}^{\alpha}:=(\nabla
_{\boldsymbol{e}_{\alpha}}\Theta)_{\beta}^{\alpha}=\boldsymbol{e}_{\alpha
}(\Theta_{\beta}^{\alpha})$ we have
\begin{equation}
\nabla_{\alpha}\Theta_{\beta}^{\alpha}=-\mathrm{\triangle}_{\alpha\iota
}^{\alpha}\Theta_{\beta}^{\iota}+\mathrm{\triangle}_{\alpha\beta}^{\iota
}\Theta_{\iota}^{\alpha}, \label{NEW}%
\end{equation}
which means that although $(\nabla\bullet\boldsymbol{\Theta})_{\beta}^{\alpha
}:=\nabla_{\alpha}\Theta_{\beta}^{\alpha}\neq0$ we can generate the conserved
currents $\boldsymbol{J}_{\mathbf{\Pi}_{\alpha}}$ in the teleparallel de
Sitter spacetime structure if \emph{Eq.(\ref{NEW})} is satisfied.
\end{remark}

\begin{remark}
To end this section and for completeness of the article it is necessary to
mention that in a remarkable paper \emph{(\cite{agp}) the gravitational
energy-momentum tensor \ in teleparallel gravity is discused in details. Also
related papers are (\cite{agp1,maluf}). A complete list of references can be
found in (\cite{ap2013})}
\end{remark}

\section{Equation of Motion for a Single-Pole Mass in a GRT Lorentzian
Spacetime}

In a classical paper Papapetrou derived the equations of motion of single pole
and spinning particles in \textbf{GRT}. Here we recall his derivation for the
case of a single pole-mass. We start recalling that in \textbf{GRT} the matter
fields is described by a energy-mometum tensor that satisfies the covariant
conservation law $\boldsymbol{D}\bullet\boldsymbol{T}=0$. If we introduce the
relative tensor
\begin{equation}
\mathfrak{T=}\boldsymbol{T\otimes\tau}_{\boldsymbol{g}}\in\sec T_{1}%
^{1}M\otimes%
%TCIMACRO{\tbigwedge \nolimits^{4}}%
%BeginExpansion
{\textstyle\bigwedge\nolimits^{4}}
%EndExpansion
T^{\ast}M \label{eq1}%
\end{equation}
we have recalling Appendix E and that $\boldsymbol{D}\bullet\boldsymbol{T}=0$
that%
\begin{equation}
\boldsymbol{D}_{\nu}\mathfrak{T}^{\mu\nu}+\Gamma_{\cdot\nu\alpha}^{\mu
\cdot\cdot}\mathfrak{T}^{\alpha\nu}=0. \label{eq2}%
\end{equation}
From Eq.(\ref{eq2}) we have%
\begin{equation}
\partial_{\nu}(x^{\alpha}\mathfrak{T}^{\mu\nu})=\mathfrak{T}^{\mu\alpha
}-x^{\alpha}\Gamma_{\cdot\nu\alpha}^{\mu\cdot\cdot}\mathfrak{T}^{\alpha\nu}.
\label{eq3}%
\end{equation}
To continue we suppose that a single-pole mass (considered as a probe
particle) is modelled by the restriction of the energy-mometum tensor
$\boldsymbol{T}$ inside a \textquotedblleft narrow\textquotedblright\ tube in
the Lorentzian spacetime representing a given gravitational field. Let us call
$\boldsymbol{T}_{0}$ that restriction and notice that $\boldsymbol{D}%
\bullet\boldsymbol{T}_{0}=0$. Inside the tube a timelike line $\gamma$ is
chosen to represent the particle motion. We restrict our analysis to
hyperbolic Lorentzian spacetimes for which a foliation $\mathbb{R\times
}\mathcal{S}$ ($\mathcal{S}$ a $3$-dimensional manifold) exists. We choose a
parametrization for $\gamma$ such that the its coordinates are $\boldsymbol{x}%
^{\mu}(\gamma(t))=\mathrm{X}^{\mu}(t)$ where $t=x^{0}$. The probe particle is
characterized by taking the coordinates of any point in the world tube to
satisfy
\begin{equation}
\mathrm{\delta}x^{\mu}=x^{\mu}-\mathrm{X}^{\mu}<<1. \label{eq4}%
\end{equation}
According to Papapetrou a single-pole particle is one for which the
integral\footnote{The integration is to be evaluated at a $t=const$ spaceline
hypersurface $\mathcal{S}$.}
\begin{equation}
\int\mathfrak{T}_{o}^{\mu\nu}dv\neq0, \label{eq5}%
\end{equation}
and all other integrals
\begin{equation}
\int\mathrm{\delta}x^{\alpha}\mathfrak{T}_{o}^{\mu\nu}dv,~~~\int%
\mathrm{\delta}x^{\alpha}\mathrm{\delta}x^{\beta}\mathfrak{T}_{o}^{\mu\nu
}dv,... \label{eq.6}%
\end{equation}
are null. We now evaluate%
\begin{equation}
\frac{d}{dt}\int\mathfrak{T}_{o}^{\mu0}=-\int\Gamma_{\cdot\nu\alpha}^{\mu
\cdot\cdot}\mathfrak{T}_{o}^{\alpha\nu}dv \label{eq.6a}%
\end{equation}
and%
\begin{equation}
\frac{d}{dt}\int x^{\alpha}\mathfrak{T}_{o}^{\mu0}=\int\mathfrak{T}_{o}%
^{\mu\alpha}dv-\int x^{\alpha}\Gamma_{\cdot\nu\alpha}^{\mu\cdot\cdot
}\mathfrak{T}_{o}^{\alpha\nu}dv. \label{eq.7}%
\end{equation}

Inside the world tube modelling the particle we can expand the connection
coefficients as%
\begin{equation}
\Gamma_{\cdot\nu\alpha}^{\mu\cdot\cdot}=\text{~}_{o}\Gamma_{\cdot\nu\alpha
}^{\mu\cdot\cdot}+~_{o}\Gamma_{\cdot\nu\alpha,\kappa}^{\mu\cdot\cdot
}\mathrm{\delta}x^{\kappa} \label{eq.8}%
\end{equation}
with ~$_{o}\Gamma_{\cdot\nu\alpha}^{\mu\cdot\cdot}$ the components of the
connection in the worldline $\gamma$. Then according to the definition of a
single-pole particle we get from Eq.(\ref{eq.6a}) and Eq.(\ref{eq.7}) along
$\gamma$:
\begin{equation}
\frac{d}{dt}\int\mathfrak{T}_{o}^{\mu0}+\text{~}_{o}\Gamma_{\cdot\nu\alpha
}^{\mu\cdot\cdot}\int\mathfrak{T}_{o}^{\alpha\nu}=0, \label{eq.9}%
\end{equation}%
\begin{equation}
\int\mathfrak{T}_{o}^{\mu\alpha}=\frac{d\mathrm{X}^{\mu}}{dt}\int%
\mathfrak{T}_{o}^{\nu0}dv. \label{eq.10}%
\end{equation}

Now, put
\begin{equation}
\gamma_{\ast}=\boldsymbol{u}:=\frac{d\mathrm{X}^{\mu}}{ds}e_{\mu}
\label{eq.11}%
\end{equation}
where $ds$ is proper time along $\gamma$ and define
\begin{equation}
M^{\mu\alpha}=u^{0}\int\mathfrak{T}_{o}^{\mu\alpha}dv. \label{eq.12}%
\end{equation}
Eq.(\ref{eq.9}) and Eq.(\ref{eq.10}) become%
\begin{equation}
\frac{d}{ds}\left(  \frac{M^{\mu0}}{u^{0}}\right)  +\text{~}_{o}\Gamma
_{\cdot\nu\alpha}^{\mu\cdot\cdot}M^{\alpha\nu}=0 \label{eq.9a}%
\end{equation}
and
\begin{equation}
M^{\mu\alpha}=u^{\mu}\frac{M^{\alpha0}}{u^{0}}. \label{eq.10a}%
\end{equation}
So, $M^{\mu0}=u^{\mu}\frac{M^{00}}{u^{0}}$ from where it follows that putting%
\begin{equation}
m:=\frac{M^{00}}{(u^{0})^{2}} \label{eq.13}%
\end{equation}
it is%
\begin{equation}
M^{\mu\alpha}=u^{\mu}\frac{M^{\alpha0}}{u^{0}}=mu^{\mu}u^{\alpha}
\label{eq.13a}%
\end{equation}
and we get%
\begin{equation}
\frac{d}{ds}(mu^{\mu})+~_{o}\Gamma_{\cdot\nu\alpha}^{\mu\cdot\cdot}mu^{\nu
}u^{\alpha}=0. \label{eq.14}%
\end{equation}

Now, the acceleration of the probe particle is $\boldsymbol{a}:=\boldsymbol{D}%
_{\boldsymbol{u}}\boldsymbol{u}$ and thus $\boldsymbol{g}(\boldsymbol{a,u}%
)=0$, i.e.,%
\begin{equation}
u_{\mu}\frac{d}{ds}u^{\mu}+~_{o}\Gamma_{\cdot\nu\alpha}^{\mu\cdot\cdot}%
mu^{\nu}u^{\alpha}u_{\mu}=0. \label{eq.15}%
\end{equation}

Multiplying Eq.(\ref{eq.14}) by $u_{\mu}$ and using Eq.(\ref{eq.15}) gives%
\begin{equation}
\frac{d}{ds}m=0 \label{eq.16}%
\end{equation}
and then Eq.(\ref{eq.14}) says that $\gamma$ is a geodesic of the Lorentzian
spacetime structure, i.e.,
\begin{equation}
\boldsymbol{D}_{\gamma_{\ast}}\gamma_{\ast}=0 \label{eq.17}%
\end{equation}
or
\[
\frac{du^{\mu}}{ds}+~_{o}\Gamma_{\cdot\nu\alpha}^{\mu\cdot\cdot}mu^{\nu
}u^{\alpha}=0.
\]

\section{Equation of Motion for a Single-Pole Mass in a de Sitter Lorentzian
Spacetime}

In this section we suppose that the arena where physical events take place is
the de Sitter spacetime structure $M^{dS\ell}$ where fields live and interact,
without never changing the metric $\boldsymbol{g}$, which we emphasize do not
represent any gravitational field here, i.e., we do not suppose here that
$M^{dS\ell}$ is a model of a gravitational field in \textbf{GRT}. As we
learned in Section 4 since de Sitter spacetime has one timelike and three
spacelike Killing vector fields $\mathbf{\Pi}_{\alpha}$ we can construct the
conserved currents $\mathcal{J}_{\mathbf{\Pi}_{\alpha}}=\Theta_{\alpha}^{\mu
}\boldsymbol{\vartheta}_{\mu}$ (see Eq.(\ref{CCURR})) from where we get
$\boldsymbol{D}\bullet\mathbf{\Theta}=0$ with $\mathbf{\Theta}=\mathcal{J}%
_{\mathbf{\Pi}_{\alpha}}\otimes e^{\alpha}=\Theta_{\alpha}^{\kappa
}\boldsymbol{\vartheta}_{\kappa}\otimes e^{\alpha}$.

Now, if we suppose that a probe free single-pole particle (i.e., one for which
its interaction with the remaining fields can be despised) is described by a
covariant conserved tensor $\mathbf{\Theta}$ in a narrow tube like the one
introduced in the previous section, we can derive (using analog notations for
$\gamma$, etc...) an equation\ like Eq.(\ref{eq.9}), i.e.,%

\begin{equation}
\frac{d}{dt}\int\Theta_{o}^{\mu0}\sqrt{-\det\boldsymbol{g}_{o}}dv+\text{~}%
_{o}\Gamma_{\cdot\nu\alpha}^{\mu\cdot\cdot}\int\Theta_{o}^{\alpha\nu}%
\sqrt{-\det\boldsymbol{g}_{o}}dv=0. \label{eq.180}%
\end{equation}
Now, we obtain an equation analogous to Eq.(\ref{eq.9a}) with $M^{\alpha\nu}$
substituted by $N^{\alpha\nu}$ i.e.,
\begin{equation}
\frac{d}{ds}\left(  \frac{N^{\mu0}}{u^{0}}\right)  +\text{~}_{o}\Gamma
_{\cdot\nu\alpha}^{\mu\cdot\cdot}N^{\alpha\nu}=0 \label{eq.18a}%
\end{equation}
with%
\begin{equation}
N^{\alpha\nu}:=u^{0}\int\Theta_{o}^{\mu\alpha}\sqrt{-\det\boldsymbol{g}_{o}%
}dv. \label{eq.18}%
\end{equation}

Putting this time%
\begin{equation}
m:=\frac{N^{00}}{(u^{0})^{2}} \label{eq.19}%
\end{equation}
we get $N^{\mu\alpha}=u^{\mu}\frac{N^{\alpha0}}{u^{0}}=mu^{\mu}u^{\alpha}$ and%

\begin{equation}
\frac{d}{ds}(mu^{\mu})+~_{o}\Gamma_{\cdot\nu\alpha}^{\mu\cdot\cdot}mu^{\nu
}u^{\alpha}=0, \label{eq.20}%
\end{equation}
from where we get\ exactly as in the previous section that $m=const$ and
$\boldsymbol{D}_{\gamma_{\ast}}\gamma_{\ast}=0$.

\begin{conclusion}
Papapetrou method applied to the $M^{dS\ell}$ structure gives for the motion
of a free single pole particle the geodesic equation $\boldsymbol{D}%
_{\gamma_{\ast}}\gamma_{\ast}=0$. Moreover, \emph{Eq.(\ref{eq.20})} is of
course, different from \emph{Eq.(\ref{k13f}) }contrary to conclusions of
authors of \emph{\ \cite{ps2012}}. It is also to be noted here that in
\emph{\cite{ps2012}} authors inferred correctly that Papapetrou method leads
to an equation that looks like to \emph{Eq.(\ref{eq.180}) }for a single-pole
particle moving in the $M^{dS\ell}$ structure. The equation that looks like
\emph{Eq.(\ref{eq.180})} in\emph{ \cite{ps2012}} is the \emph{Eq.(37)
}there\emph{, }but where in the place of \emph{ }$\Theta_{o}^{\alpha\nu}$ they
used $\Upsilon_{o}^{\alpha\nu}-\frac{1}{4\ell^{2}}\Upsilon_{o}^{\alpha\nu}$,
because they believe to be possible to use a local variation of the fields
that results in $\Lambda_{\nu}^{\mu}=0$.
\end{conclusion}

\appendix{}

\section{Clifford Bundle Formalism}

Let $(M,\boldsymbol{g},D,\tau_{\boldsymbol{g}},\uparrow)$ be an arbitrary
Lorentzian or Riemann-Cartan spacetime structure.\ The quadruple
$(M,\boldsymbol{g},\tau_{\boldsymbol{g}},\uparrow)$ denotes a four-dimensional
time-oriented and space-oriented Lorentzian manifold \cite{rc2007,sw}. This
means that $\boldsymbol{g}\in\sec T_{2}^{0}M$ is a Lorentzian metric of
signature $(1,3)$, $\tau_{\boldsymbol{g}}\in\sec\bigwedge{}^{4}T^{\ast}M$ and
$\uparrow$ is a time-orientation (see details, e.g., in \cite{sw}). Here,
$T^{\ast}M$ [$TM$] is the cotangent [tangent] bundle. $T^{\ast}M=\cup_{x\in
M}T_{x}^{\ast}M$, $TM=\cup_{x\in M}T_{x}M$, and $T_{x}M\simeq T_{x}^{\ast
}M\simeq\mathbb{R}^{1,3}$, where $\mathbb{R}^{1,3}$ is the Minkowski vector
space\footnote{Not to be confused with Minkowski spacetime \cite{rc2007,sw}.}.
$D$ is a metric compatible connection, i.e.\/, $D\boldsymbol{g}=0$. When
$D=\boldsymbol{D}$ is the Levi-Civita connection of $\boldsymbol{g}$ ,
$\mathbf{R}^{\boldsymbol{D}}\neq0$, and $\Theta^{\boldsymbol{D}}=0$,
$\mathbf{R}^{\boldsymbol{D}}$ and $\Theta^{\boldsymbol{D}}$ being respectively
the curvature and torsion tensors of the connection. $D=\boldsymbol{\nabla}$
is is a Riemann-Cartan connection , $\mathbf{R}^{\boldsymbol{\nabla}}\neq0$,
and $\Theta^{\boldsymbol{\nabla}}\neq0$ \footnote{Minkowski spacetime is the
particular case of a Lorentzian spacetime structure for which $M\simeq
\mathbb{R}^{4}$. and the curvature and torsion tensors of the Levi-Civita
connection of Minkowski metric are null. a teleparallel \ spacetime is a
particualr Riemann-Cartan spacetime such that $\mathbf{R}^{\boldsymbol{D}}=0$,
and $\Theta^{\boldsymbol{D}}\neq0$.} Let $\mathtt{g}\in\sec T_{0}^{2}M$ be the
metric of the \textit{cotangent bundle}. The Clifford bundle of differential
forms $\mathcal{C}\ell(M,\mathtt{g})$ is the bundle of algebras, i.e.,
$\mathcal{C}\ell(M,\mathtt{g})=\cup_{x\in M}\mathcal{C}\ell(T_{x}^{\ast
}M,\mathtt{g})$, where $\forall x\in M$, $\mathcal{C}\ell(T_{x}^{\ast
}M,\mathtt{g})=\mathbb{R}_{1,3}$, the so called \emph{spacetime} \emph{algebra
}\cite{rc2007}. Recall also that $\mathcal{C}\ell(M,\mathtt{g})$ is a vector
bundle associated to the \emph{orthonormal frame bundle}, i.e., we
have\footnote{\textrm{Ad}$:$~\textrm{Spin}$_{1,3}^{e}\rightarrow
Aut(\mathbb{R}_{1,3})$\textrm{ }by \textrm{Ad}$_{u}x=uxu^{-1}$.}
$\mathcal{C}\ell(M,\mathtt{g})$ $=P_{\mathrm{SO}_{(1,3)}^{e}}(M)\times
_{\mathrm{Ad}^{\prime}}\mathbb{R}_{1,3}$ \cite{lawmi,mr2004}.\ Also, when
$(M,\boldsymbol{g)}$ is a spin manifold we can show that\footnote{Take notice
that $\mathrm{Ad}:\mathrm{Spin}_{1,3}^{e}\rightarrow\mathrm{Aut}%
(\mathbb{R}_{1,3})~$such that for any $\mathcal{C\in}\mathbb{R}_{1,3}$ it is
$\mathrm{Ad}_{u}\mathcal{C=}u\mathcal{C}u^{-1}$. Since $\mathrm{Ad}%
_{-1}=\mathrm{Ad}_{1}=$ identity, $\mathrm{Ad}$ descends to a representation
of \textrm{SO}$_{1,3}^{e}$ that we denoted by $\mathrm{Ad}^{\prime}.$}
$\mathcal{C}\ell(M,\mathtt{g})$ $=P_{\mathrm{Spin}_{(1,3)}^{e}}(M)\times
_{\mathrm{Ad}}\mathbb{R}_{1,3}$\ For any $x\in M$, $\mathcal{C}\ell
(T_{x}^{\ast}M,\left.  \mathtt{g}\right\vert _{x})$ as a linear space over the
real field $\mathbb{R}$ is isomorphic to the Cartan algebra $\bigwedge
T_{x}^{\ast}M$ of the cotangent space. We have that $\bigwedge{}_{x}^{\ast
}M=\oplus_{k=0}^{4}\bigwedge^{k}T_{x}^{\ast}M$, where $\bigwedge^{k}%
T_{x}^{\ast}M$ is the $\binom{4}{k}$-dimensional space of $k$-forms. Then,
sections of $\mathcal{C}\ell(M,\mathtt{g})$ can be represented as a sum of non
homogeneous differential forms, that will be called Clifford (multiform)
fields. In the Clifford bundle formalism, of course, arbitrary basis can be
used, but in this short review of the main ideas of the Clifford calculus we
use mainly orthonormal basis. Let then $\{\boldsymbol{e}_{\mathbf{\alpha}}\}$
be an orthonormal basis for $TU\subset TM$, i.e., $\mathtt{g}(\boldsymbol{e}%
_{\mathbf{\alpha}},\boldsymbol{e}_{\mathbf{\beta}})=\eta_{\mathbf{\alpha\beta
}}=\mathrm{diag}(1,-1,-1,-1)$. Let $\boldsymbol{\gamma}^{\mathbf{\alpha}}%
\in\sec\bigwedge^{1}T^{\ast}M\hookrightarrow\sec\mathcal{C}\ell(M,\mathtt{g})$
($\mathbf{\alpha}=0,1,2,3$) be such that the set $\{\boldsymbol{\gamma
}^{\mathbf{\alpha}}\}$ is the dual basis of $\{\boldsymbol{e}_{\mathbf{\alpha
}}\}$. Also, $\{\boldsymbol{\gamma}_{\mathbf{\alpha}}\}$ is the reciprocal
basis of $\{\boldsymbol{\gamma}^{\mathbf{\alpha}}\}$, i.e., $\mathtt{g}%
$($\boldsymbol{\gamma}^{\mathbf{\alpha}},\boldsymbol{\gamma}_{\mathbf{\beta}%
})=\delta_{\mathbf{\beta}}^{\mathbf{\alpha}}$ and $\{\boldsymbol{e}%
^{\mathbf{\alpha}}\}$ is the reciprocal basis of $\{\boldsymbol{e}%
_{\mathbf{\alpha}}\}$, i.e., $\boldsymbol{g}(\boldsymbol{e}^{\mathbf{\alpha}%
},\boldsymbol{e}_{\mathbf{\beta}})=\delta_{\mathbf{\beta}}^{\mathbf{\alpha}}$.

\subsection{Clifford Product}

The fundamental \emph{Clifford product} (in what follows to be denoted by
juxtaposition of symbols) is generated by
\begin{equation}
\boldsymbol{\gamma}^{\mathbf{\alpha}}\boldsymbol{\gamma}^{\mathbf{\beta}%
}+\boldsymbol{\gamma}^{\mathbf{\beta}}\boldsymbol{\gamma}^{\mathbf{\alpha}%
}=2\eta^{\mathbf{\alpha\beta}} \label{cl}%
\end{equation}
and if $\mathcal{C}\in\sec\mathcal{C}\ell(M,\mathtt{g})$ we have%

\begin{equation}
\mathcal{C}=s+v_{\mathbf{\alpha}}\boldsymbol{\gamma}^{\mathbf{\alpha}}%
+\frac{1}{2!}f_{\mathbf{\alpha\beta}}\boldsymbol{\gamma}^{\mathbf{\alpha}%
}\boldsymbol{\gamma}^{\mathbf{\beta}}+\frac{1}{3!}t_{\mathbf{\alpha\beta
\gamma}}\boldsymbol{\gamma}^{\mathbf{\alpha}}\boldsymbol{\gamma}%
^{\mathbf{\beta}}\boldsymbol{\gamma}^{\mathbf{\gamma}}+p\boldsymbol{\gamma
}^{5}\;, \label{cl3}%
\end{equation}
where $\tau_{\boldsymbol{g}}=\boldsymbol{\gamma}^{5}=\boldsymbol{\gamma}%
^{0}\boldsymbol{\gamma}^{\mathbf{1}}\boldsymbol{\gamma}^{\mathbf{2}%
}\boldsymbol{\gamma}^{\mathbf{3}}$ is the volume element and $s$,
$v_{\mathbf{\alpha}}$, $f_{\mathbf{\beta}}$, $t_{\mathbf{\alpha\beta\gamma}}$,
$p\in\sec\bigwedge^{0}T^{\ast}M\hookrightarrow\sec\mathcal{C}\ell
(M,\mathtt{g})$.

For $A_{r}\in\sec\bigwedge^{r}T^{\ast}M\hookrightarrow\sec\mathcal{C}%
\ell(M,\mathtt{g}),B_{s}\in\sec\bigwedge^{s}T^{\ast}M\hookrightarrow
\sec\mathcal{C}\!\ell(M,\mathtt{g})$ we define the \emph{exterior product} in
$\mathcal{C}\ell(M,\mathtt{g})$\ ($\forall r,s=0,1,2,3)$ by
\begin{equation}
A_{r}\wedge B_{s}=\langle A_{r}B_{s}\rangle_{r+s}, \label{cl5}%
\end{equation}
where $\langle\;\;\rangle_{k}$ is the component in $\bigwedge^{k}T^{\ast}M$ of
the Clifford field. Of course, $A_{r}\wedge B_{s}=(-1)^{rs}B_{s}\wedge A_{r}$,
and the exterior product is extended by linearity to all sections of
$\mathcal{C}\ell(M,\mathtt{g})$.

Let $A_{r}\in\sec\bigwedge^{r}T^{\ast}M\hookrightarrow\sec\mathcal{C}%
\ell(M,\mathtt{g}),B_{s}\in\sec\bigwedge^{s}T^{\ast}M\hookrightarrow
\sec\mathcal{C}\ell(M,\mathtt{g})$. We define a \emph{scalar product
}in\emph{\ }$\mathcal{C}\!\ell(M,\mathtt{g})$ (denoted by $\cdot$) as follows:

(i) For $a,b\in\sec\bigwedge^{1}T^{\ast}M\hookrightarrow\sec\mathcal{C}%
\ell(M,\mathtt{g}),$%
\begin{equation}
a\cdot b=\frac{1}{2}(ab+ba)=\mathtt{g}(a,b). \label{cl4}%
\end{equation}

(ii) For $A_{r}=a_{1}\wedge\cdots\wedge a_{r},B_{r}=b_{1}\wedge\cdots.\wedge
b_{r}$, $a_{i},b_{j}\in\sec\bigwedge^{1}T^{\ast}M\hookrightarrow
\sec\mathcal{C}\!\ell(M,\mathtt{g})$, $i,j=1,...,r,$
\begin{align}
A_{r}\cdot B_{r}  &  =(a_{1}\wedge\cdots\wedge a_{r})\cdot(b_{1}\wedge
\cdots\wedge b_{r})\nonumber\\
&  =\left\vert
\begin{array}
[c]{lll}%
a_{1}\cdot b_{1} & .... & a_{1}\cdot b_{r}\\
.......... & .... & ..........\\
a_{r}\cdot b_{1} & .... & a_{r}\cdot b_{r}%
\end{array}
\right\vert . \label{CL6}%
\end{align}

We agree that if $r=s=0$, the scalar product is simply the ordinary product in
the real field.

Also, if $r\neq s$, then $A_{r}\cdot B_{s}=0$. Finally, the scalar product is
extended by linearity for all sections of $\mathcal{C}\!\ell(M,\mathtt{g})$.

For $r\leq s$, $A_{r}=a_{1}\wedge\cdots\wedge a_{r}$, $B_{s}=b_{1}\wedge
\cdots\wedge b_{s\text{ }}$, we define the \textit{left contraction}
$\lrcorner:(A_{r},B_{s})\mapsto A_{r}\underset{\boldsymbol{g}}{\lrcorner}%
B_{s}$ by
\begin{equation}
A_{r}\underset{\boldsymbol{g}}{\lrcorner}B_{s}=%
%TCIMACRO{\dsum \limits_{i_{1}\,<...\,<i_{r}}}%
%BeginExpansion
{\displaystyle\sum\limits_{i_{1}\,<...\,<i_{r}}}
%EndExpansion
\epsilon^{i_{1}...i_{s}}(a_{1}\wedge\cdots\wedge a_{r})\cdot(b_{_{i_{1}}%
}\wedge...\wedge b_{i_{r}})^{\sim}b_{i_{r}+1}\wedge\cdots\wedge b_{i_{s}}
\label{CL7}%
\end{equation}
where $\sim$ is the reverse mapping (\emph{reversion}) defined by
$\symbol{126}:\sec\mathcal{C}\ell(M,\mathtt{g})\rightarrow\sec\mathcal{C}%
\ell(M,\mathtt{g})$. For any $X=%
%TCIMACRO{\dbigoplus \nolimits_{p=0}^{4}}%
%BeginExpansion
{\displaystyle\bigoplus\nolimits_{p=0}^{4}}
%EndExpansion
X_{p}$, $X_{p}\in\sec%
%TCIMACRO{\dbigwedge \nolimits^{p}}%
%BeginExpansion
{\displaystyle\bigwedge\nolimits^{p}}
%EndExpansion
T^{\ast}M\hookrightarrow\sec\mathcal{C}\ell(M,\mathtt{g})$,
\begin{equation}
\tilde{X}=%
%TCIMACRO{\dsum \limits_{p=0}^{4}}%
%BeginExpansion
{\displaystyle\sum\limits_{p=0}^{4}}
%EndExpansion
\text{ }\tilde{X}_{p}=%
%TCIMACRO{\dsum \limits_{p=0}^{4}}%
%BeginExpansion
{\displaystyle\sum\limits_{p=0}^{4}}
%EndExpansion
(-1)^{\frac{1}{2}k(k-1)}X_{p}. \label{cl8}%
\end{equation}
We agree that for $\alpha,\beta\in\sec\bigwedge^{0}T^{\ast}M$ the contraction
is the ordinary (pointwise) product in the real field and that if $\alpha
\in\sec\bigwedge^{0}T^{\ast}M$, $X_{r}\in\sec\bigwedge^{r}T^{\ast}M,Y_{s}%
\in\sec\bigwedge^{s}T^{\ast}M\hookrightarrow\sec\mathcal{C}\ell(M,\mathtt{g})$
then $(\alpha X_{r})\lrcorner B_{s}=X_{r}\lrcorner(\alpha Y_{s})$. Left
contraction is extended by linearity to all pairs of sections of
$\mathcal{C}\ell(M,\mathtt{g})$, i.e., for $X,Y\in\sec\mathcal{C}%
\ell(M,\mathtt{g})$%
\begin{equation}
X\underset{\boldsymbol{g}}{\lrcorner}Y=\sum_{r,s}\langle X\rangle
_{r}\underset{\boldsymbol{g}}{\lrcorner}\langle Y\rangle_{s},~~~~r\leq s.
\label{cl9}%
\end{equation}

It is also necessary to introduce the operator of \emph{right contraction}
denoted by $\underset{\boldsymbol{g}}{\llcorner}$ . The definition is obtained
from the one presenting the left contraction with the imposition that $r\geq
s$ and taking into account that now if $A_{r}\in\sec\bigwedge^{r}T^{\ast}M,$
$B_{s}\in\sec\bigwedge^{s}T^{\ast}M$ then $A_{r}\underset{\boldsymbol{g}%
}{\llcorner}(\alpha B_{s})=(\alpha A_{r})\underset{\boldsymbol{g}}{\llcorner
}B_{s}$. See also the third formula in Eq.(\ref{cl10}).

The main formulas used in this paper can be obtained from the following ones
\begin{align}
a\mathcal{B}_{s}  &  =a\lrcorner\mathcal{B}_{s}+a\wedge\mathcal{B}%
_{s},\;\;\mathcal{B}_{s}a=\mathcal{B}_{s}\underset{\boldsymbol{g}}{\llcorner
}a+\mathcal{B}_{s}\wedge a,\nonumber\\
a\underset{\boldsymbol{g}}{\lrcorner}\mathcal{B}_{s}  &  =\frac{1}%
{2}(a\mathcal{B}_{s}-(-1)^{s}\mathcal{B}_{s}a),\nonumber\\
\mathcal{A}_{r}\underset{\boldsymbol{g}}{\lrcorner}\mathcal{B}_{s}  &
=(-1)^{r(s-r)}\mathcal{B}_{s}\underset{\boldsymbol{g}}{\llcorner}%
\mathcal{A}_{r},\nonumber\\
a\wedge\mathcal{B}_{s}  &  =\frac{1}{2}(a\mathcal{B}_{s}+(-1)^{s}%
\mathcal{B}_{s}a),\nonumber\\
\mathcal{A}_{r}\mathcal{B}_{s}  &  =\langle\mathcal{A}_{r}\mathcal{B}%
_{s}\rangle_{|r-s|}+\langle\mathcal{A}_{r}\mathcal{B}_{s}\rangle
_{|r-s|+2}+...+\langle\mathcal{A}_{r}\mathcal{B}_{s}\rangle_{|r+s|}\nonumber\\
&  =\sum\limits_{k=0}^{m}\langle\mathcal{A}_{r}\mathcal{B}_{s}\rangle
_{|r-s|+2k}\text{ }\nonumber\\
\mathcal{A}_{r}\cdot\mathcal{B}_{r}  &  =\mathcal{B}_{r}\cdot\mathcal{A}%
_{r}=\widetilde{\mathcal{A}}_{r}\underset{\boldsymbol{g}}{\lrcorner
}\mathcal{B}_{r}=\mathcal{A}_{r}\underset{\boldsymbol{g}}{\llcorner
}\widetilde{\mathcal{B}}_{r}=\langle\widetilde{\mathcal{A}}_{r}\mathcal{B}%
_{r}\rangle_{0}=\langle\mathcal{A}_{r}\widetilde{\mathcal{B}}_{r}\rangle_{0}.
\label{cl10}%
\end{align}
Two other important identities used in the main text are:%

\begin{align}
a\underset{\boldsymbol{g}}{\lrcorner}(\mathcal{X}\wedge\mathcal{Y})  &
=(a\underset{\boldsymbol{g}}{\lrcorner}\mathcal{X})\wedge\mathcal{Y}%
+\mathcal{\hat{X}}\wedge(a\underset{\boldsymbol{g}}{\lrcorner}\mathcal{Y}%
),\label{T54}\\
\mathcal{X}\underset{\boldsymbol{g}}{\lrcorner}(\mathcal{Y}%
\underset{\boldsymbol{g}}{\lrcorner}\mathcal{Z})  &  =(\mathcal{X}%
\wedge\mathcal{Y})\underset{\boldsymbol{g}}{\lrcorner}\mathcal{Z}, \label{T55}%
\end{align}
for any $a\in\sec%
%TCIMACRO{\dbigwedge \nolimits^{1}}%
%BeginExpansion
{\displaystyle\bigwedge\nolimits^{1}}
%EndExpansion
T^{\ast}M\hookrightarrow\mathcal{C}\ell(M,\mathtt{g})$ and $\mathcal{X}%
,\mathcal{Y},\mathcal{Z}\in\sec%
%TCIMACRO{\dbigwedge }%
%BeginExpansion
{\displaystyle\bigwedge}
%EndExpansion
T^{\ast}M\hookrightarrow\mathcal{C}\ell(M,\mathtt{g})$.

\subsubsection{Hodge Star Operator}

Let $\underset{\boldsymbol{g}}{\star}$ be the Hodge star operator, i.e., the
mapping $\underset{\boldsymbol{g}}{\star}:%
%TCIMACRO{\dbigwedge \nolimits^{k}}%
%BeginExpansion
{\displaystyle\bigwedge\nolimits^{k}}
%EndExpansion
T^{\ast}M\rightarrow%
%TCIMACRO{\dbigwedge \nolimits^{4-k}}%
%BeginExpansion
{\displaystyle\bigwedge\nolimits^{4-k}}
%EndExpansion
T^{\ast}M,$ $A_{k}\mapsto\underset{\boldsymbol{g}}{\star}A_{k}$. For $A_{k}%
\in\sec\bigwedge^{k}T^{\ast}M\hookrightarrow\sec\mathcal{C}\!\ell
(M,\mathtt{g})$ we have
\begin{equation}
\lbrack B_{k}\cdot A_{k}]\tau_{\mathtt{g}}=B_{k}\wedge\underset{\boldsymbol{g}%
}{\star}A_{k},\forall B_{k}\in\sec\bigwedge\nolimits^{k}T^{\ast}%
M\hookrightarrow\sec\mathcal{C}\ell(M,\mathtt{g}). \label{11a}%
\end{equation}
where $\tau_{\mathtt{g}}=\theta^{\mathbf{5}}\in\sec\bigwedge^{4}T^{\ast
}M\hookrightarrow\sec\mathcal{C}\!\ell(M,\mathtt{g})$ is a \emph{standard}
volume element. We have,
\begin{equation}
\underset{\boldsymbol{g}}{\star}A_{k}=\widetilde{A}_{k}\tau_{\mathtt{g}%
}=\widetilde{A}_{k}\underset{\boldsymbol{g}}{\lrcorner}\tau_{\mathtt{g}}.
\label{11b}%
\end{equation}
where as noted before, in this paper $\widetilde{\mathcal{A}}_{k}$ denotes the
\textit{reverse} of $\mathcal{A}_{k}$. Eq.(\ref{11b}) permits calculation of
Hodge duals very easily in an orthonormal basis for which $\tau_{\mathtt{g}%
}=\boldsymbol{\gamma}^{\mathbf{5}}$. Let $\{\vartheta^{\alpha}\}$ be the dual
basis of $\{e_{\alpha}\}$ (i.e., it is a basis for $T^{\ast}U\equiv
\bigwedge\nolimits^{1}T^{\ast}U$) which is either an \textit{orthonormal} or a
\textit{coordinate basis}. Then writing \texttt{g}$(\vartheta^{\alpha
},\vartheta^{\beta})=g^{\alpha\beta}$, with $g^{\alpha\beta}g_{\alpha\rho
}=\delta_{\rho}^{\beta}$, and $\vartheta^{\mu_{1}...\mu_{p}}=\vartheta
^{\mu_{1}}\wedge\cdots\wedge\vartheta^{\mu_{p}}$, $\vartheta^{\nu_{p+1}%
...\nu_{n}}=\vartheta^{\nu_{p+1}}\wedge\cdots\wedge\vartheta^{\nu_{n}}$ we
have from Eq.(\ref{11b})
\begin{equation}
\underset{\boldsymbol{g}}{\star}\vartheta^{\mu_{1}...\mu_{p}}=\frac{1}%
{(n-p)!}\sqrt{\left\vert \det\boldsymbol{g}\right\vert }g^{\mu_{1}\nu_{1}%
}\cdots g^{\mu_{p}\nu_{p}}\epsilon_{\nu_{1}...\nu_{n}}\vartheta^{\nu
_{p+1}...\nu_{n}}. \label{hodge dual}%
\end{equation}
where $\det\boldsymbol{g}$ denotes the determinant of the matrix with entries
$g_{\alpha\beta}=\boldsymbol{g}(e_{\alpha},e_{\beta})$, i.e., $\det
\boldsymbol{g}=\det[g_{\alpha\beta}].$ We also define the inverse
$\underset{\boldsymbol{g}}{\star}^{-1}$ of the Hodge dual operator, such that
\ $\underset{\boldsymbol{g}}{\star}^{-1}\underset{\boldsymbol{g}}{\star
}=\underset{\boldsymbol{g}}{\star}\underset{\boldsymbol{g}}{\star}^{-1}=1$. It
is given by:
\begin{align}
\underset{\boldsymbol{g}}{\star}^{-1}  &  :\sec%
%TCIMACRO{\dbigwedge \nolimits^{r}}%
%BeginExpansion
{\displaystyle\bigwedge\nolimits^{r}}
%EndExpansion
T^{\ast}M\rightarrow\sec%
%TCIMACRO{\dbigwedge \nolimits^{n-r}}%
%BeginExpansion
{\displaystyle\bigwedge\nolimits^{n-r}}
%EndExpansion
T^{\ast}M,\nonumber\\
\underset{\boldsymbol{g}}{\star}^{-1}A_{r}  &  =(-1)^{r(n-r)}\mathrm{sgn}%
\det\boldsymbol{g}\underset{\boldsymbol{g}}{\star}A_{r}, \label{h1}%
\end{align}
where \textrm{sgn }$\det\boldsymbol{g}=\det\boldsymbol{g}/|\det\boldsymbol{g}%
|$ denotes the sign of the determinant of $\boldsymbol{g}$.

Some useful identities (used in the text) involving the Hodge star operator,
the exterior product and contractions are:%

\begin{equation}%
\begin{array}
[c]{l}%
A_{r}\wedge\underset{\boldsymbol{g}}{\star}B_{s}=B_{s}\wedge
\underset{\boldsymbol{g}}{\star}A_{r};\quad r=s\\
A_{r}\cdot\underset{\boldsymbol{g}}{\star}B_{s}=B_{s}\cdot
\underset{\boldsymbol{g}}{\star}A_{r};\quad r+s=n\\
A_{r}\wedge\underset{\boldsymbol{g}}{\star}B_{s}=(-1)^{r(s-1)}%
\underset{\boldsymbol{g}}{\star}(\tilde{A}_{r}\underset{\boldsymbol{g}%
}{\lrcorner}B_{s});\quad r\leq s\\
A_{r}\underset{\boldsymbol{g}}{\lrcorner}\underset{\boldsymbol{g}}{\star}%
B_{s}=(-1)^{rs}\underset{\boldsymbol{g}}{\star}(\tilde{A}_{r}\wedge
B_{s});\quad r+s\leq n\\
\underset{\boldsymbol{g}}{\star}\tau_{g}=\mathrm{sign}\text{ }\mathbf{g}%
;~~~~~\underset{\boldsymbol{g}}{\star}1=\tau_{g}.
\end{array}
\label{440new}%
\end{equation}

\subsubsection{Dirac Operator Associated to a Levi-Civita Connection
$\boldsymbol{D}$}

Let $d$ and $\delta$ be respectively the differential and Hodge codifferential
operators acting on sections of $\mathcal{C}\!\ell(M,\mathtt{g})$. If
$A_{p}\in\sec\bigwedge^{p}T^{\ast}M\hookrightarrow\sec\mathcal{C}%
\!\ell(M,\mathtt{g})$, then $\underset{\boldsymbol{g}}{\delta}A_{p}%
=(-1)^{p}\underset{\boldsymbol{g}}{\star}^{-1}d\underset{\boldsymbol{g}%
}{\star}A_{p}$.

The Dirac operator acting on sections of $\mathcal{C}\!\ell(M,\mathtt{g})$
associated with the metric compatible connection $D$ is the invariant first
order differential operator
\begin{equation}
\boldsymbol{\partial}=\vartheta^{\alpha}\boldsymbol{D}_{e_{\alpha}},
\label{12}%
\end{equation}
where $\{e_{\alpha}\}$ is an arbitrary (coordinate or orthonormal)\emph{
basis} for $TU\subset TM$ and $\{\vartheta^{\mathbf{\alpha}}\}$ is a basis for
$T^{\ast}U\subset T^{\ast}M$ dual to the basis $\{e_{\mathbf{\alpha}}\}$,
i.e., $\vartheta^{\beta}(e_{\alpha})=\delta_{\beta}^{\alpha}$, $\alpha
,\beta=0,1,2,3$. The reciprocal basis of $\{\vartheta^{\alpha}\}$ is denoted
$\{\vartheta_{\alpha}\}$ and we have $\vartheta_{\alpha}\cdot\vartheta_{\beta
}=g_{\alpha\beta}$. Also, when $\{e_{\alpha}=\partial_{\alpha}\}$ and
$\{\vartheta^{\alpha}=dx^{\alpha}\}$ we have%
\begin{equation}
\boldsymbol{D}_{\partial_{\mathbf{\alpha}}}\partial_{\beta}=\mathbf{\Gamma
}_{\cdot\alpha\beta}^{\varepsilon\cdot\cdot}\partial_{\mathbf{\beta}%
},~~~~\boldsymbol{D}_{\partial_{\mathbf{\alpha}}}dx^{\mathbf{\beta}%
}=-\mathbf{\Gamma}_{\cdot\alpha\mathbf{\varepsilon}}^{\beta\cdot\cdot
}dx^{\varepsilon}%
\end{equation}
and when $\{e_{\alpha}=\boldsymbol{e}_{\alpha}\}$, $\{\vartheta^{\alpha
}=\boldsymbol{\gamma}^{\alpha}\}$ are orthonormal basis we have
\begin{equation}
\boldsymbol{D}_{e_{\mathbf{\alpha}}}e_{\mathbf{\beta}}=\omega_{\cdot
\alpha\mathbf{\beta}}^{\mathbf{\lambda\cdot\cdot}}e_{\mathbf{\lambda}%
}~~~~,\boldsymbol{D}_{e_{\mathbf{\alpha}}}\boldsymbol{\gamma}^{\mathbf{\beta}%
}=-\omega_{\cdot\alpha\mathbf{\lambda}}^{\mathbf{\beta\cdot\cdot}%
}\boldsymbol{\gamma}^{\lambda} \label{12n}%
\end{equation}
We define the connection\footnote{Also called \textquotedblleft spin
connection 1-forms\textquotedblright.} $1$-forms in the gauge defined by
$\{\boldsymbol{\gamma}^{\mathbf{\alpha}}\}$ as%
\begin{equation}
\omega_{\cdot\mathbf{\beta}}^{\mathbf{\alpha\cdot}}:=\omega_{\cdot
\mathbf{\lambda\beta}}^{\mathbf{\alpha\cdot\cdot}}\boldsymbol{\gamma}%
^{\lambda}. \label{12na}%
\end{equation}
Moreover, we write for an arbitrary tensor field $Y=Y_{\nu_{1}...\nu_{s}}%
^{\mu_{1}...\mu_{r}}\vartheta^{\nu_{1}}\otimes...\otimes\vartheta^{\nu_{s}%
}\otimes\partial_{\mu_{1}}\otimes...\otimes\partial_{\mu_{r}}$ in a coordinate
basis%
\begin{equation}
\boldsymbol{D}_{\mathbf{e}_{\mathbf{\alpha}}}Y:=(\boldsymbol{D}_{\alpha}%
Y_{\nu_{1}...\nu_{s}}^{\mu_{1}...\mu_{r}})\vartheta^{\nu_{1}}\otimes
...\otimes\vartheta^{\nu_{s}}\otimes\partial_{\mu_{1}}\otimes...\otimes
\partial_{\mu_{r}} \label{cd}%
\end{equation}
and also when we write $Y=Y_{\nu_{1}...\nu_{s}}^{\mu_{1}...\mu_{r}%
}\boldsymbol{\gamma}^{\nu_{1}}\otimes...\otimes\boldsymbol{\gamma}^{\nu_{s}%
}\otimes\boldsymbol{e}_{\mu_{1}}\otimes...\otimes\boldsymbol{e}_{\mu_{r}}$ we
also write%
\begin{equation}
\boldsymbol{D}_{\mathbf{e}_{\mathbf{\alpha}}}Y:=(\boldsymbol{D}_{\alpha}%
Y_{\nu_{1}...\nu_{s}}^{\mu_{1}...\mu_{r}})\boldsymbol{\gamma}^{\nu_{1}}%
\otimes...\otimes\boldsymbol{\gamma}^{\nu_{s}}\otimes\boldsymbol{e}_{\mu_{1}%
}\otimes...\otimes\boldsymbol{e}_{\mu_{r}}%
\end{equation}
so please pay attention when reading a particular formula to certificate\ the
meaning of $\boldsymbol{D}_{\alpha}Y_{\nu_{1}...\nu_{s}}^{\mu_{1}...\mu_{r}}$,
i.e., if we are using in that formula coordinate or orthonormal frames.

We have also the important results (see, e.g., \cite{rc2007}) for the Dirac
operator associated with the Levi-Civita connection $\boldsymbol{D}$ acting on
the sections of the Clifford bundle\footnote{For a general metric compatible
Riemann-Cartan connection the formula in Eq.(\ref{cl13b}) is not valid, we
have a more general relation involving the torsion tensor that will not be
used in this paper. The interested reader may consult \cite{rc2007}.}%

\begin{subequations}
\begin{align}
\boldsymbol{\partial}A_{p}  &  =\boldsymbol{\partial}\wedge A_{p\,}%
+\,\boldsymbol{\partial}\lrcorner A_{p}=dA_{p}-\underset{\boldsymbol{g}%
}{\delta}A_{p},\label{cl13a}\\
\boldsymbol{\partial}\wedge A_{p}  &  =dA_{p},~~~~~\,\boldsymbol{\partial
}\underset{\boldsymbol{g}}{\lrcorner}A_{p}=-\underset{\boldsymbol{g}}{\delta
}A_{p}. \label{cl13b}%
\end{align}
We shall need the following identity valid for any $a,b\in\sec\bigwedge
^{1}T^{\ast}M\hookrightarrow\mathcal{C}\ell(M,\mathtt{g}),$%
\end{subequations}
\begin{equation}
\boldsymbol{\partial}{(}a\cdot b)={(a\cdot}\boldsymbol{\partial}%
)b+(b\cdot\boldsymbol{\partial}{)}a+a\lrcorner(\boldsymbol{\partial}{\wedge
}b)+b\lrcorner(\boldsymbol{\partial}{\wedge}a). \label{14A}%
\end{equation}

\subsection{Covariant D' Alembertian, Hodge D'Alembertian and Ricci Operators}

The square of the Dirac operator $\Diamond=\boldsymbol{\partial}^{2}$ is
called Hodge D'Alembertian and we have the following noticeable formulas:%
\begin{equation}
\boldsymbol{\partial}^{2}=-d\underset{\boldsymbol{g}}{\delta}%
-\underset{\boldsymbol{g}}{\delta}d, \label{13bis}%
\end{equation}
and
\begin{equation}
\boldsymbol{\partial}^{2}A_{p}=\boldsymbol{\partial}{\cdot\boldsymbol{\partial
~}}A_{p}+\boldsymbol{\partial}{\wedge\boldsymbol{\partial~}}A_{p} \label{14}%
\end{equation}
where $\boldsymbol{\partial}{\cdot\boldsymbol{\partial}}$ is called the
\textit{covariant D'Alembertian} and $\boldsymbol{\partial}{\wedge
\boldsymbol{\partial}}$ is called the Ricci operator\footnote{For more details
concerning the square of Dirac (and spin-Dirac operators) on a general
Riemann-Cartan spacetime, see \cite{nra}.} If $A_{p}=\frac{1}{p!}A_{\mu
_{1}...\mu_{p}}\vartheta^{\mu_{_{1}}}\wedge\cdots\wedge\vartheta^{\mu_{p}}$,
we have
\begin{equation}
\boldsymbol{\partial}{\cdot}\boldsymbol{\partial~}A_{p}=g^{\alpha\beta
}(D_{\partial_{\alpha}}D_{\partial_{\beta}}-\Gamma_{\cdot\alpha\beta}%
^{\rho\cdot\cdot}D_{\partial_{\rho}})A_{p}=\frac{1}{p!}g^{\alpha\beta
}D_{\alpha}D_{\beta}A_{\alpha_{1}\ldots\alpha_{p}}\vartheta^{\alpha_{1}}%
\wedge\cdots\wedge\vartheta^{\alpha_{p}}, \label{15}%
\end{equation}
Also for $\boldsymbol{\partial}{\wedge}\boldsymbol{\partial}$ in an arbitrary
basis (coordinate or orthonormal)%
\begin{equation}
\boldsymbol{\partial}{\wedge}\boldsymbol{\partial~}A_{p}=\frac{1}{2}%
\vartheta^{\alpha}\wedge\vartheta^{\beta}([\boldsymbol{D}_{e_{\alpha}%
},\boldsymbol{D}_{e_{\beta}}]-(\Gamma_{\cdot\alpha\beta}^{\rho\cdot\cdot
}-\Gamma_{\cdot\beta\alpha}^{\rho\cdot\cdot})D_{e_{\rho}})A_{p}. \label{16}%
\end{equation}

In particular we have \cite{rc2007}%
\begin{equation}
\boldsymbol{\partial}{\wedge}\boldsymbol{\partial~}\vartheta^{\mu}%
=\mathcal{R}^{\mu}, \label{166}%
\end{equation}
where $\mathcal{R}^{\mu}=R_{\nu}^{\mu}\vartheta^{\nu}\in\sec\bigwedge
^{1}T^{\ast}M\hookrightarrow\sec\mathcal{C}\ell(M,\mathtt{g})$ are the Ricci
$1$-form fields, such that if $R_{\cdot\nu\sigma\mu}^{\mu\cdot\cdot\cdot}$ are
the components of the Riemann tensor we use the convention that $R_{\nu\sigma
}=R_{\cdot\nu\sigma\mu}^{\mu\cdot\cdot\cdot}$ are the components of the Ricci tensor.

Applying this operator to the 1-forms of the a $1$-form of the basis
$\{\vartheta^{\mu}\}$, we get:
\begin{equation}
\boldsymbol{\partial}\wedge\boldsymbol{\partial~}\vartheta^{\mu}=-\frac{1}%
{2}\text{ }R_{\cdot\rho\alpha\beta}^{\mu\cdot\cdot\cdot}(\vartheta^{\alpha
}\wedge\vartheta^{\beta})\vartheta^{\rho}=\mathcal{R}_{\rho}^{\mu}%
\vartheta_{\rho}.
\end{equation}

$\boldsymbol{\partial}\wedge\boldsymbol{\partial}$ is an extensor operator,
i.e., for $A\in\sec\bigwedge^{1}T^{\ast}M\hookrightarrow\sec\mathcal{C}%
\ell(M,\mathtt{g})$ it is%
\begin{equation}
\boldsymbol{\partial}\wedge\boldsymbol{\partial~}A=A_{\mu}\boldsymbol{\partial
}\wedge\boldsymbol{\partial}\vartheta^{\mu}.
\end{equation}

\begin{remark}
We remark that covariant Dirac spinor fields used in almost all Physics texts
books and research papers can be represented as certain equivalence classes of
even sections of the Clifford bundle $\mathcal{C}\ell(M,\mathtt{g})$. These
objects are now called Dirac-Hestenes spinor fields \emph{(}\textbf{DHSF}%
\emph{)} and a thoughtful theory describing them can be found in
\emph{\cite{r2004,mr2004,rc2007}}. Moreover, in \emph{\cite{lrw2015}} using
the concept of \textbf{DHSF} a new approach is given to the concept of Lie
derivative for spinor fields, which does not seems to have the objections of
previous approaches to the subject. Of course, a meaningful definition of Lie
derivative for spinor fields is a necessary condition for a formulation of
conservation laws involving bosons and fermion fields in interaction in
arbitrary manifolds. We will present the complete Lagrangian density involving
the gravitation field (interpreted as fields in the Faraday sense and
described by cotetrad fields\emph{)}, the electromagnetic and the
\textbf{DHSF} \ living in a parallelizable manifold and its variation in
another publication.
\end{remark}

\section{Lie Derivatives and Variations}

In modern field theory the physical fields are tensor and spinor fields living
on a structure $(M,\boldsymbol{g},\tau_{\boldsymbol{g}},\uparrow)$ and
interacting among themselves. Note that at this point we did not introduce any
connection in our game, since according to our view (see, e.g., Chapter 11 of
\cite{rc2007}) the introduction of a particular connection to describe Physics
is only a question of convenience. For the objective of this paper we shall
consider two structures (already introduced in the main text), a Lorentzian
spacetime $M^{dS\ell}=(M,\boldsymbol{g},\boldsymbol{D},\tau_{\boldsymbol{g}%
},\uparrow)$ where $\boldsymbol{D}$ is the Levi-Civita connection of
$\boldsymbol{g}$ and a teleparallel spacetime $M^{dSTP}=(M,\boldsymbol{g}%
,\nabla,\tau_{\boldsymbol{g}},\uparrow)$ where $\nabla$ is a metric compatible
teleparallel connection. Minkowski spacetime structure will be denoted by
$(M,\boldsymbol{\eta},\mathrm{D},\tau_{\boldsymbol{\eta}},\uparrow)$. The
equations of motion are derived from a variational principle once a given
Lagrangian density is postulated for the interacting fields of the theory.

As well known, diffeomorphism invariance is a crucial ingredient of any
physical theory. This means that if a physical phenomenon is described by
fields, say, $\phi_{1},....,\phi_{n}$ (defined in $\mathcal{U}\subset M$)
satisfying equations of motion of the theory (with appropriated initial and
boundary conditions)\ then if $\mathrm{h}:M\mapsto M$ is a diffeomorphism then
the fields $\mathrm{h}^{\ast}\phi_{1},....,\mathrm{h}^{\ast}\phi_{N}$
(where$,\mathrm{h}^{\ast}$ is the pullback mapping) describe the same physical
phenomenon\footnote{Of course, the fields $\mathrm{h}^{\ast}\phi
_{1},....,\mathrm{h}^{\ast}\phi_{N}$ must satisfy deformed initial conditions
and deformed boundary conditions.} in $\mathrm{h}\mathcal{U}$.

Suppose that fields $\phi_{1},....,$(in what follows called simply matter
fields\footnote{In truth, by matter fields we understand fields of two kinds,
fermion fields (electrons, neutrinos, quarks) and boson fields
(electromagnetic, gravitational, weak and strong fields).}) are arbitrary
differential forms. Their Lagrangian density will here be defined as the
functional mapping\footnote{A rigorous formulation needs the introduction of
jet bundles (see, e.g., \cite{F1}). We will not need such sophistication for
the goals of this paper.}
\begin{equation}
\mathcal{L}_{m}:(\phi_{1},...,\phi_{N},d\phi_{1},....,d\phi_{N})\mapsto
\mathcal{L}_{m}(\boldsymbol{\phi},d\boldsymbol{\phi})\in\sec%
%TCIMACRO{\tbigwedge \nolimits^{4}}%
%BeginExpansion
{\textstyle\bigwedge\nolimits^{4}}
%EndExpansion
TM \label{A1}%
\end{equation}
where $\mathcal{L}_{m}(\phi,d\phi)$ is here supposed to be constructed using
the Hodge star operator $\underset{\boldsymbol{g}}{\star}$. The action of the
system is%
\begin{equation}
\mathcal{A}=\int\nolimits_{\mathcal{U}}\mathcal{L}_{m}(\boldsymbol{\phi
},d\boldsymbol{\phi}). \label{A2}%
\end{equation}

Choose a chart of $M$ covering $\mathcal{U}$ and $\mathrm{h}\mathcal{U}$ with
coordinate functions $\{\boldsymbol{x}^{\mu}\}$. Then under an infinitesimal
mapping $\mathrm{h}_{\varepsilon}:M\mapsto M$ , $x\mapsto x^{\prime
}=\mathrm{h}_{\varepsilon}(x)$ generated by a one parameter group of
diffeomorphisms associated to the vector field $\boldsymbol{\xi}\in\sec TM$ we
have (with $\mathrm{h}_{\varepsilon}^{\mu}$ the coordinate representative of
the mapping \textrm{h}$_{\varepsilon}$)
\begin{equation}
x^{\mu}=\boldsymbol{x}^{\mu}(x)\mapsto x^{\prime\mu}=\boldsymbol{x}^{\mu
}(\mathrm{h}_{\varepsilon}(x))=\mathrm{h}_{\varepsilon}^{\mu}(x^{\alpha
})=x^{\mu}+\varepsilon\xi^{\mu},~~~\left\vert \varepsilon\right\vert <<1
\label{A3}%
\end{equation}
In Physics textbooks given an infinitesimal diffeormorphism \textrm{h}%
$_{\varepsilon}$\textrm{ }several different kinds of \emph{variations }(for
each one of the fields $\phi_{i}$)\textrm{ }are defined.

Let $\phi$ be one of the fields\ $\phi_{1},...,\phi_{N}$ and recall that the
Lie derivative of $\phi$ in the direction of the vector field $\mathbf{\xi}$
is given by
\begin{equation}
\pounds _{\boldsymbol{\xi}}\phi=\lim_{\varepsilon\rightarrow0}\frac
{\mathrm{h}_{\varepsilon}^{\ast}\circ\phi\circ\mathrm{h}_{\varepsilon}-\phi
}{\varepsilon} \label{A5}%
\end{equation}
As an example, take $\phi$ as $1$-form. Then, in the chart introduced above
using the definition of the pullback%
\begin{equation}
\mathrm{h}_{\varepsilon}^{\ast}\phi_{\mu}~(x^{\kappa}):=[\mathrm{h}%
_{\varepsilon}^{\ast}(\phi(\mathrm{h}_{\varepsilon}(x)))]_{\mu} \label{A6}%
\end{equation}
it is
\begin{equation}
\mathrm{h}_{\varepsilon}^{\ast}\phi~(x)=\mathrm{h}_{\varepsilon}^{\ast}%
\phi_{\mu}(x^{\kappa})dx^{\mu}:=\phi_{\kappa}(x^{\prime\kappa}(x^{\kappa
}))\frac{\partial x^{\prime\kappa}}{\partial x^{\mu}}dx^{\mu}. \label{A7}%
\end{equation}
Then, Eq.(\ref{A5}) can be written in components as%

\begin{equation}
\left(  \pounds _{\boldsymbol{\xi}}\phi(x^{\kappa})\right)  _{\mu}%
=\lim_{\varepsilon\rightarrow0}\frac{\mathrm{h}_{\varepsilon}^{\ast}\phi_{\mu
}~(x^{\kappa})-\phi_{\mu}(x^{\kappa})}{\varepsilon} \label{A8}%
\end{equation}

Now, to first order in $\varepsilon$ we have
\begin{equation}
\frac{\partial x^{\prime\kappa}}{\partial x^{\mu}}=\delta_{\mu}^{\kappa
}+\varepsilon\partial_{\mu}\xi^{\kappa} \label{A9}%
\end{equation}
and
\begin{equation}
\phi_{\kappa}(x^{\prime\kappa}(x^{k}))=\phi_{\kappa}(x^{\kappa}+\varepsilon
\xi^{\kappa})=\phi_{\kappa}(x^{\kappa})+\varepsilon\xi^{\alpha}\partial
_{\alpha}\phi_{\kappa}(x^{\kappa}) \label{A10}%
\end{equation}
So,%
\begin{gather}
\mathrm{h}_{\varepsilon}^{\ast}\phi_{\mu~}(x^{\kappa})=\phi_{\kappa}%
(x^{\prime\kappa}(x^{\kappa}))\frac{\partial x^{\prime\kappa}}{\partial
x^{\mu}}\nonumber\\
=\left(  \phi_{\kappa}(x^{\kappa})+\varepsilon\xi^{\alpha}\partial_{\alpha
}\phi_{\kappa}(x^{\kappa})\right)  (\delta_{\mu}^{\kappa}+\varepsilon
\partial_{\mu}\xi^{\kappa})\nonumber\\
=\phi_{\mu}(x^{\kappa})+\varepsilon\partial_{\mu}\xi^{\kappa}\phi_{\kappa
}(x^{\kappa})+\varepsilon\xi^{\alpha}\partial_{\alpha}\phi_{\mu}(x^{\kappa})
\label{A11}%
\end{gather}

Then,%
\begin{equation}
\left(  \pounds _{\boldsymbol{\xi}}\phi(x^{\kappa})\right)  _{\mu}%
=\partial_{\mu}\xi^{\kappa}\phi_{\kappa}(x^{\kappa})+\xi^{\alpha}%
\partial_{\alpha}\phi_{\mu}(x^{\kappa}) \label{A11a}%
\end{equation}

Now, we define following Physics textbooks the \emph{horizontal variation}%
\footnote{In \cite{rc2007} the horizontal variation is denoted by
$\boldsymbol{\delta}_{\mathbf{h}}$, where a vertical variation denoted by
$\boldsymbol{\delta}_{\mathbf{v}}$ (associated with gauge transformations) is
also introduced. Moreover, let us recall that $\delta^{0}${\footnotesize \ }%
has been usely extensively after a famous paper by L. Rosenfeld
\cite{rosenfeld}, but appears also for the best of our knowledge in Section 23
of Pauli's book \cite{pauli} on Relativity theory.} by%
\begin{equation}
\boldsymbol{\delta}^{0}\phi:=-\pounds _{\boldsymbol{\xi}}\phi. \label{A4}%
\end{equation}

This definition (with the negative sign) is used by physicists because they
usually work only with the components of the fields and diffeomorphism
invariance is interpreted as invariance under choice of coordinates. Then,
they interpret Eq.(\ref{A3}) as a coordinate transformation between two charts
whose intersection of domains cover the regions $\mathcal{U}$ and
\textrm{h}$\mathcal{U}$ of interest\ with \emph{coordinate functions}
$\{\boldsymbol{x}^{\mu}\}$ and $\{\boldsymbol{x}^{\prime\mu}\}$ such that
\begin{equation}
\boldsymbol{x}^{\prime\mu}:=\boldsymbol{x}^{\mu}\circ\mathrm{h}_{\varepsilon}
\label{A4a}%
\end{equation}
and then
\begin{equation}
\boldsymbol{x}^{\prime\mu}(x)=x^{\mu}+\varepsilon\xi^{\mu} \label{A4aa}%
\end{equation}
The field $\phi$ at $x\in\mathcal{U\subset}M$ has the representations%
\[
\phi(x)=\phi_{\mu}^{\prime}(x^{\prime\iota})dx^{\prime\mu}=\phi_{\kappa
}(x^{\iota})dx^{\kappa}%
\]
and in first order in $\varepsilon$ it is%
\begin{align}
\phi_{\mu}^{\prime}(x^{\prime\iota})  &  =\frac{\partial x^{\kappa}}{\partial
x^{\prime\mu}}\phi_{\kappa}(x^{\iota})\label{A12}\\
&  =\phi_{\mu}(x^{\iota})-\varepsilon\partial_{\mu}\xi^{\kappa}\phi_{\kappa
}(x^{\iota})\nonumber
\end{align}
and on the other hand since
\begin{equation}
\phi_{\mu}^{\prime}(x^{\prime\iota})=\phi_{\mu}^{\prime}(x^{\iota}%
+\varepsilon\xi^{\iota})=\phi_{\mu}^{\prime}(x^{\iota})+\varepsilon\xi
^{\kappa}\partial_{\kappa}\phi_{\mu}^{\prime}(x^{\iota}). \label{A13}%
\end{equation}
we have in first order in $\varepsilon$ that
\begin{equation}
\phi_{\mu}^{\prime}(x^{\iota})=\phi_{\mu}(x^{\iota})-\varepsilon\partial_{\mu
}\xi^{\kappa}\phi_{\kappa}(x^{\iota})-\varepsilon\xi^{\kappa}\partial_{\kappa
}\phi_{\mu}(x^{\iota}). \label{A14}%
\end{equation}
from where we get
\begin{equation}
\boldsymbol{\delta}^{0}\phi_{\mu}(x)=\lim_{\varepsilon\rightarrow0}\frac
{\phi_{\mu}^{\prime}(x^{\kappa})-\phi_{\mu}(x^{\kappa})}{\varepsilon}=-\left(
\pounds _{\boldsymbol{\xi}}\phi(x^{\kappa})\right)  _{\mu}. \label{A15}%
\end{equation}

\begin{remark}
The above calculating can be done in a while recalling Cartan's\ `magical'
formula, which with $\xi:=\boldsymbol{g}(\boldsymbol{\xi,~})$ reads:
\begin{equation}
\pounds _{\boldsymbol{\xi}}\phi=\xi\lrcorner d\phi+d(\xi\lrcorner\phi)
\label{A16}%
\end{equation}
In components we have%

\begin{align}
\left(  \xi\lrcorner d\phi\right)  _{\alpha}  &  =\xi^{\mu}\partial_{\mu}%
\phi_{\alpha}-\xi^{\mu}\partial_{\alpha}\phi_{\mu},\nonumber\\
\left(  d(\xi\lrcorner\phi)\right)  _{\alpha}  &  =\partial_{\alpha}\xi^{\mu
}\phi_{\mu}+\xi^{\mu}\partial_{\alpha}\phi_{\mu} \label{A17}%
\end{align}
from where the substituting these results in \emph{Eq.(\ref{A16})},
\emph{Eq.(\ref{A11a})} follows immediately.
\end{remark}

\begin{remark}
If we have chosen the coordinate functions $\{\boldsymbol{x}^{\mu}\}$ and
$\{\boldsymbol{x}^{\prime\mu}\}$ related by\emph{\footnote{\emph{A}s, e.g., in
\cite{trautman}.}}
\begin{equation}
\boldsymbol{x}^{\prime\mu}:=\boldsymbol{x}^{\mu}\circ\mathrm{h}_{\varepsilon
}^{-1}%
\end{equation}
we would\ get that $\boldsymbol{\delta}^{0}\phi_{\mu}(x)=\left(
\pounds _{\boldsymbol{\xi}}\phi(x^{\kappa})\right)  _{\mu}$.
\end{remark}

\begin{remark}
Take into account for applications that for any $\mathcal{C\in}\sec%
%TCIMACRO{\tbigwedge }%
%BeginExpansion
{\textstyle\bigwedge}
%EndExpansion
T^{\ast}M$
\begin{equation}
d\pounds _{\boldsymbol{\xi}}\mathcal{C}=\pounds _{\boldsymbol{\xi}%
}d\mathcal{C} \label{A17a}%
\end{equation}

\end{remark}

Now, physicists introduce another two variations $\boldsymbol{\delta}^{a}$ and
$\boldsymbol{\delta}$ defined by
\begin{equation}
\boldsymbol{\delta}^{a}\phi_{\mu}(x):=\lim_{\varepsilon\rightarrow0}\frac
{\phi_{\mu}(x^{\prime\kappa})-\phi_{\mu}(x^{\kappa})}{\varepsilon}=\xi^{\iota
}\partial_{\iota}\phi_{\mu}(x^{\kappa}) \label{A18}%
\end{equation}
and
\[
\boldsymbol{\delta}\phi_{\mu}(x):=\lim_{\varepsilon\rightarrow0}\frac
{\phi_{\mu}^{\prime}(x^{\prime\kappa})-\phi_{\mu}(x^{\kappa})}{\varepsilon}%
\]
called, e.g., in \cite{roman} \emph{local variation\footnote{Some authors call
$\boldsymbol{\delta}\phi_{\mu}(x)$ the total variation, but we think that this
is not an appropriate name.}. We have}%

\begin{align}
\left(  \boldsymbol{\delta}\phi(x)\right)  _{\mu}  &  =\lim_{\varepsilon
\rightarrow0}\frac{\phi_{\mu}^{\prime}(x^{\kappa})+\varepsilon\xi^{\iota
}\partial_{\iota}\phi_{\mu}(x^{\kappa})-\phi_{\mu}(x^{\kappa})}{\varepsilon
}=\left(  \boldsymbol{\delta}^{0}\phi(x^{\kappa})\right)  _{\mu}+\xi^{\iota
}\partial_{\iota}\phi_{\mu}(x^{\kappa})\nonumber\\
&  =-\partial_{\mu}\xi^{\iota}\phi_{\iota}(x^{\kappa})-\xi^{\alpha}%
\partial_{\alpha}\phi_{\mu}(x^{\kappa})+\xi^{\iota}\partial_{\iota}\phi_{\mu
}(x^{\kappa})=-\partial_{\mu}\xi^{\iota}\phi_{\iota}(x^{\kappa}) \label{A19}%
\end{align}

\begin{remark}
In what follows we shall use the above terminology for the various variations
introduced above for an arbitrary tensor field. The definition of the Lie
derivative of spinor fields is is still a subject of recent researh with many
conflicting views. In \cite{lrw2015} we present a novel geoemtrical aprroach
to this subject using the theory of Clifford and spin-Clifford bundles which
seems to lead to a consistent results.
\end{remark}

\begin{remark}
Defining
\begin{equation}
\boldsymbol{\delta}^{0}\mathcal{A}:=-\int\nolimits_{\mathcal{U}}%
\pounds _{\boldsymbol{\xi}}\mathcal{L}_{m}(\phi,d\phi), \label{A19a}%
\end{equation}
we have returning to \emph{Eq.(\ref{A2})} that Stokes theorem permit us to
write
\begin{align}
\boldsymbol{\delta}^{0}\mathcal{A}  &  =-\int\nolimits_{\mathcal{U}%
}\pounds _{\boldsymbol{\xi}}\mathcal{L}_{m}(\phi,d\phi)\nonumber\\
&  =-\int\nolimits_{\mathcal{U}}d(\xi\lrcorner\mathcal{L}_{m}(\phi
,d\phi))-\int\nolimits_{\mathcal{U}}\xi\lrcorner d\mathcal{L}_{m}(\phi
,d\phi)\nonumber\\
&  =-\int\nolimits_{\partial\mathcal{U}}\xi\lrcorner\mathcal{L}_{m}(\phi
,d\phi). \label{A20}%
\end{align}

\end{remark}

\section{The Generalized Energy-Momentum Current\newline in $(M,\boldsymbol{g}%
,\tau_{\boldsymbol{g}},\uparrow)$}

\subsection{The Case of a General Lorentzian Spacetime Structure}

In this subsection $(M,\boldsymbol{g},\tau_{\boldsymbol{g}},\uparrow)$ is an
arbitrary oriented and time oriented Lorentzian manifold $(M,\boldsymbol{g})$
which will be supposed to be the arena where physical phenomena takes place.
We choose coordinate charts $(\mathcal{U}_{1},\chi_{1})$ and $(\mathcal{U}%
_{2},\chi_{2})$ with coordinate functions $\{\boldsymbol{x}^{\mu}\}$ and
$\{\boldsymbol{x}^{\prime\mu}\}$ covering $\mathcal{U}_{1}\cap\mathcal{U}%
_{2}=\mathcal{U}$.\ We call $\chi_{1}(\mathcal{U})=U$, $\chi_{2}%
(\mathcal{U})=U^{\prime}$. We take $\mathcal{U}$ such that $\partial
\mathcal{U}=\Sigma_{2}-\Sigma_{1}+\mathbf{\Xi}$\ , i.e., $\mathcal{U}$ is
bounded from above and below by spacelike surfaces $\Sigma_{1}$ and
$\Sigma_{2}$ such that $\sigma_{1}=\chi_{1}(\Sigma_{1})$ and $\sigma_{2}%
=\chi_{1}(\Sigma_{2})$ and moreover we suppose that set of the $N$ matter
fields in interaction denoted%
\begin{equation}
\boldsymbol{\phi}=\{\phi_{A}\},~~A=1,2,...,N \label{new3}%
\end{equation}
living in $\mathcal{U}$ satisfy in $\mathbf{\Xi}$ (a timelike boundary)%
\begin{equation}
\left.  \phi_{A}\right\vert _{\mathbf{\Xi}}=0. \label{new3b}%
\end{equation}
\ In what follows the action functional for the fields is written%
\begin{equation}
\mathcal{A=}\int\nolimits_{\mathcal{U}}\mathfrak{L}_{m}(\boldsymbol{\phi}%
_{\mu},\partial_{\mu}\boldsymbol{\phi})d^{4}x=\int\nolimits_{\mathcal{U}}%
L_{m}(\boldsymbol{\phi}_{\mu},\partial_{\mu}\boldsymbol{\phi})\sqrt
{-\det\boldsymbol{g}}d^{4}x \label{new3a}%
\end{equation}
Under a coordinate transformation corresponding to a diffeomorphism generated
by a one parameter group of\ diffeomorphisms,%
\begin{align}
x^{\mu}  &  \mapsto x^{\prime\mu}=x^{\mu}+\varepsilon\xi^{\mu}=x^{\mu
}+\varepsilon\boldsymbol{\xi~[}x^{\mu}]\nonumber\\
&  =x^{\mu}-\pounds _{\varepsilon\boldsymbol{\xi}}x^{\mu}=x^{\mu
}+\mathbf{\delta}x^{\mu} \label{s1}%
\end{align}
\newline we already know that the fields suffers the variation
\begin{equation}
\boldsymbol{\phi}\mapsto\boldsymbol{\phi}^{\prime}=\boldsymbol{\phi
}+\boldsymbol{\delta}^{0}\boldsymbol{\phi.} \label{s2}%
\end{equation}

We have in first order in $\varepsilon$ and recalling that $\boldsymbol{\delta
}^{0}$ and $\partial_{\mu}$ commutes that
\begin{gather}
\boldsymbol{\delta}^{0}\mathcal{A=}\int\nolimits_{\mathcal{U}}\mathfrak{L}%
_{m}(\boldsymbol{\phi}_{\mu}^{\prime},\partial_{\mu}\boldsymbol{\phi}^{\prime
})d^{4}x^{\prime}-\int\nolimits_{\mathcal{U}}\mathfrak{L}_{m}(\boldsymbol{\phi
}_{\mu},\partial_{\mu}\boldsymbol{\phi})d^{4}x\nonumber\\
=\int\nolimits_{\mathcal{U}}\left(  \frac{\partial\mathfrak{L}_{m}}%
{\partial\phi_{A}}\boldsymbol{\delta}^{0}\phi_{A}+\frac{\partial
\mathfrak{L}_{m}}{\partial\partial_{\mu}\phi_{A}}\partial_{\mu}%
\boldsymbol{\delta}^{0}\phi_{A}+\mathfrak{L}_{m}\frac{\partial\mathbf{\delta
}x^{\mu}}{\partial x^{\mu}}\right)  d^{4}x\nonumber\\
=\int\nolimits_{\mathcal{U}}\left(  \frac{\partial\mathfrak{L}_{m}}%
{\partial\phi_{A}}-\partial_{\mu}\frac{\partial\mathfrak{L}_{m}}%
{\partial\partial_{\mu}\phi_{A}}\right)  \boldsymbol{\delta}^{0}\phi
_{A}+\partial_{\mu}\left(  \frac{\partial\mathfrak{L}_{m}}{\partial
\partial_{\mu}\phi_{A}}\boldsymbol{\delta}^{0}\phi_{A}+\mathfrak{L}%
_{m}\mathbf{\delta}x^{\mu}\right)  d^{4}x \label{s3}%
\end{gather}
Putting%
\begin{equation}
\mathcal{J}^{\mu}:=\frac{\partial\mathfrak{L}_{m}}{\partial\partial_{\mu}%
\phi_{A}}\boldsymbol{\delta}^{0}\phi_{A}+\mathfrak{L}_{m}\boldsymbol{\delta
}x^{\mu} \label{s4}%
\end{equation}
the second term in Eq.(\ref{s3}) can be written using Gauss theorem as%
\begin{equation}
\int\nolimits_{\mathcal{U}}\partial_{\mu}\left(  \frac{\partial\mathfrak{L}%
_{m}}{\partial\partial_{\mu}\phi_{A}}\boldsymbol{\delta}^{0}\phi
_{A}+\mathfrak{L}_{m}\mathbf{\delta}x^{\mu}\right)  d^{4}x=\int%
\nolimits_{\sigma_{2}}J^{\mu}d\sigma_{\mu}-\int\nolimits_{\sigma_{1}}J^{\mu
}d\sigma_{\mu} \label{s5}%
\end{equation}

Recalling the concept of local variation introduced above we have%
\begin{equation}
\boldsymbol{\delta}\phi_{A}=\boldsymbol{\delta}^{0}\phi_{A}+\mathbf{\delta
}x^{\nu}\partial_{\nu}\phi_{A} \label{s6}%
\end{equation}
Putting
\begin{equation}
\pi_{A}^{\mu}=\frac{\partial\mathfrak{L}_{m}}{\partial\partial_{\mu}\phi_{A}}
\label{ss6}%
\end{equation}
we call%
\begin{equation}
\pi_{A}:=\pi_{A}^{\mu}\partial_{\mu} \label{ss6a}%
\end{equation}
the \emph{canonical momentum canonically conjugated} to the field $\phi_{A}$.
Moreover,\ putting
\begin{equation}
\Upsilon_{\nu}^{\mu}:=\pi_{A}^{\mu}\partial_{\nu}\phi_{A}-\delta_{\nu}^{\mu
}\mathfrak{L}_{m} \label{sss6}%
\end{equation}
we call
\begin{equation}
\boldsymbol{\Upsilon}:=\Upsilon_{\nu}^{\mu}dx^{\nu}\otimes\partial_{\mu}
\label{sss6a}%
\end{equation}
the \emph{canonical} \emph{energy-momentum tensor} of the closed physical
system described by the fields.

Now, we can write Eq.(\ref{s4}) as%
\begin{equation}
\mathcal{J}^{\mu}:=\pi_{A}^{\mu}\boldsymbol{\delta}\phi_{A}-\Upsilon_{\nu
}^{\mu}\mathbf{\delta}x^{\nu}. \label{s7}%
\end{equation}
Moreover, defining
\begin{equation}
F(\sigma):=\int\nolimits_{\sigma}\left(  \pi_{A}^{\mu}\boldsymbol{\delta}%
\phi_{A}-\Upsilon_{\nu}^{\mu}\mathbf{\delta}x^{\nu}\right)  d\sigma_{\mu}
\label{s8}%
\end{equation}
we can rewrite Eq.(\ref{s3})
\begin{equation}
\boldsymbol{\delta}^{0}\mathcal{A=}\int\nolimits_{U}\left(  \frac
{\partial\mathfrak{L}_{m}}{\partial\phi_{A}}-\partial_{\mu}\frac
{\partial\mathfrak{L}_{m}}{\partial\partial_{\mu}\phi_{A}}\right)
\boldsymbol{\delta}^{0}\phi_{A}d^{4}x+F(\sigma_{2})-F(\sigma_{1}) \label{s9}%
\end{equation}
Now, the action principle establishes that $\boldsymbol{\delta}^{0}%
\mathcal{A}=0$ and then we must have
\begin{equation}
\frac{\partial\mathfrak{L}_{m}}{\partial\phi_{A}}-\partial_{\mu}\frac
{\partial\mathfrak{L}_{m}}{\partial\partial_{\mu}\phi_{A}}=0, \label{s10}%
\end{equation}
which are the Euler-Lagrange equations satisfied by each one of the fields
$\phi_{A}$ and also
\begin{equation}
F(\sigma_{2})-F(\sigma_{1})=0 \label{s11}%
\end{equation}
Now, if $\tau_{\boldsymbol{g}}$ is the volume element, taking into account
that we took $\Sigma_{2}-\Sigma_{1}+\mathbf{\Xi}=\partial\mathcal{U}$ where
$\mathbf{\Xi}$ is a timelike surface such that $\left.  \mathcal{J}\right\vert
_{\mathbf{\Xi}}=0$ and introducing the current
\begin{equation}
\mathcal{J}:=\mathcal{J}_{\mu}dx^{\mu}=g_{\mu\nu}\mathcal{J}^{\nu}dx^{\mu}%
\in\sec\bigwedge\nolimits^{1}T^{\ast}M\hookrightarrow\sec\mathcal{C\ell
}(M,\mathtt{g}) \label{S12}%
\end{equation}
we can rewrite Eq.(\ref{s11}) using Stokes theorem as%
\begin{equation}
\int\nolimits_{\sigma_{2}}\mathcal{J}^{\mu}d\sigma_{\mu}-\int\nolimits_{\sigma
_{1}}\mathcal{J}^{\mu}d\sigma_{\mu}=\int_{\partial\mathcal{U}}%
\underset{\boldsymbol{g}}{\star}\mathcal{J}=\int_{U}d\underset{\boldsymbol{g}%
}{\star}\mathcal{J}\boldsymbol{.} \label{s13}%
\end{equation}

\subsection{Introducing $\boldsymbol{D}$ and the Covariant \textquotedblleft
Conservation\textquotedblright\ Law for $\mathbf{\Upsilon}$}

If we add\ $\boldsymbol{D}$, the Levi-Civita connection of $\boldsymbol{g}$ to
$(M,\mathbf{g},\tau_{\boldsymbol{g}},\uparrow)$ we get a Lorentzian spacetime
structure $(M,\boldsymbol{g},\boldsymbol{D},\tau_{\boldsymbol{g}},\uparrow)$.
Then recalling from Appendix A the definitions of the Hodge coderivative
and\ of the Dirac operator we can write:
\begin{align}
\int_{\mathcal{U}}d\star\mathcal{J} &  =\int_{\mathcal{U}}%
\underset{\boldsymbol{g}}{\star}\underset{\boldsymbol{g}}{\star}%
^{-1}d\underset{\boldsymbol{g}}{\star}\mathcal{J}\nonumber\\
&  =-\int_{\mathcal{U}}\left(  \underset{\boldsymbol{g}}{\delta}%
\mathcal{J}\right)  \tau_{\boldsymbol{g}}\nonumber\\
&  =\int_{\mathcal{U}}\boldsymbol{\partial\lrcorner}\mathcal{J}\boldsymbol{~}%
\tau_{\boldsymbol{g}}\nonumber\\
&  =\int_{U}\boldsymbol{D}_{\mu}\mathcal{J}^{\mu}~\tau_{\boldsymbol{g}%
}\label{s14}%
\end{align}
and we arrive at the conclusion that $\boldsymbol{\delta}^{0}\mathcal{A}=0$
implies that
\begin{equation}
d\underset{\boldsymbol{g}}{\star}\mathcal{J=}0~~~~\boldsymbol{\Leftrightarrow
D}_{\mu}\mathcal{J}^{\mu}=0~~\Leftrightarrow\frac{1}{\sqrt{-\det
\boldsymbol{g}}}\partial_{\mu}(\sqrt{-\det\boldsymbol{g}}\mathcal{J}^{\mu
})=0.\label{s15}%
\end{equation}
\ \ 

\begin{remark}
Recalling the definition of the canonical energy-momentum tensor
$\mathbf{\Upsilon}$ \emph{(Eq.(\ref{sss6a})}) gives a covariant
\textquotedblleft conservation\textquotedblright\ law for $\mathbf{\Upsilon}$,
i.e.,%
\begin{equation}
\boldsymbol{D}\bullet\mathbf{\Upsilon}=0 \label{s15a}%
\end{equation}
only if the term $\pi_{A}^{\mu}\boldsymbol{\delta}\phi_{A}$ in the current
$J^{\mu}$ is null. This, of course, happens if the local variation of the
fields $\boldsymbol{\delta}\phi_{A}=0$, something that cannot happens in an
arbitrary structure $(M,\boldsymbol{g},\tau_{\boldsymbol{g}},\uparrow)$. So,
we need to investigate when this occurs.
\end{remark}

\begin{remark}
We observe here that comparison of \emph{Eq.(\ref{s9}) with Eq.(\ref{A20})
permit us to write }%
\begin{equation}
\int_{\partial\mathcal{U}}\underset{\boldsymbol{g}}{\star}\mathcal{J}%
=\int_{\mathcal{U}}d\underset{\boldsymbol{g}}{\star}\mathcal{J}=-\int%
\nolimits_{\partial\mathcal{U}}\xi\lrcorner\mathcal{L}(\boldsymbol{\phi
},d\boldsymbol{\phi})=-\int\nolimits_{\mathcal{U}}d(\xi\lrcorner
\mathcal{L}(\boldsymbol{\phi},d\boldsymbol{\phi})). \label{s16}%
\end{equation}

\end{remark}

\subsection{The Case of Minkowski Spacetime}

\subsubsection{The Canonical Energy-Momentum Tensor in Minkowski Spacetime}

We now, apply the results of the last section to the case where the fields
live in Minkowski spacetime $(M,\boldsymbol{\eta},\mathrm{D},\tau
_{\boldsymbol{\eta}},\uparrow)$. In this case we can introduce global
coordinates $\{\mathrm{x}^{\mu}\}$ in Einstein-Lorentz-Poincar\'{e} gauge (see
Section 2.3).

We now construct the conserved current associated to the diffeomorphisms
generated \ by the vector fields $e_{\mu}=\partial/\partial\mathrm{x}^{\mu}$
which are Killing vector fields on $(M,\boldsymbol{\eta})$. Consider then the
Killing vector field%
\begin{equation}
\boldsymbol{\xi}:=\varepsilon^{\mu}e_{\mu} \label{m1}%
\end{equation}
where $\varepsilon^{\mu}\in\mathbb{R}$ are\emph{ constants }such that
$\left\vert \varepsilon^{\mu}\right\vert <<1$ and the coordinate
transformation%
\begin{equation}
x^{\mu}\mapsto x^{\prime\mu}=x^{\mu}+\pounds _{\boldsymbol{\xi}}x^{\mu}%
=x^{\mu}+\varepsilon^{\mu}. \label{m2}%
\end{equation}

Recalling\ the definitions of $\boldsymbol{\pi}_{A}$, the momentum canonically
conjugated to the field $\phi_{A}$ (Eq.(\ref{ss6a})) and of the
$\boldsymbol{\Upsilon}$ (Eq.(\ref{sss6a})) and recalling that in the present
case it is $\boldsymbol{\delta}\phi_{A}=0$\ we have that the conserved Noether
current is:%
\begin{equation}
\mathcal{J}^{\mu}=-\varepsilon^{\nu}\Upsilon_{\nu}^{\mu}. \label{m3}%
\end{equation}
and the \emph{canonical} energy-momentum tensor of the physical system
described by the fields $\phi_{A}$ is conserved, i.e.,
\begin{equation}
\partial_{\mu}\Upsilon_{\nu}^{\mu}=0. \label{mc1}%
\end{equation}

\begin{remark}
Of course, if we introduce an arbitrary coordinate system $\{x^{\mu}\}$
covering an open set $\mathcal{U}$ of the Minkowski spacetime
manifold\ $M\simeq\mathbb{R}^{4}$ \emph{Eq.(\ref{mc1}) }reads%
\begin{equation}
\mathrm{D}\bullet\boldsymbol{\Upsilon}=0\Leftrightarrow\mathrm{D}_{\mu
}\Upsilon_{\nu}^{\mu}=0. \label{mc2}%
\end{equation}

\end{remark}

\subsubsection{The Energy-Momentum $1$-Form Fields in Minkowski Spacetime}

Recall that the objects
\begin{equation}
\boldsymbol{\Upsilon}^{\mu}:=\left(  \pi_{A}^{\mu}\partial_{\nu}\phi
_{A}-\delta_{\nu}^{\mu}\mathfrak{\ell}\right)  d\mathrm{x}^{\nu}=\Upsilon
_{\nu}^{\mu}d\mathrm{x}^{\nu}\in\sec\bigwedge\nolimits^{1}T^{\ast
}M\hookrightarrow\sec\mathcal{C\ell}(M,\mathtt{\eta}) \label{mm8}%
\end{equation}
are conserved currents. They may be called the generalized energy momentum
$1$-form fields of the physical system described by the fields $\phi_{A}$. We
have
\begin{equation}
\underset{\boldsymbol{\eta}}{\delta}\boldsymbol{\Upsilon}^{\mu}%
=0\Leftrightarrow d\underset{\boldsymbol{\eta}}{\star}\boldsymbol{\Upsilon
}^{\mu}=0. \label{mm9}%
\end{equation}

\subsubsection{The Belinfante Energy-Momentum Tensor in Minkowski Spacetime}

It happens that given an arbitrary field theory the canonical energy-momentum
tensor is in general not symmetric, i.e.,%
\begin{equation}
\Upsilon^{\mu\nu}\neq\Upsilon^{\nu\mu}. \label{mm11}%
\end{equation}
But, of course, if $\boldsymbol{\Upsilon}^{\mu}$ is conserved, so it is
\begin{equation}
\mathcal{T}^{\mu}:=\boldsymbol{\Upsilon}^{\mu}+\boldsymbol{\digamma}^{\mu}
\label{mm12}%
\end{equation}
where
\begin{equation}
\boldsymbol{\digamma}^{\mu}:=\underset{\boldsymbol{\eta}}{\delta
}\boldsymbol{\Psi}^{\mu} \label{mm13}%
\end{equation}
with each one of the $\boldsymbol{\Psi}^{\mu}\in\sec\bigwedge\nolimits^{2}%
T^{\ast}M\hookrightarrow\sec\mathcal{C\ell}(M,\mathtt{g})$. So, it is always
possible for any field theory to find\footnote{For example, in \cite{ll} \ the
condition is fixed in such a way that the orbital angular momentum tensor of
the system defined as $\boldsymbol{M}=\frac{1}{2}M_{\mu\nu}\left.
d\mathrm{x}^{\mu}\right\vert _{o}\wedge\left.  d\mathrm{x}^{\nu}\right\vert
_{o}$ (with $M^{\mu\nu}=\int(\mathrm{x}^{\mu}T^{\nu\kappa}-\mathrm{x}^{\nu
}T^{\mu\kappa})d\sigma_{\kappa}$) be automatically conserved. However take
into account that since the fields possess in general intrinsic spin an
angular momentum conservation law can only be formulated by taking into
account the orbital and spin angular momenta. It can be shown (see Chapter 8
of \cite{rc2007} that the antisymmetric part of the canonical energy-momentum
tensor is the source of the spin tensor of the field. In \cite{barut} it is
shown how to obtain a conserved symmetrical energy-momentum tensor by studying
the conservation laws that come from a general Poincar\'{e} variation, which
involves translations and general Lorentz transformations.} a condition on the
$\Psi^{\mu}$ such that \ the components $T_{\nu}^{\mu}$ of $\mathcal{T}^{\mu
}=T_{\nu}^{\mu}d\mathrm{x}^{\nu}$ satisfy the symmetry condition.
\begin{equation}
T^{\mu\nu}=T^{\nu\mu}. \label{mm14}%
\end{equation}
When this is the case $\boldsymbol{T}=T_{\nu}^{\mu}d\mathrm{x}^{\nu}%
\otimes\partial_{\mu}$ will be called the \emph{Belinfante} energy-momentum
tensor of the system.

\subsubsection{The Energy-Momentum Covector in Minkowski Spacetime}

Since Minkowski spacetime is parallelizable we can identify all tangent and
cotangent spaces and thus define a \emph{covector} in a vector space
$\mathcal{V\simeq}\mathbb{R}^{4}$. Fixing (global) coordinates in
Einstein-Lorentz-Poincar\'{e} gauge $\{\mathrm{x}^{\mu}\}$ a vector
$\mathbf{v}_{x}\in$ $T_{x}M$ can be identified by a pair \cite{oneillo}
$(\mathbf{x},\mathbf{v})$ where $(\mathbf{x},\mathbf{v})\in\mathbb{R}%
^{4}\times\mathbb{R}^{4}$ and $\mathbf{x}=(\mathrm{x}^{0},\mathrm{x}%
^{1},\mathrm{x}^{2},\mathrm{x}^{3})$. If two vectors $\mathbf{v}_{x}\in$
$T_{x}M$, $\mathbf{v}_{y}\in$ $T_{y}M$ are such that%
\begin{equation}
\mathbf{v}_{x}=(\mathbf{x,v}),\mathbf{v}_{y}=(\mathbf{y,v}), \label{mm14a}%
\end{equation}
i.e., they have the same vector part we will say that they can be identified
as a a vector of some vector space $\mathbf{V}\mathcal{\simeq}\mathbb{R}^{4}.$
With these considerations we write $\forall x,y\in M\simeq\mathbb{R}^{4}$%

\begin{equation}
\left.  \frac{\partial}{\partial\mathrm{x}^{\mu}}\right\vert _{x}%
\approx\left.  \frac{\partial}{\partial\mathrm{x}^{\mu}}\right\vert
_{y}=\mathbf{E}_{\mu}~~\text{and ~~}\left.  d\mathrm{x}^{\mu}\right\vert
_{x}\approx\left.  d\mathrm{x}^{\mu}\right\vert _{y}=\boldsymbol{E}^{\mu}
\label{mmm14}%
\end{equation}
where $\{\mathbf{E}_{\mu}\}$ is a basis of $\mathbf{V}$ and $\{\boldsymbol{E}%
^{\mu}\}$ is a basis of a $\mathcal{V\simeq}\mathbb{R}^{4}$. Then we can write
(with $o\in M$ an \emph{arbitrary} point, taken in general, for convenience as
origin of the coordinate system)
\begin{equation}
\boldsymbol{P}=P_{\mu}\boldsymbol{E}^{\mu}=P_{\mu}\left.  d\mathrm{x}^{\mu
}\right\vert _{0}:=\left(  \int\underset{\boldsymbol{\eta}}{\star}%
\mathcal{T}_{\mu}\right)  \left.  d\mathrm{x}^{\mu}\right\vert _{0}
\label{mm10}%
\end{equation}
as the energy-momentum \emph{covector} of the closed physical system described
by the fields $\phi_{A}$.

\begin{remark}
Note that under a global \emph{(}constant\emph{)} Lorentz transformation
$\partial/\partial\mathrm{x}^{\mu}\mapsto\partial/\partial\mathrm{x}%
^{\prime\mu}=\Lambda_{\mu}^{\nu}\partial/\partial\mathrm{x}^{\nu}$we have that
$\underset{\boldsymbol{\eta}}{\star}\mathcal{T}_{\mu}\mapsto
\underset{\boldsymbol{\eta}}{\star}\mathcal{T}_{\mu}^{\prime}%
=\underset{\boldsymbol{\eta}}{\star}\mathcal{T}_{\nu}\Lambda_{\mu}^{\nu}$ and
it results $P_{\mu}\mapsto P_{\mu}^{\prime}=P_{\nu}\Lambda_{\mu}^{\nu}$, i.e.,
the $P_{\mu}$ are indeed the components of a covector under any global
\emph{(}constant\emph{)} Lorentz transformation.
\end{remark}

\section{The Energy-Momentum Tensor of Matter in GRT}

The result of the previous section shows that in a general Lorentzian
spacetime structure the canonical Lagrangian formalism does not give a
covariant \textquotedblleft conserved\textquotedblright\ energy-mometum tensor
unless the local variations of the matter fields $\boldsymbol{\delta}\phi_{A}$
are null. So, in \textbf{GRT }the matter energy momentum tensor
$\boldsymbol{T}=T_{\nu}^{\mu}d\mathrm{x}^{\nu}\otimes\frac{\partial}%
{\partial\mathrm{x}^{\mu}}$ that enters Einstein equation is symmetric (i.e.,
$T^{\mu\nu}=T^{\nu\mu}$) and it is obtained in the following way. We start
with the matter action%
\begin{equation}
\mathcal{A}_{m}=\int\nolimits_{\mathcal{U}}L_{m}(\boldsymbol{\phi}_{\mu
},\partial_{\mu}\boldsymbol{\phi})\sqrt{-\det\boldsymbol{g}}d^{4}x \label{gr1}%
\end{equation}

\begin{remark}
Consider as above, a diffeomorphism $x\mapsto x^{\prime}=\mathrm{h}%
_{\varepsilon}(x)$ generated by a one parameter group associated to a vector
field $\mathbf{\xi=\xi}^{\mu}\partial_{\mu}$ \emph{(}such that components
$\left\vert \mathbf{\xi}^{\mu}\right\vert \mathbf{<<}1$ and the $\mathbf{\xi
}^{\mu}\rightarrow0$ at $\mathbf{\Xi}$\emph{)} and a corresponding coordinate
transformation%
\begin{equation}
x^{\mu}\mapsto x^{\prime\mu}=x^{\mu}+\varepsilon\mathbf{\xi}^{\mu} \label{gr2}%
\end{equation}
with $\left\vert \varepsilon\right\vert <<1$ and study the variation of
$\mathcal{A}_{m}$ induced the variation of the \emph{(}gravitational\emph{)}
field $\boldsymbol{g}$ \emph{(}without changing the fields $\phi_{A}$\emph{)}
induced by the coordinate transformation of \emph{Eq.(\ref{gr2})}. We have
immediately that
\begin{equation}
g^{_{\mu\nu}}(x^{\kappa})\mapsto g^{\prime_{\mu\nu}}(x^{\prime\kappa
})=g^{_{\mu\nu}}(x^{\kappa})+\varepsilon g^{_{\mu\iota}}(x^{\kappa}%
)\partial_{\iota}\xi^{\nu}+\varepsilon g^{_{\nu\iota}}(x^{\kappa}%
)\partial_{\iota}\xi^{\mu}. \label{gr3}%
\end{equation}
To first order in $\varepsilon$ it is%
\begin{equation}
g^{\prime_{\mu\nu}}(x^{\kappa})=g^{_{\mu\nu}}(x^{\kappa})+\varepsilon
g^{_{\mu\iota}}(x^{\kappa})\partial_{\iota}\xi^{\nu}+\varepsilon g^{_{\nu
\iota}}(x^{\kappa})\partial_{\iota}\xi^{\mu}-\varepsilon\xi^{\iota}%
\partial_{\iota}g^{_{\mu\nu}} \label{gr4}%
\end{equation}
and
\begin{equation}
\boldsymbol{\delta}^{0}g^{_{\mu\nu}}=-\pounds _{\boldsymbol{\xi}}g^{\mu\nu
}=\boldsymbol{D}^{\nu}\xi^{\mu}+\boldsymbol{D}^{\mu}\xi^{\nu}\text{ and
}\boldsymbol{\delta}^{0}g_{\mu\nu}=-\pounds _{\boldsymbol{\xi}}g_{\mu\nu
}=-\boldsymbol{D}_{\nu}\xi_{\mu}-\boldsymbol{D}_{\mu}\xi_{\nu}. \label{gr5}%
\end{equation}
Then, under the above conditions, using Gauss theorem and supposing that
$\boldsymbol{\delta}^{0}g^{_{\mu\nu}}$ vanishes at $\mathbf{\Xi}$ it is%
\begin{align}
\boldsymbol{\delta}^{0}\mathcal{A}_{m}  &  :=\int\nolimits_{\mathcal{U}%
}\left\{  \frac{\partial L_{m}\sqrt{-\det\boldsymbol{g}}}{\partial g^{_{\mu
\nu}}}\boldsymbol{\delta}^{0}g^{_{\mu\nu}}+\frac{\partial L_{m}\sqrt
{-\det\boldsymbol{g}}}{\partial\partial_{\iota}g^{_{\mu\nu}}}%
\boldsymbol{\delta}^{0}\left(  \partial_{\iota}g^{_{\mu\nu}}\right)  \right\}
d^{4}x\nonumber\\
&  =\int\nolimits_{\mathcal{U}}\left\{  \frac{\partial L_{m}\sqrt
{-\det\boldsymbol{g}}}{\partial g^{_{\mu\nu}}}\boldsymbol{\delta}^{0}%
g^{_{\mu\nu}}+\frac{\partial L_{m}\sqrt{-\det\boldsymbol{g}}}{\partial
\partial_{\iota}g^{_{\mu\nu}}}\partial_{\iota}\boldsymbol{\delta}^{0}%
g^{_{\mu\nu}}\right\}  d^{4}x\nonumber\\
&  =\int\nolimits_{\mathcal{U}}\left\{  \frac{\partial L_{m}\sqrt
{-\det\boldsymbol{g}}}{\partial g^{_{\mu\nu}}}-\frac{\partial}{\partial
x^{\iota}}\left(  \frac{\partial L_{m}\sqrt{-\det\boldsymbol{g}}}%
{\partial\partial_{\iota}g^{_{\mu\nu}}}\right)  \right\}  \boldsymbol{\delta
}^{0}g^{_{\mu\nu}}d^{4}x\nonumber\\
&  =:\frac{1}{2}\int\nolimits_{\mathcal{U}}T_{\mu\nu}\boldsymbol{\delta}%
^{0}g^{_{\mu\nu}}\sqrt{-\det\boldsymbol{g}}d^{4}x=-\frac{1}{2}\int%
\nolimits_{\mathcal{U}}T^{\mu\nu}\boldsymbol{\delta}^{0}g_{\mu\nu}\sqrt
{-\det\boldsymbol{g}}d^{4}x. \label{gr6}%
\end{align}
with%
\begin{equation}
\frac{1}{2}\sqrt{-\det\boldsymbol{g}}T_{\mu\nu}:=\frac{\partial L_{m}%
\sqrt{-\det\boldsymbol{g}}}{\partial g^{_{\mu\nu}}}-\frac{\partial}{\partial
x^{\iota}}\left(  \frac{\partial L_{m}\sqrt{-\det\boldsymbol{g}}}%
{\partial\partial_{\iota}g^{_{\mu\nu}}}\right)  \label{gr7}%
\end{equation}
Since $T_{\mu\nu}=T_{\nu\mu}$ we can write%
\begin{align}
\boldsymbol{\delta}^{0}\mathcal{A}_{m}  &  =\frac{1}{2}\int%
\nolimits_{\mathcal{U}}T_{\mu\nu}\boldsymbol{\delta}^{0}g^{_{\mu\nu}}%
\sqrt{-\det\boldsymbol{g}}d^{4}x\nonumber\\
&  =\frac{1}{2}\int\nolimits_{\mathcal{U}}T^{\mu\nu}(\boldsymbol{D}_{\nu}%
\xi_{\mu}+\boldsymbol{D}_{\mu}\xi_{\nu})\sqrt{-\det\boldsymbol{g}}%
d^{4}x\nonumber\\
&  =\int\nolimits_{\mathcal{U}}T^{\mu\nu}\boldsymbol{D}_{\nu}\xi_{\mu}%
\sqrt{-\det\boldsymbol{g}}d^{4}x\nonumber\\
&  =\int\nolimits_{\mathcal{U}}\boldsymbol{D}_{\nu}(T_{\mu}^{\nu}\xi^{\mu}%
\mu)\sqrt{-\det\boldsymbol{g}}d^{4}x-\int\nolimits_{\mathcal{U}}%
(\boldsymbol{D}_{\nu}T_{\mu}^{\nu})\xi^{\mu}\sqrt{-\det\boldsymbol{g}}%
d^{4}x\nonumber\\
&  =\int\nolimits_{\mathcal{U}}\partial_{\nu}(\sqrt{-\det\boldsymbol{g}}%
T_{\mu}^{\nu}\xi^{\mu}\mu)d^{4}x-\int\nolimits_{\mathcal{U}}(\boldsymbol{D}%
_{\nu}T_{\mu}^{\nu})\xi^{\mu}\sqrt{-\det\boldsymbol{g}}d^{4}x\nonumber\\
&  =-\int\nolimits_{\mathcal{U}}(\boldsymbol{D}_{\nu}T_{\mu}^{\nu})\xi^{\mu
}\sqrt{-\det\boldsymbol{g}}d^{4}x \label{gr8}%
\end{align}

\end{remark}

\begin{remark}
\label{dtt=0}Contrary to what is stated in many textbooks in \textbf{GRT} we
cannot conclude with the ingredients introduced in this section that
$\boldsymbol{\delta}^{0}\mathcal{A}_{m}=0$.

However, if we take into account that in \textbf{GRT} the total action
describing the mater fields and the gravitational field is
\[
\mathcal{A=}%
%TCIMACRO{\tint }%
%BeginExpansion
{\textstyle\int}
%EndExpansion
\mathcal{L}_{g}+%
%TCIMACRO{\tint }%
%BeginExpansion
{\textstyle\int}
%EndExpansion
\mathcal{L}_{m}=-\frac{1}{2}%
%TCIMACRO{\tint }%
%BeginExpansion
{\textstyle\int}
%EndExpansion
R\boldsymbol{\tau}_{\boldsymbol{g}}+%
%TCIMACRO{\tint }%
%BeginExpansion
{\textstyle\int}
%EndExpansion
\mathcal{L}_{m}%
\]
we get from the variation $\boldsymbol{\delta}^{0}\mathcal{A}$ induced by the
variation of the \emph{(}gravitational\emph{)} field $\boldsymbol{g}$
\emph{(}without changing the fields $\phi_{A}$\emph{)} and induced by the
coordinate transformation given by \emph{Eq.(\ref{gr2})} the Einstein field
equation $\boldsymbol{G=-T}$ which reads in components as
\begin{equation}
G_{\nu}^{\mu}=R_{\nu}^{\mu}-\frac{1}{2}R\delta_{\nu}^{\mu}=-T_{\nu}^{\mu}.
\label{gr9}%
\end{equation}
Since it is $\boldsymbol{D\bullet G}=0$ it follows that in GRT we have
\begin{equation}
\boldsymbol{D\bullet T}=0. \label{gr1010}%
\end{equation}

\end{remark}

\begin{remark}
It is opportune to recall that as observed, e.g., by Weinberg
\emph{\cite{weinberg}} that for the case of Minkowski spacetime the symmetric
energy-momentum tensor obtained by the above method is always equal to a
convenient symmetrization of the canonical energy-momentum tensor. But it is
necessary to have in mind that the $\mathbf{GRT}$ procedure eliminates a
legitimate conserved current $\mathcal{J}$ introducing a covariant
\textquotedblleft conserved\textquotedblright\ energy-momentum tensor that
does not give any legitimate energy-momentum conserved current for the matter
fields, except for the particular Lorentzian spacetimes containing appropriate
Killing vector fields. And even in this case no energy-momentum covector as it
exists in special relativistic theories can be defined. Moreover, at this
point we cannot forget the existence of the quantum structure of matter fields
which experimentally says the Minkowskian concept of energy and momentum
\emph{(}in general quantized\emph{)} being carried by field excitations that
one calls particles. This strongly suggests that parodying \ (again) Sachs and
Wu \emph{\cite{sw}} it is really a shame to loose the special relativistic
conservations laws in $\mathbf{GRT}$.
\end{remark}

\section{Relative Tensors and their Covariant Derivatives}

Now, recall that given arbitrary coordinates $\{x^{\alpha}\}$ covering
$U\subset$ $M$ \ and $\{x^{\prime\alpha}\}$ covering covering $V\subset$ $M$
($U\cap V\neq\varnothing$)\ a relative tensor $\mathfrak{A}$ of type $(r,s)$
and weight\footnote{The number $w$ is an integer. Of course, if $w=0$ we are
back to tensor fields.} $w$ is a section of the bundle\footnote{The notation
$(%
%TCIMACRO{\tbigwedge \nolimits^{4}}%
%BeginExpansion
{\textstyle\bigwedge\nolimits^{4}}
%EndExpansion
T^{\ast}M)^{\otimes w}$ means the $w$-fold tensor product of $%
%TCIMACRO{\tbigwedge \nolimits^{4}}%
%BeginExpansion
{\textstyle\bigwedge\nolimits^{4}}
%EndExpansion
T^{\ast}M$ with itself.} $T_{q}^{p}M\otimes(%
%TCIMACRO{\tbigwedge \nolimits^{4}}%
%BeginExpansion
{\textstyle\bigwedge\nolimits^{4}}
%EndExpansion
T^{\ast}M)^{\otimes w}$.

We have
\[
\mathfrak{A=A}_{\nu_{1}...\nu_{s}}^{\mu_{1}...\mu_{r}}(x^{\alpha}%
)\partial_{\mu_{1}}\otimes\cdots\otimes\partial_{\mu_{r}}\otimes dx^{\nu_{1}%
}\otimes\cdots\otimes dx^{\nu_{s}}\otimes(\tau)^{\otimes w},
\]
with $\tau:=dx^{0}\wedge\cdots\wedge dx^{3}$. The set of functions%
\[
\mathfrak{A}_{\nu_{1}...\nu_{s}}^{\mu_{1}...\mu_{r}}(x^{\alpha})=\left(
\sqrt{-\det\boldsymbol{g}}\right)  ^{w}A_{\nu_{1}...\nu_{s}}^{\mu_{1}%
...\mu_{r}}(x^{\alpha})
\]
is said to be the components of the relative tensor field $\mathfrak{A\in
\sec(}T_{s}^{r}M\otimes(%
%TCIMACRO{\tbigwedge \nolimits^{4}}%
%BeginExpansion
{\textstyle\bigwedge\nolimits^{4}}
%EndExpansion
T^{\ast}M)^{w})$ and under a coordinate transformation $x^{\alpha}\mapsto
x^{\prime\beta}$ with Jacobian $J=\det\left(  \frac{\partial x^{\prime\alpha}%
}{\partial x^{\beta}}\right)  $ these functions transform as
\cite{lovrund,tiee}%
\begin{equation}
\mathfrak{A}_{\kappa_{1}...\kappa_{s}}^{\prime\lambda_{1}...\lambda_{r}%
}(x^{\prime\beta})=J^{w}\frac{\partial x^{\prime\lambda_{1}}}{\partial
x^{\mu_{1}}}\cdots\frac{\partial x^{\prime\lambda_{1}}}{\partial x^{\mu_{1}}%
}\frac{\partial x^{\nu_{1}}}{\partial x^{\kappa_{1}}}...\frac{\partial
x^{\nu_{s}}}{\partial x^{\kappa_{s}}}\mathfrak{A}_{\nu_{1}...\nu_{s}}^{\mu
_{1}...\mu_{r}}(x^{\alpha}). \label{relative tensor}%
\end{equation}
On a manifold $M$ equipped with a metric tensor field $\boldsymbol{g}$ we can
write $\mathfrak{A}_{\nu_{1}...\nu_{s}}^{\mu_{1}...\mu_{r}}(x^{\alpha
})=\left(  \sqrt{-\det\boldsymbol{g}}\right)  ^{w}A_{\nu_{1}...\nu_{s}}%
^{\mu_{1}...\mu_{r}}(x^{\alpha})$ where the $A_{\nu_{1}...\nu_{s}}^{\mu
_{1}...\mu_{r}}(x^{\alpha})$ are the components of a tensor field $A\in\sec
T_{s}^{r}M$.

The \textit{covariant derivative of a relative tensor field} relative to a
given arbitrary connection $\nabla$ defined on $M$ such that $\nabla
_{\frac{\partial}{\partial x^{\nu}}}dx^{\mu}=-\ell_{\cdot\nu\alpha}^{\mu
\cdot\cdot}dx^{\alpha}$ is given (as the reader may easily find) by%
\begin{equation}
\nabla_{\partial_{\kappa}}\mathfrak{A:=(}\nabla_{\kappa}\mathfrak{A}_{\nu
_{1}...\nu_{s}}^{\mu_{1}...\mu_{r}})\partial_{\mu_{1}}\otimes\cdots
\otimes\partial_{\mu_{r}}\otimes dx^{\nu_{1}}\otimes\cdots\otimes dx^{\nu_{s}%
}\otimes(\tau)^{\otimes w}, \label{cdrti}%
\end{equation}
where%
\begin{align}
\nabla_{\kappa}\mathfrak{A}_{\nu_{1}...\nu_{s}}^{\mu_{1}...\mu_{r}}  &
=\frac{\partial}{\partial x^{\kappa}}\mathfrak{A}_{\nu_{1}...\nu_{s}}^{\mu
_{1}...\mu_{r}}+\ell_{\cdot\iota\kappa}^{\mu_{p\cdot\cdot}}\mathfrak{A}%
_{\nu_{1}....................\nu_{s}}^{\mu_{1}...\mu_{p-1}\iota\mu_{p+1}%
...\mu_{r}}\nonumber\\
&  -\ell_{\cdot\nu_{q}\kappa}^{\iota\cdot\cdot}\mathfrak{A}_{\nu_{1}%
...\nu_{q-1}\iota\nu_{q+1}...\nu_{s}}^{\mu_{1}....................\mu_{r}%
}-w\ell_{\cdot\kappa\sigma}^{\sigma\cdot\cdot}\mathfrak{A}_{\nu_{1}...\nu_{s}%
}^{\mu_{1}...\mu_{r}}. \label{cdrt}%
\end{align}

In particular for the Levi-Civita connection $\boldsymbol{D}$ of
$\boldsymbol{g}$ we have for the relative tensor
\begin{equation}
\tau_{\boldsymbol{g}}=\sqrt{-\det\boldsymbol{g}}\otimes dx^{0}\wedge
\cdots\wedge dx^{3}%
\end{equation}
that:
\begin{align}
\boldsymbol{D}_{\alpha}\left(  \sqrt{-\det\boldsymbol{g}}\right)   &
=\partial_{\gamma}\left(  \sqrt{\det\boldsymbol{g}}\right)  -\mathbf{\Gamma
}_{\cdot\gamma\rho}^{\rho\cdot\cdot}\sqrt{\det\boldsymbol{g}}=0,\nonumber\\
\boldsymbol{D}_{\alpha}\left(  \frac{1}{\sqrt{-\det\boldsymbol{g}}}\right)
&  =\partial_{\gamma}\left(  \frac{1}{\sqrt{-\det\boldsymbol{g}}}\right)
+\mathbf{\Gamma}_{\cdot\gamma\rho}^{\rho\cdot\cdot}\frac{1}{\sqrt
{-\det\boldsymbol{g}}}=0. \label{particular}%
\end{align}

\section{Explicit Formulas for $\mathbf{J}_{\mu\nu}$ and $\mathbf{J}_{\mu4}$
in Terms of Projective Conformal Coordinates}

We have taking into account Eqs.(\ref{ds1a}), (\ref{ds2}) and (\ref{ds3}) the
following identities%

\begin{align}
\frac{\partial X^{\kappa}}{\partial x^{\alpha}}  &  =\frac{\Omega^{2}}%
{2\ell^{2}}x_{\alpha}x^{\kappa}+\Omega\delta_{\alpha}^{\kappa},\nonumber\\
\frac{\partial X^{4}}{\partial x^{\alpha}}  &  =-\frac{\Omega^{2}}{\ell
}x_{\alpha},~~~X^{\mu}=\Omega x^{\mu}.\nonumber
\end{align}
where $x_{\mu}:=\eta_{\mu\nu}x^{\nu}$ and $X_{\mu}:=\mathring{\eta}_{\mu\nu
}X^{\nu}$. We want to prove that:%
\begin{align}
\mathbf{(a)}\text{~~~}\mathbf{J}_{\mu\nu}  &  =\eta_{\mu\beta}x^{\beta}%
\frac{\partial}{\partial x^{\nu}}-\eta_{\nu\beta}x^{\beta}\frac{\partial
}{\partial x^{\mu}}=\mathring{\eta}_{\mu\beta}X^{\beta}\frac{\partial
}{\partial X^{\nu}}-\mathring{\eta}_{\nu\beta}X^{\beta}\frac{\partial
}{\partial X^{\mu}},\label{3L}\\
\mathbf{(b)}~~~\mathbf{J}_{\mu4}  &  =\ell\frac{\partial}{\partial x^{\mu}%
}-\frac{1}{4\ell}\left(  2\eta_{\mu\nu}x^{\nu}x^{\lambda}-\sigma^{2}%
\delta_{\mu}^{\lambda}\right)  \frac{\partial}{\partial x^{\lambda}}%
=-X^{4}\frac{\partial}{\partial X^{\mu}}+X_{\mu}\frac{\partial}{\partial
X^{4}}. \label{4L}%
\end{align}
\medskip

\textbf{Proof of (a):}%

\begin{gather*}
\mathbf{J}_{\mu\nu}=\eta_{\mu\beta}x^{\beta}\frac{\partial X^{\kappa}%
}{\partial x^{\nu}}\frac{\partial}{\partial X^{\kappa}}+\eta_{\mu\beta
}x^{\beta}\frac{\partial X^{4}}{\partial x^{\nu}}\frac{\partial}{\partial
X^{4}}-\eta_{\nu\beta}x^{\beta}\frac{\partial X^{\kappa}}{\partial x^{\mu}%
}\frac{\partial}{\partial X^{\kappa}}-\eta_{\nu\beta}x^{\beta}\frac{\partial
X^{4}}{\partial x^{\mu}}\frac{\partial}{\partial X^{4}}\\
=x_{\mu}\frac{\partial X^{\kappa}}{\partial x^{\nu}}\frac{\partial}{\partial
X^{\kappa}}-x_{\nu}\frac{\partial X^{\kappa}}{\partial x^{\mu}}\frac{\partial
}{\partial X^{\kappa}}+x_{\mu}\frac{\partial X^{4}}{\partial x^{\nu}}%
\frac{\partial}{\partial X^{4}}-x_{\nu}\frac{\partial X^{4}}{\partial x^{\mu}%
}\frac{\partial}{\partial X^{4}}\\
=x_{\mu}\left(  -\frac{\Omega^{2}}{2\ell^{2}}x_{\nu}x^{\kappa}+\Omega
\delta_{\nu}^{\kappa}\right)  \frac{\partial}{\partial X^{\kappa}}-x_{\nu
}\left(  -\frac{\Omega^{2}}{2\ell^{2}}x_{\mu}x^{\kappa}+\Omega\delta_{\mu
}^{\kappa}\right)  \frac{\partial}{\partial X^{\kappa}}\\
+\left(  \frac{\Omega^{2}}{\ell}x_{\nu}x_{\mu}-\frac{\Omega^{2}}{\ell}x_{\mu
}x_{\nu}\right)  \frac{\partial}{\partial X^{4}}\\
=X_{\mu}\frac{\partial}{\partial X^{\nu}}-X_{\nu}\frac{\partial}{\partial
X^{\mu}}-\frac{\Omega^{2}}{2\ell^{2}}x_{\nu}x_{\mu}x^{\kappa}\frac{\partial
}{\partial X^{\kappa}}+\frac{\Omega^{2}}{2\ell^{2}}x_{\nu}x_{\mu}x^{\kappa
}\frac{\partial}{\partial X^{\kappa}}\\
=X_{\mu}\frac{\partial}{\partial X^{\nu}}-X_{\nu}\frac{\partial}{\partial
X^{\mu}}=\mathring{\eta}_{\mu\beta}X^{\beta}\frac{\partial}{\partial X^{\nu}%
}-\mathring{\eta}_{\nu\beta}X^{\beta}\frac{\partial}{\partial X^{\mu}%
}.\blacksquare
\end{gather*}

\textbf{Proof of (b):}%

\begin{align*}
\mathbf{J}_{\mu4}  &  =\ell\frac{\partial}{\partial x^{\mu}}-\frac{1}{4\ell
}\left(  2\eta_{\mu\nu}x^{\nu}x^{\lambda}-\sigma^{2}\delta_{\mu}^{\lambda
}\right)  \frac{\partial}{\partial x^{\lambda}}\\
&  =\ell\frac{\partial}{\partial x^{\mu}}+\frac{1}{4\ell}\sigma^{2}%
\frac{\partial}{\partial x^{\mu}}-\frac{1}{4\ell}2\eta_{\mu\nu}x^{\nu
}x^{\lambda}\frac{\partial}{\partial x^{\lambda}}=-\frac{1}{\Omega}X^{4}%
\frac{\partial}{\partial x^{\mu}}-\frac{1}{2\ell}\eta_{\mu\nu}x^{\nu
}x^{\lambda}\frac{\partial}{\partial x^{\lambda}}\\
&  =-\frac{1}{\Omega}X^{4}\left(  \frac{\Omega^{2}}{2\ell^{2}}x_{\mu}%
x^{\kappa}+\Omega\delta_{\mu}^{\kappa}\right)  \frac{\partial}{\partial
X^{\kappa}}-\frac{1}{\Omega}X^{4}\frac{\partial X^{4}}{\partial x^{\mu}}%
\frac{\partial}{\partial X^{4}}\\
&  -\frac{1}{2\ell}x_{\mu}x^{\lambda}\left(  \frac{\Omega^{2}}{2\ell^{2}%
}x_{\lambda}x^{\kappa}+\Omega\delta_{\lambda}^{\kappa}\right)  \frac{\partial
}{\partial X^{\kappa}}-\frac{1}{2\ell}\eta_{\mu\nu}x^{\nu}x^{\lambda}%
\frac{\partial X^{4}}{\partial x^{\lambda}}\frac{\partial}{\partial X^{4}}\\
&  =-X^{4}\frac{\partial}{\partial X^{\mu}}-\frac{\Omega}{2\ell^{2}}%
X^{4}x_{\mu}x^{\kappa}\frac{\partial}{\partial X^{\kappa}}-\frac{1}{2\ell
}x_{\mu}x^{\lambda}\Omega\frac{\partial}{\partial X^{\lambda}}-\frac{1}%
{4\ell^{3}}\Omega^{2}\sigma^{2}x_{\mu}x^{\kappa}\frac{\partial}{\partial
X^{\lambda}}\\
&  +X^{4}\Omega x_{\mu}\frac{\partial}{\partial X^{4}}+\frac{1}{2\ell^{2}%
}x_{\mu}\Omega^{2}x^{\lambda}\frac{\partial X^{4}}{\partial x^{\lambda}}%
\frac{\partial}{\partial X^{4}}\\
&  =-X^{4}\frac{\partial}{\partial X^{\mu}}+\left\{  X^{4}\Omega+\frac
{1}{2\ell^{2}}\Omega^{2}\sigma^{2}\right\}  x_{\mu}\frac{\partial}{\partial
X^{4}}-\left\{  \frac{1}{\ell}X^{4}+1+\frac{1}{2\ell^{2}}\Omega\sigma
^{2}\right\}  \frac{1}{2\ell}\Omega x_{\mu}x^{\lambda}\frac{\partial}{\partial
X^{\lambda}}\\
&  =-X^{4}\frac{\partial}{\partial X^{\mu}}+X_{\mu}\frac{\partial}{\partial
X^{4}}-\left\{  -\left(  1+\frac{\sigma^{2}}{4\ell^{2}}\right)  +\frac
{1}{\Omega}+\frac{1}{2\ell^{2}}\sigma^{2}\right\}  \frac{1}{2\ell}\Omega
^{2}x_{\mu}x^{\lambda}\frac{\partial}{\partial X^{\lambda}}\\
&  =-X^{4}\frac{\partial}{\partial X^{\mu}}+X_{\mu}\frac{\partial}{\partial
X^{4}}-\left\{  -1+\frac{1}{4\ell^{2}}\sigma^{2}+1-\frac{1}{4\ell^{2}}%
\sigma^{2}\right\}  \frac{1}{2\ell}\Omega^{2}x_{\mu}x^{\lambda}\frac{\partial
}{\partial X^{\lambda}}\\
&  =-X_{4}\frac{\partial}{\partial X^{\mu}}+X_{\mu}\frac{\partial}{\partial
X^{4}}.\blacksquare
\end{align*}

where we used that:%

\begin{align*}
&  \left\{  X^{4}\Omega\frac{1}{\ell}+\frac{1}{2\ell^{2}}\Omega^{2}\sigma
^{2}\right\}  x_{\mu}\frac{\partial}{\partial X^{4}}\\
&  =\left\{  -\left(  1+\frac{1}{4\ell^{2}}\sigma^{2}\right)  +\frac{2}%
{4\ell^{2}}\sigma^{2}\right\}  \Omega^{2}x_{\mu}\frac{\partial}{\partial
X^{4}}\\
&  \left\{  -\left(  1+\frac{1}{4\ell^{2}}\sigma^{2}\right)  +\frac{2}%
{4\ell^{2}}\sigma^{2}\right\}  \Omega^{2}x_{\mu}\frac{\partial}{\partial
X^{4}}\\
&  =-\frac{1}{\Omega}\Omega^{2}x_{\mu}\frac{\partial}{\partial X^{4}}=-X_{\mu
}\frac{\partial}{\partial X^{4}}.
\end{align*}


\begin{thebibliography}{99}                                                                                               %


\bibitem {ap2013}{\footnotesize Aldrovandi, R. and Perira, J.G.,
\emph{Teleparallel Gravity: An Introduction}, FundamentalTheories of Physics
\textbf{173}, Springer, Dordrechet (2013).}

\bibitem {arc}{\footnotesize Arcidiacono, G., \emph{Relativit\'{a} e
Cosmologia,} vol. 2 (IV editione), Libreria Eredi Virgilio Veschi, Roma, 1987}

\bibitem {barut}{\footnotesize Barut, A. O., \emph{Electrodynamics and
Classical Theory of Fields and Particles}, Dover Publ. Inc., New York, 1980.}

\bibitem {berg}{\footnotesize Bergmann, P. G, Conservation Laws in General
Relativity as the Generators of Coordinate Transformations, \emph{Phys.Rev.
}\textbf{112}, 287-289 (1958).}

\bibitem {bt1987}{\footnotesize Benn I. M., and Tucker, R. W.,\emph{ An
Introduction to Spinors and Geometry with Applications in Physics}, Adam
Hilger, Bristol and New York, 1987}

\bibitem {bm}{\footnotesize Blaine Lawson, H., Jr., Michelsohn, M.-L. ,
\emph{Spin Geometry}, Princeton University Press, New Jersey.(1989).}

\bibitem {choquet}{\footnotesize Choquet-Bruhat, Y., DeWitt- Morette, C., and
Deillard-Bleick, M., \emph{Analysis, Manifolds and Physics} (revised edition),
North-Holland, Amsterdam , 1982.}

\bibitem {cru}{\footnotesize Crummeyrole, A., \emph{Orthogonal and Symplectic
Clifford Algebras}, Kluwer Acad. Publ.., Dordrecht, 1990.}

\bibitem {agp}{\footnotesize de Andrade, V.C., Guillen, L.C.T., Pereira, J.G.,
Gravitational Energy-Momentum Tensor in Teleparallel Gravity, \emph{Phys. Rev.
Lett}. \textbf{84}, 4533-4536 (2000).}

\bibitem {agp1}{\footnotesize de Andrade, V.C., Guillen, L.C.T., Pereira,
J.G., Teleparallel Gravity: An Overview,in Gurzadyar, V. G, Jansen, R. T. and
Ruffini, R. (eds.) Proceedings of IX Marcel Grosamman Metting on General
Relativity, World Sci. Publ., Singapore, 2002 , [arXiv:gr-qc/0011087]}

\bibitem {eck}{\footnotesize Eck, D. J, Gauge Natural Bundles and Gauge
Theories,\emph{ Mem. Amer. Math. Soc}. \textbf{33}, number 247, (1981).}

\bibitem {F1}{\footnotesize Fatibene, L. and Francaviglia, M, Natural and
Gauge Natural Formalism for Classical Field Theories, Kluwer Academic Publ.,
Dordrecht, 2003.}

\bibitem {fr2010}{\footnotesize Fern\'{a}ndez, V. V. and Rodrigues, W. A.
Rodrigues Jr., Gravitation as a Plastic Distortion of the Lorentz Vacuum,
\emph{Fundamental Theories of Physics} \textbf{168}, Springer, Berlin, 2010.
errata at http://www.ime.unicamp.br/\symbol{126}walrod/plastic04162013}

\bibitem {gursey}{\footnotesize G\"{u}rsey, F., Introduction to Group Theory,
in De Witt, C. and De Witt, B. (eds.), \emph{Relativity, Groups and Topology},
Gordon and Breach, New York, 1963.}

\bibitem {he}{\footnotesize Hawking, S. W., and Ellis, G. F. R., \emph{The
Large Scale Structure of Space-Time}, Cambridge Univ. Press, Cambridge, 1973.}

\bibitem {KMS}{\footnotesize Kol\'{a}\v{r}, I., Michor, P. W.., Slov\'{a}k,
J., \emph{Natural Operations in Differential Geometry}, Springer-Verlag,
Berlin, 1993}

\bibitem {kn}{\footnotesize Kobayashi, S., and Nomizu, K., \emph{Foundations
of Differential Geometry, vol I}, J. Wiley-Interscience, New York, 1981.}

\bibitem {komar}{\footnotesize Komar, A., Covariant Conservation Laws in
General Relativity, \emph{Phys. Rev.}\textbf{ 113}, 934-936 (1959).}

\bibitem {ll}{\footnotesize Landau, L., and Lifshitz, E. M., \emph{The
Classical Theory of Fields} ( fourth revised English edition), Pergamon Press,
New York, 1975.}

\bibitem {lrw2015}{\footnotesize Le\~{a}o, R. F., Rodrigues, W. A. Jr. and
Wainer, S. A., Concept of Lie Derivative of Spinor Fields. A Geometric
Motivated Approach, \emph{Adv. in Appl. Clifford Algebras }(2015), DOI
0.1007/s0006-015-0560-y [ arXiv:1411.7845 [math-ph]].}

\bibitem {lawmi}{\footnotesize Lawson, H. B. Jr. and Michelson, M-L.,
\emph{Spin Geometry}, Princeton University Press, Princeton, 1989.}

\bibitem {lounesto}{\footnotesize Lounesto, P., \emph{Clifford Algebras and
Spinors}, Cambridge University Press, Cambridge, 1997.}

\bibitem {lovrund}{\footnotesize Lovelok, D., and Rund, H., Tensors,
\emph{Differential Forms, and Variational Principles}, J. Wiley \& Sons, New
York, 1975.}

\bibitem {maluf}{\footnotesize Maluf, J.W., The Teleparallel Equivalent of
General RelativIty, \emph{Annalen der Physik} \textbf{525}, 339-357 (2013).}

\bibitem {mr2004}{\footnotesize Mosna, R. A. and Rodrigues, W. A. Jr., The
Bundles of Algebraic and Dirac-Hestenes Spinor Fields,.\emph{J. Math. Phys}.
\textbf{45}, 2945-2966 (2004). [ arXiv:math-ph/0212033]}

\bibitem {nra}{\footnotesize Notte-Cuello, E., Rodrigues, W. A. Jr., and Q. A.
G. de Souza,The Square of the Dirac and spin-Dirac Operators on a
Riemann-Cartan Space(time), \emph{Rep. Math. Phys}. \textbf{60}, 135-157
(2007) [arXiv:math-ph/0703052 ]}

\bibitem {oneillo}{\footnotesize O'Neill, B.,\textit{ Elementary Differential
Geometry}, Academic Press, New York, 1966.}

\bibitem {oneill}{\footnotesize O'Neill, B., \emph{Semi-Riemannian Geometry},
Academic Press, New York, 1893.}

\bibitem {papapetrou}{\footnotesize Papapetrou, A., Spinning Test-Particles in
General Relativity I, \emph{Proc. Royal Soc. A }\textbf{209}, 248-25 (1951).}

\bibitem {pauli}{\footnotesize Pauli, W., \emph{Theory of Relativity}, Dover
Publ. Inc., New York (1981), originally published in German:
Relativit\"{a}tstheorie, \textit{Encyclop\"{a}die der Mathematischen
Wissenschaften}.\textbf{19}, B.G., Teubner, Leipzig,1921.}

\bibitem {pr2}{\footnotesize Penrose, R. and Rindler, W., \emph{Spinors and
Spacetime, }vol.2\emph{, Spinor and Twistor Methods in Spacetime Geometry},
Cambridge Univ. Press, Cambridge, 1986.}

\bibitem {ps2012}{\footnotesize Pereira, J. G. and Sampson, A., De Sitter
Geodesics: Reappraising the Notion of Motion, \emph{Gen. Rel. Grav.}
\textbf{44},1299-1308 (2012) [arXiv:110.0965 [gr-qc]]}

\bibitem {r2004}{\footnotesize Rodrigues, W. A. Jr., Algebraic and
Dirac-Hestenes Spinors and Spinor Fields, \emph{J. Math. Phys}. \textbf{45} ,
2908-2994 (2004) [ arXiv:math-ph/0212030]. }

\bibitem {rc2007}{\footnotesize Rodrigues, W. A. Jr. and Capelas de Oliveira,
E., The Many Faces of Maxwell, Dirac and Einstein Equation,. A Clifford Bundle
Approach,\emph{ Lecture Notes in Physics}\textbf{ 722}, Springer, Heidelberg,
2007. A preliminary enlarged second edition may be found
at\ http://www.ime.unicamp.br/\symbol{126}walrod/svmde04092013.pdf}

\bibitem {R2010}{\footnotesize Rodrigues, W. A. Jr., Killing Vector Fields,
Maxwell Equations and Lorentzian Spacetimes, \emph{Adv. Appl. Clifford
Algebras.} \textbf{20}, 871-884 (2010).}

\bibitem {rrr2012}{\footnotesize Rodrigues, F. G., Rodrigues, W. A. Jr., and
Rocha, R., The Maxwell and Navier-Stokes that Follow from Einstein Equation in
a Spacetime Containing a Killing Vector Field, in Rodrigues, W. A. Jr.,
Kerner, R., Pires, G.O, and Pinheiro C. (eds.), Proc. of the Sixth Int. School
on Field Theory and Gravitation-2012, \emph{AIP Conf. Proc.} \textbf{1483},
277-295 (2012). [arXiv:1109.5274 [math-ph]]}

\bibitem {rod2012}{\footnotesize Rodrigues, W. A. Jr., Nature of the
Gravitational Field and its Legitimate Energy-Momentum Tensor, \emph{Rep.
Math. Phys.}\textbf{ 69}, 265-279 (2012) [arXiv:1109.5272 [math-ph]]}

\bibitem {rs2013}{\footnotesize Rodrigues, W. A. Jr. and Wainer, S., A
Clifford Bundle Approach to the Geometry of Branes, \emph{Adv. Applied.
Clifford Algebras} 24, 817-847 (2015). DOI: 10.1007/s00006-014-0452-6
[arXiv:1309.4007 [math-ph]]}

\bibitem {RWC2015}{\footnotesize Rodrigues, W. A. Jr., Wainer, S. A ,
Rivera-Tapia, M., Notte-Cuello,E. A., and Kondrashuk, I. , A Clifford Bundle
Approach to the Wave Equation of a Spin 1/2 Fermion in the de Sitter Manifold,
\textit{Adv. Applied Clifford Algebra }(2015). DOI: 10.1007/s00006-015-0588-z
\ \ [arXiv:1502.05685v3 [math-ph]]}

\bibitem {roman}{\footnotesize Roman, P., \emph{Introduction to Quantum Field
Theory}, J. Willey and Sons, Inc., New York, 1969.}

\bibitem {rosenfeld}{\footnotesize Rosenfeld, Sur\ le Tenseur
d'Impulsion-Energy, \textit{M\'{e}moires Acad. Roy. d Belg,} \textbf{18}, 1-30
(1940).}

\bibitem {schmidt}{\footnotesize Schmidt, H-J., On the de Sitter
Space-Time-The Geometric Foundation of Inflationary Cosmology,\emph{ Fortschr.
Phys.} \textbf{41}, 179-199 (1933).}

\bibitem {sw}{\footnotesize Sachs, \ R. K., and Wu, H., \emph{General
Relativity for Mathematicians}, Springer-Verlag, Berlin, 1977.}

\bibitem {sampsom}{\footnotesize Sampson, A. , \emph{Transitividade e
Movimento em Relatividade de de Sitter} (tese de doutorado, IFT-UNESP), 2013.}

\bibitem {tiee}{\footnotesize Tiee, C., \emph{Contravariance, Covariance,
Densities, and all That: An Informal Discussion on Tensor Calculus}, 2006.
[http://math.ucsd.edu/\symbol{126}ctiee/tensors.pdf]}

\bibitem {trautman}{\footnotesize Trautman, A., Foundations and Current
Problems in General Relativity, in Deser, S. and Ford, K. W. (eds.),
\emph{Lectures on General Relativity, Brandeis Summer Institute in Theoretical
Physics }1964, \ Prentice Hall Inc., Englewood Cliffs, 1965.}

\bibitem {weinberg}{\footnotesize Weinberg, S., \emph{Gravitation and
Cosmology; Principles and Applications of the General Theory of Relativity},
J. Wiley and Sons,New York, 1972}

\bibitem {wald}{\footnotesize Wald, R. M., \emph{General Relativity}, Univ.
Chicago Press, Chicago,1984.}
\end{thebibliography}
\end{document}